\documentclass[aps,prx,twocolumn,nofootinbib,superscriptaddress,longbibliography,citeautoscript]{revtex4-1}
\pdfoutput=1

\usepackage[utf8]{inputenc}
\DeclareUnicodeCharacter{2009}{\,}  

\usepackage{amsmath}
\usepackage{amssymb}	
\usepackage{graphicx} 
\usepackage{comment}
\usepackage{color}
\usepackage[english]{babel}

\usepackage{hyperref} 
\hypersetup{
	pdftoolbar=true,        		 
	pdfmenubar=true,     		 
	pdffitwindow=false,     	 
	pdfstartview={FitH}, 
	pdftitle={Long-lived interacting phases of matter protected by multiple time-translation symmetries in quasiperiodically-driven systems},  
	pdfauthor={Else, Ho, Dumitrescu},     		 
	pdfsubject={},   			 
	pdfcreator={},   	        	   	
	pdfproducer={}, 		
	pdfnewwindow=true,      	
	colorlinks=true,       		
	linkcolor=black,          	
	citecolor=blue,        		
	filecolor=magenta,    		
	urlcolor=blue          		
}

\usepackage{amsthm}
\newtheorem{lemma}{Lemma}
\newtheorem{proposition}{Proposition}
\newtheorem{theorem}{Theorem}
\theoremstyle{definition}
\newtheorem{defn}{Definition}

\newcommand{\eqnref}[1]{Eq.~\eqref{#1}}

\renewcommand{\vec}{\boldsymbol}

\newcommand{\vt}{{\vec{\theta}}}
\newcommand{\vw}{{\vec{\omega}}}
\newcommand{\va}{{\vec{n}}}

\newcommand{\ket}[1]{|#1\rangle}
\newcommand{\bra}[1]{\langle #1|}

\DeclareMathOperator{\Tr}{Tr}

\usepackage{xcolor}

\begin{document}
\title{Long-lived interacting phases of matter protected by multiple time-translation symmetries in quasiperiodically-driven systems}
\author{Dominic V.~Else}
\email{d\_else@mit.edu}
\affiliation{Department  of  Physics,  Massachusetts  Institute  of  Technology,  Cambridge, MA 02139, USA}
\author{Wen Wei Ho}
\email{wenweiho@fas.harvard.edu}
\affiliation{Department of Physics, Harvard University, Cambridge, MA 02138, USA}
\author{Philipp T.~Dumitrescu}
\email{pdumitrescu@flatironinstitute.org}
\affiliation{Center for Computational Quantum Physics, Flatiron Institute, 162 5th Avenue, New York, NY 10010,  USA}
\date{\today} 
\begin{abstract}
We show how a large family of interacting nonequilibrium phases of matter can arise from the presence of multiple time-translation symmetries, which occur by quasiperiodically driving an isolated quantum many-body system with two or more incommensurate frequencies. These phases are fundamentally different from those realizable in time-independent or periodically-driven (Floquet) settings. Focusing on high-frequency drives with smooth time-dependence, we rigorously establish general conditions for which these phases are stable in a parametrically long-lived `preheating' regime. We develop a formalism to analyze the effect of the multiple time-translation symmetries on the dynamics of the system, which we use to classify and construct explicit examples of the emergent phases. In particular, we discuss time quasi-crystals which spontaneously break the time-translation symmetries, as well as time-translation symmetry protected topological phases.
\end{abstract}
\maketitle

\section{Introduction}

Many-body quantum systems give rise to a vast array of interesting phases of matter.
The last decade has seen a dramatic expansion and refinement in our understanding of the landscape of such phases~\cite{Hasan_1002,Qi_1008,Senthil_1405,Wen_1610}.
Recently, it was understood how time-translational symmetry (TTS) itself can give rise to and protect intrinsically out-of-equilibrium phases of matter in isolated quantum systems~\cite{Eckardt_1606,Moessner_1701,Else_1905,Harper_1905}. 
Arguably the simplest example is the discrete time crystal~\cite{Khemani_1508,vonKeyserlingk_1602_b,Else_1603,Else_1607,Yao_1608,Ho_1703,Else_1905}, characterized by the spontaneous breaking of the discrete TTS of a periodic (Floquet) drive.
This is manifested in physical observables exhibiting robust, long-lived oscillations at an integer-multiple of the base driving period. 
Experimental signatures of this behavior have been reported in various platforms~\cite{Zhang_1609,Choi_1610,Rovny_1802}. 
Going beyond discrete time crystals, a large number of other Floquet phases in which TTS plays an essential role, including topological phases beyond the equilibrium classification, have also been discussed.

The richness of Floquet phases naturally raises the questions: are there fundamentally different nonequilibrium interacting phases beyond Floquet, which are not dynamically engineered analogs of static phases? What is the role of TTS in characterizing these phases?

Apart from theoretical interest, these questions are placed upon us by the dramatic experimental advances in controlling and manipulating isolated quantum many-body systems, such as cold atoms~\cite{Bloch_0704}, trapped ions~\cite{Leibfried_2003,Duan_2010,Blatt_2012}, nitrogen-vacancy centers~\cite{Doherty_1302,Schirhagl_2014}, and superconducting qubits~\cite{Kelly_1411,Roushan_1709}. These systems provide a natural platform to realize dynamical protocols, allowing us to systematically study physics in out-of-equilibrium settings, including thermalization and equilibration.  Classifying nonequilibrium phases tells us exactly what long-time, dynamical collective behaviors are possible and the universal features defining them.

Driven interacting systems are, however, generically expected to heat up to a featureless infinite temperature state due to a lack of energy conservation~\cite{Deutsch_1991,Srednicki_9403,DAlessio_1402,Lazarides_1403,Ponte_1403}. Thus, to meaningfully define phases of matter in such settings, systems must be protected against heating, at least for some long timescale. For Floquet systems, this challenge can be overcome by applying high-frequency drives leading to exponentially long-lived prethermal regimes~\cite{Abanin_1507,Kuwahara_1508,Mori_1509,Abanin_1509,Abanin_1510,Weidinger_1609} or by applying strong disorder leading to many-body localization (MBL)~\cite{Ponte_1403,Lazarides_1410, Abanin_1412}. Generalizing these ideas to more generic driving scenarios remains an  important open question. %

In this paper, we consider interacting quantum many-body systems subject to a quasiperiodic drive that consists of several incommensurate frequencies and is smooth in time.
We rigorously show that under such driving scenarios, the system is protected at high driving frequencies from heating for a parametrically long time, giving rise to a so-called `preheating' regime. The heating time scales as a stretched exponential of the ratio of the drive frequency to local coupling strengths. We demonstrate through a recursive construction that there is an effective static Hamiltonian governing time-evolution in the preheating regime, which generalizes the Floquet analysis of~\cite{Abanin_1510,Else_1607} to a large class of new dynamical systems.

The presence of a preheating regime  in quasiperiodically-driven systems opens up an avenue to realize novel long-lived, nonequilibrium phases of matter. 
We provide a set of general driving conditions to realize these phases and discuss how they are distinguished by a notion of \emph{multiple} time-translation symmetries (TTSes) of the drives; thus, they are fundamentally different from those in static or Floquet settings. In particular, we classify two exemplars of such phases: the discrete time quasi-crystal (DTQC), which spontaneously breaks the multiple TTSes, as well as quasiperiodic symmetry-protected and symmetry-enriched topological phases, which are protected by them.
Our results showcase the richness of the landscape of quasiperiodically-driven phases, and excitingly opens up new directions in the rapidly developing field of nonequilibrium quantum matter.

Before we proceed, let us note that the study of the dynamics of quantum systems under quasiperiodic driving has a venerable history, encompassing diverse applications from experiments in chemistry and physics, to the basic structure of first-order differential equations~\cite{Ho_1983,Luck_1988,Jauslin_1991,Jauslin_1992,Blekher_1992,Feudel_1995,Casati_1989,Jorba_1992,Bambusi_2001,Gentile_2002,Chu_2004,Gommers_0610,Chabe_0709,Lemarie_1005,Cubero_1309,Verdeny_2016,Nandy_1701,Cubero_1806,Nandy_1810,Ray_1907}. 
Interesting dynamical behavior related to topology in few-body or non-interacting scenarios have also been reported~\cite{Martin_1612,Peng_1805,Crowley_1808,Nathan_1811,Crowley_1908}.
Many-body quasiperiodically-driven quantum systems with interactions have received increasing attention comparatively recently~\cite{Dumitrescu_1708,Giergiel_1807,Zhao_1906}. Our approach is distinguished from previous studies in that we explicitly establish the stability of the nonequilibrium phases we discuss, by rigorously providing a bound on their lifetimes.
We additionally demonstrate the robustness of their universal properties against small changes in the driving protocol, which justifies their characterization as `phases of matter'.%

The outline of the rest of this paper is as follows. In Section \ref{sec:Overview}, we summarize our main results on establishing a long preheating regime in quasiperiodically-driven systems, in which one can discuss phases of matter. 
In Section \ref{sec:emergentsymmetries}, we introduce the notion of a ``frame-twisted high-frequency limit'', which will allow us to find new phases of matter with no static or periodically driven analogs. 
We will show how TTSes act in this regime and make this precise by defining ``twisted time-translation symmetries''. 
In Section \ref{sec:tqc}, we discuss spontaneous symmetry breaking for the multiple TTSes, which lead to discrete time quasi-crystal phases.
In Section \ref{sec:TopologicalPhase}, we define and classify topological phases protected by the multiple TTSes.
In Section \ref{sec:heating_estimate}, we return to the stability of the prethermal regime and show how the scaling of the heating time with frequency can be intuitively understood in terms of simple linear response arguments.
In Section \ref{sec:proof}, we state and sketch the proof of our rigorous results on the heating bounds and description of the dynamics; technical details are relegated to Appendix \ref{appendix:proof}. Finally, in Section \ref{sec:extensions} we discuss various extensions and future directions, and we conclude in Section \ref{sec:conclusion}.

\emph{Remark on notation:}~the term `time quasi-crystal' (TQC) has been used for systems that show a quasiperiodic response, arising from either a quasiperiodic~\cite{Dumitrescu_1708} or a periodic drive~\cite{Autti_1712, Pizzi_1907}. In this manuscript, we will restrict use of the term time quasi-crystal to the first sense, where a quasiperiodic drive gives rise to a  response with a different quasiperiodic pattern (see Sec.~\ref{sec:tqc}).

\tableofcontents

\section{Overview: Main ideas and key results}\label{sec:Overview}

\subsection{Multiple time-translational symmetries in quasiperiodically-driven systems}
\label{sec:intro_qp_tts}

In this paper, we show the existence of long-lived nonequilibrium phases of matter protected by multiple time-translational symmetries (TTSes) in quasiperiodically-driven systems, which are defined as follows. 
Consider  an at-least piecewise continuous Hamiltonian $H(\vec\theta)$ parametrized by the $m$-dimensional ``standard'' torus $\mathbb{T}^m \ni \vec\theta =(\theta_1,\cdots,\theta_m)$, which is $2\pi$-periodic in each angle $\theta_i$. 
Additionally, let us pick a vector of rationally independent frequencies $\vec\omega = (\omega_1, \cdots, \omega_m)$, so that  $\vec{n}\cdot\vec{\omega} \neq 0$ for any non-zero integer vector $\vec{n} \in \mathbb{Z}^m$.  
When $m=2$, it suffices to choose the ratio $\omega_2/\omega_1$ to be an irrational number, such as the golden ratio $\omega_2/\omega_1 = \varphi = ({1+\sqrt{5}})/{2}$.
The dynamics of the system under the time-dependent Hamiltonian
\begin{align}\label{eq:hambasic}
H(t) = H(\vec\omega t + \vec\theta_0)
\end{align}
is time-quasiperiodic for $m \geq 2$, and constitutes a quasiperiodic drive. Here, $t$ is the physical `single' time and $\vec\theta_0 \in \mathbb{T}^m$ are some fixed initial angles (see Fig.~\ref{fig:TorusCutProj}). The class of drives of \eqnref{eq:hambasic} encompasses Floquet driving as the case $m=1$, which will enable us to directly compare our analysis to previous work.  Note that throughout the paper, we will use the same symbol -- e.g.~$H$ above -- to refer to the same physical quantity viewed either as a function on the torus, $H(\vec{\theta})$, or as a function of single time, $H(t)$.  
 As in \eqnref{eq:hambasic}, the single time function is obtained by evaluating the torus function along the particular trajectory $\vec\theta(t) = \vec\omega t + \vec\theta_0$.
While these  are technically different mathematical functions, they can be distinguished by their scalar $t$ or vector $\vec{\theta}$ arguments.

At first glance, it is puzzling how one could obtain phases protected by TTSes in quasiperiodically-driven systems. After all, the incommensurate nature of the driving frequencies implies that the Hamiltonian \eqnref{eq:hambasic} has not even a single time-translational symmetry (there is no period $T \neq 0$ for which $H(t) = H(t+T)$ for all $t$), let alone multiple TTSes. 
However, $H(t)$ derives from an underlying $H(\vec\theta)$ on $\mathbb{T}^m$ through \eqnref{eq:hambasic}, which is symmetric under translations $\vec\theta \to \vec\theta + \vec\tau$. Here the translation vectors belong to a lattice $\vec\tau \in \mathcal{L} = 2\pi \mathbb{Z}^m$ generated by $m$ independent symmetries. It is thus conceivable that these symmetries  have meaningful and nontrivial implications for the single-time system.
Since for $m = 1$ this is nothing but the time-translation symmetry of a Floquet system, we will still refer to these symmetries as ``time-translation symmetries'' for quasiperiodically-driven systems ($m \geq 2$).

\begin{figure}[t]
\includegraphics[width=\columnwidth]{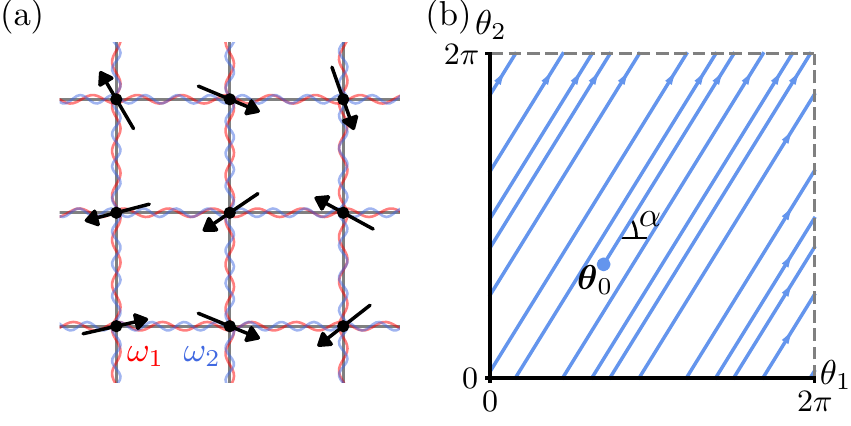}
\caption{(a) Cartoon of an ensemble of spins that is quasiperiodically-driven. (b) A time-quasiperiodic Hamiltonian $H(t)$ arises from evaluating a Hamiltonian  $H(\vec\theta)$ defined on the $m$-dimensional torus $\mathbb{T}^m \ni \vec\theta$  on the trajectory $\vec\theta(t) = \vec\omega t + \vec\theta_0$ (blue arrows). As the frequency vector $\vec\omega$ is a set of rationally independent frequencies, the trajectory never returns to itself, but covers the torus uniformly.  Shown here is the case $m = 2$, with $\tan\alpha = \omega_2/\omega_1$ irrational.
}
\label{fig:TorusCutProj} 
\end{figure}

In what follows, we will make the above statements precise by interpreting the consequences of the multiple TTSes 
in a certain class of physical systems. Specifically, we consider  quantum many-body systems   defined on lattices in arbitrary spatial dimensions with bounded local Hilbert spaces (i.e.~spins or fermions), and which respect a sense of locality -- the interaction strength decays sufficiently fast with distance. In particular, we allow for a Hamiltonian having interactions with amplitude that are at least exponentially-decaying with distance (termed `quasilocal'), see Sec.~\ref{sec:proof}.
We will show that:
\begin{enumerate}
\item For such strongly interacting many-body systems under some non-fine-tuned quasiperiodic driving conditions, the multiple time-translation symmetries of $H(\vec\theta)$ give rise to an actual symmetry of the effective time-independent Hamiltonian that describes the dynamics in a long-lived preheating regime. This enables
 the existence of novel nonequilibrium phases of matter protected by these symmetries.
\item The classification of quasiperiodically-driven many-body phases of matter
-- both spontaneous symmetry breaking and topological phases -- is the same as the classification of equilibrium phases with a symmetry group extended by $\mathbb{Z}^{\times m}$. This is a direct generalization of  the $m=1$ Floquet results of Ref.~\cite{Else_1602}.
\end{enumerate}

\subsection{Long lifetimes in quasiperiodically-driven systems}
\label{subsec:longlifetimes}

Owing to the lack of energy conservation, a generic ergodic interacting driven system is expected to heat to a featureless infinite-temperature state, where symmetries act trivially and a discussion of phases of matter is moot. 
Therefore, before we can even discuss new phases realizable with multiple TTSes, we must 
establish that there exist suitable quasiperiodic driving conditions where such deleterious heating is controlled, at least for some parametrically long time.

In Floquet systems ($m = 1$), the conditions needed to achieve this are relatively mild.
Energy absorption or emission  between the system and the drive can only take place in integer multiples of the driving frequency, i.e.~$\Delta E = n \omega$ for integer $n$.
 By suppressing resonances between energy eigenstates connected by such discrete energy levels, heating can be slow.
Both Floquet-MBL and Floquet prethermalization involve suppressing heating in a such a way, though through different physical mechanisms. The former uses strong disorder to directly curtail the probability of local resonances~\cite{Ponte_1403, Lazarides_1410,Abanin_1412}.  The latter entails driving at such high frequencies compared to local energy scales $J \ll \omega$ that the system can only absorb the large drive quanta $n\omega$ by performing a multiple-spin rearrangement. This rate is heavily suppressed giving rise to a long heating timescale $t_* \sim e^{\text{const.}\omega/J}$~\cite{Abanin_1507}.
Other scenarios where heating is slow that result not from disorder or high-frequency driving have also been considered, see~\cite{Lindner_1603,Haldar_1803,Haldar_1909}.

In   quasiperiodically-driven systems ($m \geq 2$), by contrast, energy absorption or emission occurs in units of $\Delta E_{\vec{n}} = \vec{\omega} \cdot \vec{n}$ for any integer vector $\vec{n} \in \mathbb{Z}^m$. 
As the frequencies $\vec \omega = (\omega_1,\cdots,\omega_m)$ are rationally independent, the set of all possible such quanta $\Delta E_{\vec{n}}$ is dense on the real line. Superficially, it seems impossible to avoid immediate heating.

A more careful consideration shows however that this is not necessarily an insurmountable problem.
Consider, as an example, a static Hamiltonian $H_0$ weakly driven by a quasiperiodic perturbation $V(t) = V(\vec{\omega} t + \vec{\theta}_0)$. Expanding $V(\vec{\theta})$ in a Fourier series
\begin{equation}
V(\vec{\theta}) = \sum_{\vec{n} \in \mathbb{Z}^m }  V_{\vec{n}} e^{i \vec{n} \cdot \vec{\theta}},
\end{equation}
one sees that $V_{\vec{n}}$ induces transitions between energy levels of $H_0$ separated by energies $\Delta E_{\vec{n}}$ in linear response theory.
The set of all $\Delta E_{\vec{n}}$ with $|\vec{n}|$ below some cutoff becomes ever more closely spaced as the cutoff increases.
Nonetheless, as long as the amplitude $\| V_{\vec{n}} \|$ decays fast enough with $|\vec n|$ as compared to this spacing, resonances will not proliferate at large $|\vec n|$.
This observation suggests a restriction of $V(\vec\theta)$ to be  `sufficiently' smooth on the torus, which in turn translates to smooth drives in time. Resonances arising from small $|\vec{n}|$ processes also need to be suppressed -- but this should be achievable with the same mechanisms as in the Floquet case (strong disorder or high-frequency driving).

One of the key aspects of our present work is to make this plausible stability statement concrete.
Specifically, we will consider the case corresponding to a quasiperiodically-driven many-body Hamiltonian $H(\vec{\omega} t + \vec\theta_0)$ under two conditions. First, that it is smooth in time, in the sense that 
the amplitude of the Fourier coefficients $H_{\vec{n}}$ of the driving Hamiltonian decay as $\| H_{\vec{n}} \| \lesssim e^{ - \text{const.} |\vec n|}$ at large $\vert\vec{n}\vert$.
Second, that the driving frequencies $\omega_i$ are large compared to any local energy scales $J$ of the Hamiltonian.
Under these conditions, 
we will show that the heating time $t_*$ is rigorously bounded for any $\epsilon > 0$, and for all except a set of measure zero choices of frequency vectors $\vec\omega$ by 
\begin{align}
t_* \geq \frac{C'}{J} \exp\left[C \left(\frac{\omega}{J}\right)^{1/(m + \epsilon )}\right].
\label{eqn:heating_t}
\end{align}
Here  $\omega = |\vec\omega|$ is the norm of the driving frequency, $C, C'$ are  dimensionless numbers depending on the number-theoretic properties of the irrational ratios $\omega_i/\omega_j$, and $m$ is the number of incommensurate frequencies.
Thus, keeping the ratios $\omega_i/\omega_j$ fixed, the bound on heating time $t_*$ follows a simple stretched exponential in frequency dependence.
While  
this
 can be intuitively understood within linear response theory (Sec.~\ref{sec:heating_estimate}), we will prove the bound \eqnref{eqn:heating_t} using a recursive construction beyond linear response (Sec.~\ref{sec:proof}). Note that for non-smooth drives, where $V_{\vec{n}}$ decays slower than exponentially at large $\vert\vec{n}\vert$, the heating time $t_*$ scales with a different functional form. For a power law decay of $V_{\vec{n}}$, such as for step-drives, we expect that $t_*$ scales as a power law in $\omega/J$ (see Sec.~\ref{subsec:nonsmooth}).  While Eq.~\eqref{eqn:heating_t} is analogous to the high-frequency heating bound known for Floquet systems, the heating time in those systems is always a simple exponential in $\omega / J$ without restrictions on the drive smoothness. 

\subsection{Description of preheating dynamics}
\label{sec:preheating_dynamics}

The existence of a long timescale \eqnref{eqn:heating_t} implies a long-lived preheating regime and opens up the possibility to define phases of matter in this time interval. How can we concretely describe the dynamics, and eventually characterize phases, in the preheating regime?

Let us briefly recall the Floquet scenario ($m=1$), where generally the dynamics in a preheating regime is approximately governed by an effective, static, quasilocal Hamiltonian $D$.
More precisely, the time-evolution operator $U(t) = \mathcal{T} \exp[-i \int_{0}^{t} H(t') dt']$ can be written as 
\begin{equation}
\label{eqn:Uapprox}
U(t) \approx P(t) e^{- i D t} P^{\dagger}(0),
\end{equation}
where $P(t) = P(t+T)$ is a unitary change of frame which is periodic in time, and the approximate equality reflects an omission of small local terms in the Hamiltonian that do not affect the dynamics up to the heating time $t_*$ \cite{Kuwahara_1508,Mori_1509, Abanin_1509, Abanin_1510}. 
A quantum state's dynamical evolution can therefore be understood as comprised of two parts: time-evolution generated by the static Hamiltonian $D$, and an additional `micromotion' governed by $P(t)$.
If we only consider the state of the system at stroboscopic times, i.e.~integer multiples of the driving period, then the dynamics is generated by the Floquet operator $U_F := U(T)$. This operator can be written [if we choose $P(0) = P(T) = \mathbb{I}$] as
$U_F \approx e^{-i D T}$, in which case the study of the dynamical evolution of the system up to time $t_*$ entails studying the eigenstates of the   Hamiltonian $D$.

We strongly emphasize here the importance of $D$ being quasilocal.  In fact, the 
Floquet-Bloch theorem asserts that the decomposition \eqnref{eqn:Uapprox} exists with an exact equality if $D$ is replaced by the Floquet Hamiltonian $H_F$. However, $H_F$ will be highly non-local in a generic ergodic many-body system -- it must after all describe the eventual heating to an infinite-temperature state --  and therefore is not very insightful to use when studying the preheating regime, as compared to $D$.

One scenario where a long heating time emerges and the approximate decomposition \eqref{eqn:Uapprox} holds is in the limit of high-frequency driving, in which case $t_* \gtrsim e^{\text{const.}\omega/J}$.
The small ratio of the local energy scale to the driving frequency $J/\omega$ naturally enables schemes for an order-by-order expansion of  $D$ and $P$. For example, the commonly-used ``Floquet-Magnus'' expansion \cite{Bukov_1407, Eckardt_1606} with $P(0) = P(T) = \mathbb{I}$ gives at lowest orders
\begin{align} 
D^{(0)} &= \frac{1}{T} \int_0^T\!\! dt'\, H(t') , \nonumber  \\
D^{(1)} &= \frac{1}{2T}\int_0^T \!\!dt' \int_0^{t'} \!\!dt''\, [H(t'),H(t'')],
\label{eqn:MagnusD_original}
\end{align}
and is generally an asymptotic series.
Refs.~\cite{Kuwahara_1508,Mori_1509} showed that if truncated at some optimal order, \eqnref{eqn:Uapprox} is satisfied with an error $\sim 1/t_* = O(e^{-\text{const.}\omega/J})$. Ref.~\cite{Abanin_1509} also constructed an effective static Hamiltonian that provably approximates the dynamics until the same $t_*$, although it was not directly expressed in terms of the Floquet-Magnus expansion.

We now return to quasiperiodically-driven systems ($m \geq 2$). For the smooth high-frequency drives that we consider, we will prove that a similar decomposition of the unitary time evolution operator $U(t)$ as \eqnref{eqn:Uapprox} exists,
\begin{align}
U(t) \approx  P(\vec{\omega} t + \vec{\theta_0}) e^{-i D t} P^\dagger(\vec\theta_0).
\label{eqn:Uapprox2}
\end{align}
A quantum state's dynamics is again effectively comprised of time-evolution by some static Hamiltonian $D$ and a micromotion given by some unitary time-quasiperiodic change of  frame $P(t) := P(\vec{\omega} t + \vec{\theta_0})$ with underlying $P(\vec{\theta})$ smooth on the torus.
In Sec.~\ref{sec:proof}, we will demonstrate how to construct the effective Hamiltonian $D$ and unitary $P$  through an iterative renormalization procedure of the driving Hamiltonian $H(t)$
that can be understood as a generalization of  the methods of Ref.~\cite{Abanin_1509}, as well as bound the optimal order to which the procedure should be carried out. 
This gives rise to an optimal $D$ such that the description \eqnref{eqn:Uapprox2} is valid at least for times $t \lesssim  t_*$, with $t_*$ satisfying \eqnref{eqn:heating_t}.

While it seems natural to assume that the decomposition \eqnref{eqn:Uapprox} 
carries over from the Floquet case to   quasiperiodically-driven systems,
this is far from obvious. It is known rigorously that the decomposition \eqnref{eqn:Uapprox2} with an exact equality (i.e.~Floquet-Bloch theorem) is not guaranteed in general quasiperiodic systems, there being obstructions to defining a generalized Floquet Hamiltonian~\cite{Jauslin_1991}.
To understand why one does expect \eqref{eqn:Uapprox2} to hold in the quasiperiodically-driven case with conditions given in Sec.~\ref{subsec:longlifetimes}, observe that the high-frequency assumption suggests that an  expansion analogous to the Floquet-Magnus expansion \eqnref{eqn:MagnusD_original} can be  written down. 
Representing $H(\vec{\theta})$ as a Fourier series   $H(\vec{\theta}) = \sum_{\vec{n}} H_{\vec{n}} e^{i\vec{n} \cdot \vec{\theta}}$ and assuming the form \eqnref{eqn:Uapprox2} with $P(t)$ quasiperiodic, one can perform a formal expansion in powers of the inverse norm of driving frequencies $1/\omega$ of the effective Hamiltonian $D = D^{(0)} + D^{(1)} + \cdots$ as well as the unitary $P(\vec\theta) = \exp(\Omega^{(1)}(\vec\theta) + \Omega^{(2)}(\vec\theta) + \cdots)$, whose leading order terms read (see Appendix \ref{appendix:FloquetMagnus})
\begin{align}
D^{(0)} &= H_{\vec{0}}  = \int_{\mathbb{T}^m} \frac{d^m\vec\theta}{(2\pi)^m} H(\vec\theta), \nonumber  \\
D^{(1)} &= \frac{1}{2} \sum_{\vec{n}} \frac{1}{\vec{\omega} \cdot \vec{n}} [H_{\vec{n}}, H_{-\vec{n}}]
\label{eqn:MagnusD}.
\end{align}
Here $\Omega^{(1)}(\vec\theta) = - \sum_{\vec n \neq \vec 0} {H_{\vec n}} e^{i \vec{n} \cdot \vec{\theta}} / (\vec{n} \cdot \vec{\omega})$.
However, while relatively simple to construct, even the low order terms in \eqnref{eqn:MagnusD} already signal a difficulty not present in the Floquet case: The denominator $\vec{n}\cdot \vec{\omega}$ can be arbitrarily small, leading to possible divergences and bringing into question the validity of the expansion. This is precisely the manifestation of the denseness of resonances discussed in Sec.~\ref{subsec:longlifetimes}. As before, we observe that this issue can potentially be circumvented if the size of the Fourier coefficients $H_{\vec n}$ decay sufficiently rapidly with $|\vec{n}|$, such as in the case of smooth driving, so that the small denominators are suppressed.
Note also that while similar to \eqnref{eqn:MagnusD_original}, the expansion of \eqnref{eqn:MagnusD} does not reduce to it upon setting $m=1$, as we explain in Appendix \ref{appendix:FloquetMagnus}. This point   will be important in the discussion of emergent symmetries.

A central contribution of our paper is to show how imposing the smoothness conditions on the drive indeed leads to a meaningful high-frequency expansion which can be used to construct an effective static Hamiltonian.
The expansion we develop, which is different from that of  \eqnref{eqn:MagnusD}, is given in Sec.~\ref{sec:proof}.

\subsection{Equilibration and steady states in the preheating regime}
\label{subsec:long_time_preheating}
 
Let us now describe the kind of dynamics and `steady states' one can expect in the preheating regime, given an effective, static, quasilocal Hamiltonian $D$ via \eqnref{eqn:Uapprox2}.

If $D$ is generic and non-integrable, one expects dynamics from a simple initial state $|\psi_0\rangle$ to lead to thermalization with respect to $D$, when viewed in the time-dependent frame defined by $P(t)$. That is, the system locally approaches an equilibrium distribution $\rho_\beta \propto e^{-\beta D}$, where the inverse temperature $\beta$ depends on the energy of the initial state as measured by $D$. Precisely, $\beta$ is obtained from the relation $\frac{1}{Z}\Tr(D e^{-\beta D}) = \langle \psi_0| P(\vec\theta_0) D P^\dagger(\vec\theta_0) |\psi_0\rangle$ where $Z = \Tr(e^{-\beta D})$.
This happens provided local relaxation timescales $t_r$ are much less than the heating timescale $t_*$, i.e.~$t_r \sim J^{-1} \ll t_*$, which always occurs for high-frequency drives.
As the system is expected to eventually thermalize to an infinite-temperature state, $\rho \propto \mathbb{I}$ for $t > t_*$, due to the corrections in \eqnref{eqn:Uapprox2} that cannot eventually be neglected, one refers to equilibration to $\rho_\beta$ in the preheating regime as \emph{prethermalization}. In particular, we can then talk about `prethermal quasiperiodically-driven phases of matter' in this steady state. Of course, one must remember to include the effects of the time-quasiperiodic unitary $P(t)$ upon moving back to the laboratory frame. However, as $P(t)$ when constructed in the high-frequency limit is perturbatively close to the identity, this simply endows the steady state $\rho_\beta$ with additional micromotion of small amplitude $\sim J/\omega$; see Fig.~\ref{fig:prethermalCartoon}.

\begin{figure}[t]
\includegraphics[width=\columnwidth]{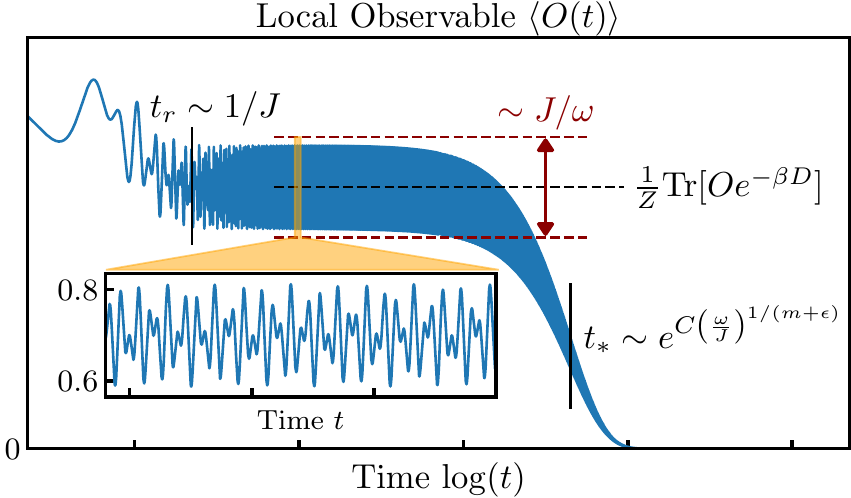}
\caption{\label{fig:prethermalCartoon} 
Prethermalization in quasiperiodically-driven systems at high frequencies. Shown is a cartoon of the dynamics of a generic traceless local observable $\langle O(t) \rangle$ under time-evolution by \eqnref{eqn:Uapprox2}. There are three regimes: First, a brief transient regime, where the local observable relaxes on a short timescale $t_r \sim 1/J$, where $J$ is the local energy scale of the system.  Second, a prethermalization regime, where the system has locally equilibrated to a thermal ensemble of an effective static Hamiltonian $D$, when viewed in the rotating frame defined by $P(t)$. The evolution $\langle O(t) \rangle$ in the laboratory frame shows a plateau around the prethermal value (black dashed line), with small time-quasiperiodic oscillations of amplitude $\sim J/\omega$ (red dashed lines); here $\omega = \vert \vec{\omega}\vert$ is the norm of driving frequencies. This lasts up to the long heating time $t_* \sim \exp[{C (\omega/J)^{1/(m+\epsilon)} }]$. Third, a final  featureless infinite-temperature state, reached after the system has fully heated. (Inset) A zoom-in on the orange shaded region in the prethermal plateau.}
\end{figure} 

We can also consider the case where $D$ is non-ergodic, such as when it is  highly disordered leading to many-body localization (MBL). In this case, there is a complete set of quasilocal integrals of motion $\tau_i^z$ (`$l$-bits') satisfying $[\tau_i^z,D]=0$. The system evolving under $D$, when viewed in the rotating frame $P(t)$, will not thermalize, but instead exhibit MBL phenomenology. This includes logarithmic entanglement growth, initial state memory and localization-protected quantum order~\cite{Nandkishore_1404,Alet_1711,Parameswaran_1801,Abanin_1804}.

In the laboratory frame, the dynamics is rather more interesting. Owing to the rotating change of frame, the $l$-bits are not constants of motion, but rather always evolve in time.
Despite this motion, there is a sense in which the system is still localized: Consider the dressed $l$-bit operator $\widetilde{\tau}^z_i := P(0) \tau^z_i P(0)^\dagger$, which is localized near site $i$. 
Under ``reverse evolution'' defined as
\begin{align}
\widetilde{\tau}_i^z(t)_R = U(t) \widetilde{\tau}_i^z U(t)^{\dagger},
\end{align}
we obtain using \eqnref{eqn:Uapprox2} and $[\tau_i^z,D]=0$  that
\begin{align}
\widetilde{\tau}^z_i(t)_R \approx P(t) e^{-i t D}  \tau_i^z   e^{i t D} P(t)^\dagger = P(t)\tau^z_i P(t)^\dagger.
\end{align}
This means that motion of the $\widetilde{\tau}_i^z(t)_R$ is time-quasiperiodic, which in turn implies that there is an infinite sequence in time whereby the operator returns arbitrarily close (but never exactly) to the initial operator $\widetilde{\tau}_i^z$. 
Why should the ``reverse evolution'' of the dressed $l$-bit $\tilde{\tau}_i^z$ be a  useful concept?
Consider the \emph{forward} Heisenberg time evolution of an operator $O_i$ localized near site $i$ via $O_i(t) := U(t)^\dagger O_i U(t)$ and ask how this operator spreads over time in space. In an ergodic system, we expect that any local operator spreads generically ballistically, or diffusively at the slowest. 
In our present case, computing the overlap of $O_i(t)$ with the localized dressed $l$-bit $\widetilde{\tau}_i^z$ 
\begin{align}
&  \Tr( O_i(t) \widetilde{\tau}_i^z) = \Tr( O_i \widetilde{\tau}_i^z(t)_R ) 
\end{align}
reveals that this overlap varies quasiperiodically in time rather than decaying to zero. We can interpret this as the statement that some fraction of the operator $O_i(t)$ remains localized near its origin, rather than being transported away.

What we have described is thus a new kind of dynamical localization that can be dubbed ``quasiperiodically-driven MBL'', although we have only shown that it is stable until the timescale $t_*$ bounded by \eqnref{eqn:heating_t}. Proving whether the quasiperiodically-driven MBL is stable to beyond this time, perhaps even forever, remains an interesting direction for future work.

In this paper, we will consider quasiperiodically-driven phases of matter realizable in one or the other of the scenarios described above: prethermalization or (stretched-exponentially long-lived) quasiperiodically-driven MBL.

\section{Emergent symmetries protected by multiple time translation symmetries} 
\label{sec:emergentsymmetries}

Having motivated quasiperiodically-driven systems and outlined their dynamics in suitable regimes, we now analyze what kind of new phases of matter can arise in these systems. As a first step, let us consider the scenario where a direct high-frequency drive is applied to a system.
This procedure is often referred to as `high-frequency Floquet engineering', as the drive is used to modify and control interactions of an underlying Hamiltonian. Indeed, the ground states of the effective static Hamiltonian $D$ that is generated in a high-frequency expansion can be different from those of the original undriven Hamiltonian~\cite{Kitagawa_1104, Grushin_1309, Meinert_1602, Eckardt_1606, Cooper_1803, Lee_1805}. 

However, from a phases of matter point of view, a direct high-frequency drive will not yield fundamentally new long-time collective behavior that is not already reproducible in some -- possibly complicated -- static system at equilibrium.
This is because a quantum state's evolution is effectively governed entirely by $D$ and never has any significant nontrivial micromotion during its time evolution. Precisely, this stems from the fact that the unitary frame transformation $P(t)$ in the description of the time evolution operator \eqnref{eqn:Uapprox2} is perturbatively close to the identity.
To uncover novel phases, especially those that are inherently out-of-equilibrium, we need to go beyond.

In order to do this, we generalize the idea of a ``frame-twisted high-frequency limit'', introduced in Ref.~\cite{Else_1607} for Floquet systems and reviewed in Sec.~\ref{sec:twisted_high_freq_Floquet}, to the quasiperiodically-driven scenario. This will be the context in which fundamentally new long-lived phases of matter can emerge. 
In order to analyze the manifestation of TTSes in this regime, we will introduce the notion of ``twisted time-translation symmetries'' (Sec.~\ref{sec:twistedTTSes}). This will allow us to analyze the quasiperiodic case 
but also gives a simpler perspective on the results in the Floquet case compared to the original constructions of Ref.~\cite{Else_1607}. Finally, in Sec.~\ref{sec:twisted_high_freq_QP}, we explain how to realize these twisted time-translation symmetries in a frame-twisted high-frequency limit in quasiperiodically driven systems.

\subsection{Review: Frame-twisted high-frequency limit in Floquet systems}
\label{sec:twisted_high_freq_Floquet}

For Floquet systems ($m=1$), Ref.~\cite{Else_1607} provided general periodic driving conditions which do give rise to fundamentally new non-equilibrium phases. We briefly review these here.

The main idea is to consider periodically-driven Hamiltonians that approach the high-frequency limit, but only when viewed in a certain rotating frame, a so-called `frame-twisted high-frequency' limit.
Consider a  many-body driven system with Hamiltonian of the form $H(t) = H_0(t) + V(t)$. Here $H_0(t) = H_0(t+T)$ is a sum of quasilocal terms, with associated time-evolution operator $U_0(t) = \mathcal{T} \exp[-i \int_{0}^{t} H_0(t') dt']$. Evolution under $H_0$ is taken to have the special property   $U_0(NT) = U_0(T)^N = X^N = \mathbb{I}$, for some positive integer $N > 1$ and an operator $X$, which is not itself the identity. The term $V(t)=V(t+T)$ describes interactions assumed to have  local energy scale $J \ll \omega/N$. 
Since $X$ is not perturbatively accessible from the identity, a strong drive is required to realize this evolution -- the local energy scale of $H_0$ is $\sim \omega/N$, and thus increases with $\omega$. One therefore cannot naively apply the high-frequency expansion \eqnref{eqn:MagnusD_original}.

In the rotating frame defined by the interaction picture of $H_0(t)$,  time-evolution is governed by the interaction Hamiltonian 
\begin{align}
H_\text{int}(t) = U_0^\dagger(t) V(t) U_0(t).
\label{eqn:Hint}
\end{align}
This is a quasilocal Hamiltonian which is still time-periodic, albeit with period $NT$. 
Since $J \ll \omega/N$, a high-frequency expansion can be meaningfully applied to it, and one can see how there is a long-lived preheating regime in this frame of reference.

However, the frame-twisted high frequency limit imposes a stronger condition than just long-lived preheating.  Specifically, as shown in Ref.~\cite{Else_1607}, in the laboratory frame the Floquet unitary $U_F := U(T)$ takes on the special structure
\begin{align}
U_F \approx \mathcal{V} (X e^{-i D T} ) \mathcal{V}^\dagger,
\label{eqn:prethermal_result}
\end{align}
where $D$ is a quasilocal Hamiltonian which additionally satisfies $[D,X] = 0$ identically. Here $\mathcal{V}$ is a time-independent quasilocal unitary that is perturbatively close to the identity $\mathbb{I}$.
The approximate equality reflects an omission of small time-dependent local terms whose effects only become relevant after times $t_* \sim O(e^{\text{const.} \omega/J})$.

The physical statement is that a system periodically driven under these conditions always has an emergent $\mathbb{Z}_N$ symmetry generated by $X$.
One can add small, potentially time-dependent perturbations to $H(t)$ as long as they respect the time-periodic nature of the drive, and the structure of \eqnref{eqn:prethermal_result} will be unchanged. The emergent symmetry is therefore robust and underpins the stability of inherently nonequilibrium Floquet phases of matter in many-body systems that can now emerge. 
From \eqnref{eqn:prethermal_result}, one sees that when observed at times $t$ that are integer multiples of $NT$ (i.e.~at times which are stroboscopic with respect to the longer period), the system, in the time-independent frame described by $\mathcal{V}$, settles into an equilibrium state of $D$ distinguished by the emergent $\mathbb{Z}_N$ symmetry.
However, when viewed after every time interval $T$ (the original period of the driving Hamiltonian), due to the  action of the symmetry operator $X$, the state of the system can transform nontrivially. This happens, for example, if the system spontaneously breaks the symmetry. 
Note that both the concepts of equilibration and spontaneous symmetry breaking are only sharply defined in a thermodynamically large system.
This additional periodic action is precisely what makes the long-time collective behavior of this system inherently out-of-equilibrium, and  the robustness of the phenomenology  justifies the terminology of them being called fundamentally nonequilibrium phases of matter.
We reiterate that this remarkable result is a consequence of the discrete TTS of the Floquet drive and guaranteed to happen with no additional symmetry requirements.

\subsection{Twisted time-translation symmetries and emergent symmetries}
\label{sec:twistedTTSes}

In quasiperiodically-driven $(m \geq 2)$ systems, it is natural to look for an expression of the form \eqnref{eqn:prethermal_result}. Since there is no single time-translation symmetry, however, there is no analog of the Floquet operator $U_F$ and the construction of Ref.~\cite{Else_1607} does not carry over.

Our key observation is that one can rederive the results of Ref.~\cite{Else_1607} for Floquet systems in a considerably simpler way that  does accord an extension to quasiperiodically-driven systems. This relies on realizing that the interaction Hamiltonian  $H_\text{int}(t)$ described in the previous section possesses a symmetry as a consequence of the TTS of the original laboratory frame Hamiltonian. We will refer to this symmetry of $H_{\mathrm{int}}(t)$ as a `twisted time-translation symmetry'.

Precisely, in the Floquet setting, we say that a time-periodic operator $O(t) = O(t+T)$ has a twisted TTS, if there is an integer $N > 1$ and a unitary operator ${g}$ satisfying ${g}^N = \mathbb{I}$, such that for all $t$
\begin{align}
{O}(t+\tilde{T}) = {g} {O}(t) {g}^\dagger,
\end{align}
where $\tilde{T} = T/N$.
In terms of the Fourier modes ${O}_n$ in $O(t) = \sum_n O_n e^{i n \omega t}$, the twisted-TTS states that ${g} {O}_n {g}^{\dagger} = e^{2\pi i n/N} {O}_n$.  
Note that $H_\text{int}(t)$ in \eqnref{eqn:Hint} has a twisted-TTS with unitary ${g} = X^\dagger$, provided we rescale time $t \mapsto t/N$ so that the periodicity of $H_\text{int}(t)$, originally $NT$, becomes $T$.

To gain some intuition as to why this concept is useful, suppose that we had a Hamiltonian $H(t)$ with a twisted-TTS and we were to construct the effective Hamiltonian $D$ from the high frequency expansion given by \eqnref{eqn:MagnusD} with $m = 1$.
One immediately sees from the action of twisted-TTS in Fourier space that $D$ constructed this way commutes with ${g}$ to all orders.
Additionally, it can be shown that the change of frame $P(t)$ will also inherit the same twisted-TTS as $H(t)$, i.e.~$P(t+\tilde{T}) = g P(t) g^\dagger$.
Note that the usual Floquet-Magnus expansion \eqnref{eqn:MagnusD_original} will not give this result, see Appendix \ref{appendix:FloquetMagnus}.
The physical conclusion is that  dynamics of a driven Hamiltonian with a twisted TTS can always be viewed in some frame as effectively governed by a time-independent Hamiltonian $D$ with an emergent $\mathbb{Z}_N$ internal symmetry, generated by ${g}$.

Applying these considerations to the rotating frame Hamiltonian $H_\text{int}(t)$ in \eqnref{eqn:Hint} to construct the effective Hamiltonian $D$ using \eqnref{eqn:MagnusD}, ones recovers the statement \eqnref{eqn:prethermal_result} of Ref.~\cite{Else_1607} in a transparent fashion: Since $U_F   = U_0(T) U_\text{int}(T) = X \mathcal{T}e^{- i \int_0^T H_\text{int}(t') dt'}$, we can write it as
\begin{align*}
U_F \approx X P(T) e^{-i D T} P(0)^\dagger = P(0) \left( X e^{-i D T} \right) P(0)^\dagger,
\end{align*}
with $[D,X] = 0$, and where we have used $X P(T) = P(0) X$.

The twisted-TTS concept immediately generalizes to quasiperiodically-driven systems ($m \geq 2$). Recall that a time-quasiperiodic operator ${O}(t) = {O}(\vec{\omega} t + \vec{\theta}_0)$ is derived from an operator ${O}(\vec{\theta})$ that is parameterized by a variable $\vec{\theta}$ living in a higher-dimensional space, where ${O}(\vec{\theta} + \vec{\tau}) = {O}(\vec{\theta})$ for any $\vec{\tau} \in \mathcal{L}= 2\pi \mathbb{Z}^m$.   Suppose, there is additionally some finite translation vector $\tilde{\vec{\tau}}$ and a unitary operator ${g}_{\tilde{\vec{\tau}}}$ satisfying ${g}_{\tilde{\vec{\tau}}}^N = \mathbb{I}$ for some integer $N>1$, such that 
\begin{align}
\label{eqn:twistedTTS}
{O}(\vec{\theta} + \tilde{\vec{\tau}}) = {g}_{\tilde{\vec{\tau}}} {O}(\vec{\theta}) {g}_{\tilde{\vec{\tau}}}^{\dagger}.
\end{align}
We then say that ${O}(\vec{\theta})$ has a ${g}_{\tilde{\vec{\tau}}}$-twisted time-translation symmetry. 
Note that as ${O}(\vec\theta + N \tilde{\vec{\tau}}) = {O}(\vec\theta)$,   $ N \tilde{\vec{\tau}} \in \mathcal{L}$.
In terms of Fourier modes $O_{\vec{n}}$, the twisted-TTS acts as ${g}_{\tilde{\vec{\tau}}} {O}_{\vec n} {g}_{\tilde{\vec{\tau}}}^{\dagger} = e^{i \vec{n} \cdot \vec{\tilde\tau}} {O}_{\vec n}$.
Furthermore, since ${O}(\vec\theta)$ is defined  on a torus with dimension $m\geq 2$, it can have \emph{multiple} independent twisted-TTSes corresponding to different translation vectors $\tilde{\vec{\tau}}$ and unitary operators ${g}_{\tilde{\vec{\tau}}}$. 

As in the Floquet case, for a Hamiltonian with   twisted-TTSes, the effective Hamiltonian $D$ constructed through \eqnref{eqn:MagnusD} manifestly commutes with $g_{\tilde{\vec{\tau}}}$. 
Also, $P(\vec\theta)$ will   inherit the same twisted-TTSes $P(\vec\theta + \vec{\tilde\tau}) = g_{\vec{\tilde\tau}} P(\vec\theta) g^\dagger_{\vec{\tilde\tau}}$. 
These properties will hold for other high-frequency expansions as well, like the one we develop in Sec.~\ref{sec:proof}.

\subsection{Frame-twisted high-frequency limit in quasiperiodically-driven systems}
\label{sec:twisted_high_freq_QP}

Although \eqnref{eqn:twistedTTS} seems like an obscure condition, analogous to the Floquet case, twisted TTSes can arise naturally in a frame-twisted high-frequency limit of a drive that does not need to satisfy additional symmetry constraints beyond time-quasiperiodicity itself. 
Indeed, we will see how $m$ twisted-TTSes can manifest from $m$ `untwisted' time-translations of the original driving Hamiltonian $H(\vec\theta)$, when viewed in a suitable frame of reference.

As discussed in Sec.~\ref{sec:intro_qp_tts}, the many-body Hamiltonian of the system $H(t) = H(\vec\omega t + \vec\theta_0)$, derives from a Hamiltonian $H(\vec\theta)$ on the standard torus $\mathbb{T}^m$. Assume now this has the form
\begin{equation}
H(\vec{\theta}) = H_0(\vec{\theta}) + V(\vec{\theta}),
\label{eqn:QP_H}
\end{equation}
with $H_0(\vec{\theta} + \vec\tau) = H_0(\vec{\theta})$ and $V(\vec{\theta} + \vec\tau) = V(\vec{\theta})$, for any   $\vec\tau \in \mathcal{L} = 2\pi \mathbb{Z}^m$.  Let $\Gamma_{1},\ldots,\Gamma_{r}$ be a set of $r$ mutually commuting operators, each of which is a sum of quasilocal terms and has integer eigenvalues. We take
\begin{equation}
\label{eq:time_ev_h0}
H_0(\vec{\theta}) = f_i(\vec{\theta}) \Gamma_i.
\end{equation}
Unless otherwise stated, the index summation convention is implied.
Here $f_i(\vec{\theta})$ are real-valued functions satisfying $\overline{f}_i := \int_{\mathbb{T}^m} \frac{d^m\vec\theta}{(2\pi)^m} f_i(\vec\theta) =  Q_{ij}\omega_j$, where $Q_{ij}$ is a dimensionless $r \times m$ matrix with rational entries. The interaction term $V(\vec{\theta})$ is a quasilocal Hamiltonian with local energy scale $J$. 

Under these conditions, we can solve for the evolution operator $U_0(t) = \mathcal{T} \exp[{ -i\int_0^t H_0(\vec{\theta_0} + \vec{\omega} t') dt'  }]$. Since it is made from commuting terms, the time-ordering can be neglected. It is itself quasiperiodic $U_0(t) = U_0( \vec{\omega} t + \vec{\theta}_0)$, and can be expressed as 
\begin{equation}
\label{eq:gi_periodicity_new}
U_0(\vec{\theta}) = e^{ -ih_i(\vec{\theta})\Gamma_i }, 
\end{equation}
for some functions $h_i(\vec\theta)$ satisfying $h_i(\vec{\theta} + \vec{\tau}) = h_i(\vec{\theta}) + Q_{ij} \tau_j$
for all  $\vec{\tau} \in \mathcal{L}$ (see Appendix \ref{appendix:gi_details}). 

Notice that $U_0(\vec{\theta})$ is defined on a larger torus than $H_0(\vec\theta)$. Specifically:
\begin{equation}
\label{eq:time_ev_periodicity_new}
U_0(\vec{\theta} + \vec{\tau}') = U_0(\vec{\theta}), \quad \mbox{for all $\vec{\tau}' \in \mathcal{L}'$}.
\end{equation}
Here $\mathcal{L}'$ is a sublattice of the original lattice $\mathcal{L}$ defined by $\vec{\tau}' \in \mathcal{L'}$ if and only if $\vec{\tau}' \in \mathcal{L}$ and $e^{iQ_{ij} {\tau}_j'} = 1$.
A simple example is where $m=r$ and $Q = ({1}/{N}) \mathbb{I}_m$ for some integer $N > 1$; in that case $\mathcal{L}' = 2\pi N \mathbb{Z}^m = N \mathcal{L}$, so that the basic original cell is enlarged by $N$ in each direction.
We emphasize that \eqnref{eq:time_ev_periodicity_new} holds only due to our special form of $H_0(t)$ from \eqnref{eq:time_ev_h0}. In general, one does not expect $U_0(t)$ to be quasiperiodic, even if $H_0(t)$ is.

\begin{figure}[t]
\includegraphics[width=\columnwidth]{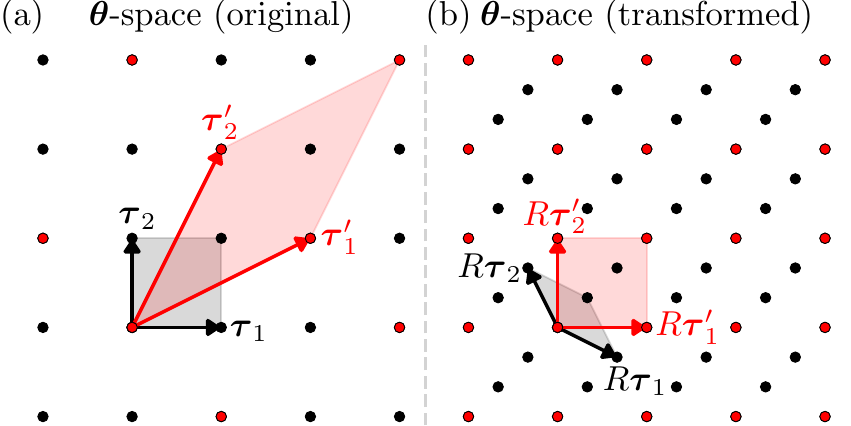}
\caption{\label{fig:freqlattices} 
Twisted-TTS and action of linear transformation $R$ on $\vec\theta$-space for a specific $m=2$ example.
(a) Lattice $\mathcal{L}$ of points (black) is defined by the translation vectors $\vec\tau_1 = 2\pi(1,0)$ and $\vec\tau_2 = 2\pi(0,1)$. The original driving Hamiltonian $H(\vec\theta)$, \eqnref{eqn:QP_H}, has periodicity on the standard torus (grey shaded area).
The interaction Hamiltonian $H_\text{int}(\vec\theta)$  generally has different periodicity. Shown is an example where its periodicity is on a sheared torus (red shaded area), given by the translation vectors $\vec{\tau}'_1 = 2\pi(2,1)$ and $\vec{\tau}'_2 = 2\pi(1,2)$ which defines a lattice $\mathcal{L}'$ (red points).
(b) There exists an invertible linear transformation $R$ mapping $\vec\theta \mapsto R\vec\theta$ so that $(\vec{\tau}'_1,\vec{\tau}'_2 ) \mapsto (R\vec{\tau}'_1, R\vec{\tau}'_2) = (\vec\tau_1,\vec\tau_2)$. The Hamiltonian $H_\text{int}'(\vec\theta) := H_\text{int}(R^{-1} \vec\theta)$  has periodicity on the standard torus (red shaded area), and additionally has $g_{{\vec{\tau}}}$-twisted TTSes $H_\text{int}'(\vec\theta + R \vec\tau) = g_{\vec\tau}H'_\text{int}(\vec\theta) g_{\vec\tau}^\dagger$
where $\vec\tau \in 2\pi \mathbb{Z}^2$.}
\end{figure} 

We are now in a position to see how twisted TTSes emerge. By transforming into the interaction picture of $H_0(t)$, the interaction Hamiltonian is
\begin{align}
H_{\mathrm{int}}(\vec{\omega} t + \vec{\theta_0})= U_0^{\dagger}(t) (H(t) - i \partial_t) U_0(t),
\label{eqn:interactingH}
\end{align}
and has local energy scale $J$.
Furthermore, it derives from a Hamiltonian
$H_{\mathrm{int}}(\vec{\theta}) = U_0^{\dagger}(\vec{\theta}) V(\vec{\theta}) U_0(\vec{\theta})$, which has 
periodicity in the larger unit cell $H_{\mathrm{int}}(\vec{\theta} + \vec{\tau}') = H_{\mathrm{int}} (\vec{\theta})$ with  $\vec{\tau}' \in \mathcal{L}'$. Choosing $Q$ to have rational entries ensures that the new unit cell is still finite (see Appendix \ref{appendix:rational}).

Since the $\Gamma_i$ commute, \eqnref{eq:gi_periodicity_new} implies $U_0(\vec{\theta} + \vec{\tau})^\dagger = g_{\vec{\tau}} U_0(\vec{\theta})^\dagger = U_0(\vec{\theta})^\dagger g_{\vec{\tau}}$ where
\begin{align}
g_{\vec{\tau}} = \exp\left( i Q_{ij}  \Gamma_i \tau_j  \right)
\end{align}
for $\vec{\tau} \in \mathcal{L}$. Together with $V(\vec\theta) = V(\vec\theta + \vec\tau)$, this yields  
\begin{equation}
\label{eq:rotatingframetwistedtts}
H_{\mathrm{int}}(\vec{\theta} + \vec{\tau}) = g_{\vec{\tau}} H_{\mathrm{int}} (\vec{\theta}) g_{\vec{\tau}}^\dagger.
\end{equation}
This is almost the twisted-TTS condition \eqnref{eqn:twistedTTS}, except the periodicity of $H_{\mathrm{int}}(\vec\theta)$ is not on the standard torus $\mathbb{T}^m$. However, there is always an invertible linear transformation $R$ on $\vec\theta$-space (Fig.~\ref{fig:freqlattices}) so that the Hamiltonian $H'_\text{int}(\vec\theta) := H_\text{int}(R^{-1} \vec\theta)$ is periodic on the standard torus and \eqnref{eqn:twistedTTS} then holds exactly for $H'_\text{int}(\vec\theta)$.

We will  generally not need to invoke the transformation $R$, except when we discuss estimates of heating times and our construction of effective Hamiltonians in the preheating regime, which take as input high-frequency Hamiltonians that are periodic on the standard torus.
In those cases, we should remember that
the coordinate transformation means the frequency vector is also rescaled $\vec\omega \mapsto R \vec \omega$ when we consider the single-time evolution, since $H_\text{int}(\vec\omega t + \vec\theta_0) = H'_\text{int}(R\vec\omega t + R \vec\theta_0)$.
This implies that, the dynamics governed by the effective Hamiltonian $D$ constructed from $H'_\text{int}(\vec\theta)$ in a high-frequency expansion  assuming $J \ll R_{ij} \omega_j$ will last for a long time $t_* \sim O(e^{\mathrm{const.} \times (|R \vec\omega| /J)^{1/(m+\epsilon)}})$; see Sec.~\ref{sec:proof}. Then, for times less than $t_*$, the evolution operator in the interaction frame can be written as
\begin{align}
U_{\mathrm{int}}(t) \approx  P(\vec{\omega} t + \vec{\theta}_0) e^{-i D t} P^\dagger(\vec\theta_0),
\end{align}
where $P(\vec{\theta})$ is periodic with respect to translations   $\vec\tau' \in \mathcal{L}'$.

According to the discussion of Sec.~\ref{sec:twistedTTSes}, the effective Hamiltonian $D$ in the preheating regime will have emergent symmetries arising from the twisted TTSes,  
\begin{equation}
\label{eq:Dsym_new}
[D, g_{\vec{\tau}}] = 0, \quad \mbox{for all }\vec{\tau} \in \mathcal{L}.
\end{equation}
Moreover, $P(\vec\theta + \vec\tau) = g_{\vec\tau} P(\vec\theta) g_{\vec\tau}^\dagger$ for $\vec\tau \in \mathcal{L}$.
Analogous to the Floquet case, these symmetry properties of $D$ are robust to  small, potentially time-dependent perturbations to the driving protocol, as long as they respect the time-quasiperiodicity of the system. 

Although \eqnref{eq:Dsym_new} holds for each translation vector $\vec{\tau} \in \mathcal{L}$, not every such $\vec{\tau}$ corresponds to a different operator $g_{\vec{\tau}}$. In fact, $g_{\vec{\tau}} = \mathbb{I}$ if and only if $\vec{\tau} \in \mathcal{L}'$. 
The unitary operators $g_{\vec{\tau}}$ therefore belong to a finite Abelian group of emergent symmetries, $\mathcal{G} \cong \mathcal{L}/\mathcal{L}'$. 
In the simple case where $Q = ({1}/{N}) \mathbb{I}_{m }$, we find that $m = r$, $\mathcal{G} = \mathbb{Z}_N^{\times m}$.  
As a slightly less trivial example, consider $r = m = 2$ and $Q = \frac{1}{3} \begin{pmatrix} -1 & 2 \\ 2 & 1  \end{pmatrix}$. This results in a lattice $\mathcal{L}'$ as seen in Fig.~\ref{fig:freqlattices}, and $\mathcal{G} = \mathbb{Z}_3$.
We refer the reader to Appendix \ref{appendix:rational} where  we show, given some rational matrix $Q$, how to compute $\mathcal{G}$ from the Smith decomposition of $Q$. Since $\mathcal{G}$ is a finite Abelian group, it is always of the form $\mathcal{G} = \mathbb{Z}_{q_1} \times \mathbb{Z}_{q_2} \times ...$ \cite{ClarkBook}.

Collecting all the ingredients discussed above, we are now in the position to realize novel, inherently out-of-equilibrium phases of matter. Time evolution in the laboratory frame for $t < t_*$, under the driving scenarios outlined in this section, is governed by the evolution operator $U(t) = U_0(t) U_\text{int}(t)$, which can be written as
\begin{align}
U(t) & \approx U_0(\vec\omega t + \vec\theta_0) P(\vec\omega t + \vec\theta_0) e^{-i D t} P^\dagger(\vec\theta_0) \nonumber \\
& = \mathcal{V}(t) \Big[ U_0(t) e^{-i D t}  \Big] \mathcal{V}^\dagger(0).
\label{eqn:prethermal_result_QP}
\end{align}
Here $\mathcal{V}(t)  =  U_0(\vec \omega t + \vec\theta_0) P(\vec\omega t + \vec\theta_0) U_0^\dagger(\vec \omega t  + \vec\theta_0)$ is time-quasiperiodic  with underlying $\mathcal{V}(\vec\theta) = U_0(\vec\theta) P(\vec\theta) U_0^\dagger(\vec\theta)$ that has periodicity on the standard torus $\mathbb{T}^m$, that is, for translations $\vec\tau \in \mathcal{L}$. {Another way to state \eqnref{eqn:prethermal_result_QP} is that in the rotating frame defined by $\mathcal{V}(t) U_0(t)$, time evolution of the system is simply governed effectively by the static Hamiltonian $D$. However, if one goes back to the laboratory frame, this frame transformation} -- being not perturbative close to the identity -- can endow the state of the system (in particular, the steady state of $D$) with large, structured time-quasiperiodic micromotion, giving rise to a panoply of different long-time dynamical collective behaviors whose dynamical signatures are robust and universal.

In the next two sections, we will illustrate the physical implications of our results with two examples of such phases: time quasi-crystals and dynamic quasiperiodic topological phases.
We will return to the important task of formalizing the preceding discussions on dynamics as well as explicitly constructing $D$, in the later sections.

\section{Discrete time quasi-crystals}
\label{sec:tqc}

A discrete time quasi-crystal (DTQC) is a phase which spontaneously breaks some or all of the time-translation symmetries of a quasiperiodic drive~\cite{Dumitrescu_1708}.
It is characterized by a dynamical response of physical observables, which display stable long-time oscillations with a time-quasiperiodicity that is different from the time-quasiperiodicity of the original driving Hamiltonian $H(t)$. 
This can be diagnosed by computing the power spectra of local observables, which will exhibit robust peaks at frequencies which are shifted from the base frequencies by a fractional amount.
Since there are several ($m \geq 2$) independent time-translation symmetries, there  are a multitude of ways that these symmetries can be spontaneously broken, leading to a variety of patterns and associated DTQC phases.
The DTQC generalizes the discrete time crystal, a phase which spontaneously breaks the single time-translation symmetry of Floquet systems~\cite{Khemani_1508,Else_1603}.

{Note that the concept of a DTQC  as well as some aspects of its phenomenology have been proposed in \cite{Dumitrescu_1708}, which numerically observed a DTQC-like signal in a quasiperiodic step-drive with disorder; albeit with a slow logarithmic decay of the envelope in time. Our present work explains more generally the precise role of time-translational symmetries of quasiperiodic drives in delineating such a phase, and also shows how the logarithmic decay can be avoided (even without disorder) through smooth driving, hence rigorously proving the stability of the phase up to the long heating time $t_*$.
In addition, we provide drives that generalize to a large class of   symmetry breaking patterns.
This large class includes, among many others, the DTQC pattern introduced in [64], as well as that of [66], and we address the stability of such patterns when using smooth driving.
 
To understand DTQC phases, consider the time evolution of a quantum state $|\psi_0\rangle$ with a quasiperiodic Hamiltonian of the type discussed in Sec.~\ref{sec:twisted_high_freq_QP}. Then one can take a frame-twisted high-frequency limit, so that for times $t < t_*$, the time-evolution operator can be decomposed  as in  \eqnref{eqn:prethermal_result_QP}.  Recall that the effective time-independent Hamiltonian $D$ in the preheating regime possesses multiple unitary symmetries ${g}_{\vec\tau}$ where $\vec\tau \in \mathcal{L} = 2\pi\mathbb{Z}^m$, which belong to some finite Abelian group $\mathcal{G}$.

Now let us consider times where the system has prethermalized, 
so that in the rotating frame $\mathcal{V}(t) U_0(t)$ the state is locally described by a thermal state $\rho_\beta$; see the discussion in Sec.~\ref{subsec:long_time_preheating}. 
 In the laboratory frame, the state of the system when probed by local observables is $\rho(t) = \rho(\vec{\omega} t + \vec{\theta}_0)$, where
\begin{equation}
\label{eqn:rho_rotating}
\rho(\vec{\theta}) = \mathcal{V}(\vec{\theta}) U_0(\vec{\theta}) \rho_\beta U_0^{\dagger}(\vec{\theta}) \mathcal{V}^{\dagger}(\vec{\theta}).
\end{equation}
We see that $\rho(t)$ is time-quasiperiodic, since $\rho(\vec{\theta})$ \emph{at least} satisfies $\rho(\vec{\theta} + \vec{\tau}') = \rho(\vec{\theta})$ for $\vec{\tau}' \in \mathcal{L}'$. But does $\rho(\vec{\theta})$ have the periodicity of the original drive $H(\vec\theta)$ of \eqnref{eqn:QP_H}, characterized by the lattice $\mathcal{L} = 2\pi \mathbb{Z}^m$? This turns out to depend on whether or not $\rho_\beta$ is symmetric under the emergent symmetries $g_{\vec\tau}$.

To see this explicitly, we can write
\begin{equation}
\label{eqn:rhotheta}
\rho(\vec{\theta} + \vec{\tau}) = \mathcal{V}(\vec{\theta}) U_0(\vec{\theta}) g_{\vec{\tau}}^{\dagger} \rho_\beta g_{\vec{\tau}} U_0^{\dagger}(\vec{\theta}) \mathcal{V}^{\dagger}(\vec{\theta}),
\end{equation}
where we use that $\mathcal{V}(\vec{\theta} + \vec{\tau}) = \mathcal{V}(\vec{\theta})$ and $U_0(\vec{\theta} + \vec{\tau}) = U_0(\vec{\theta}) g_{\vec{\tau}}^{\dagger}$. Therefore $\rho(\vec{\theta} + \vec{\tau}) = \rho(\vec{\theta})$ for $\vec\tau \in \mathcal{L}$ if and only if $g_{\vec{\tau}} \rho_\beta g_{\vec{\tau}}^{\dagger} = \rho_\beta$.
In other words, if $\rho_\beta$ is a state that preserves all the symmetries $g_{\vec \tau}$ of the effective Hamiltonian $D$, then $\rho(\vec\theta)$ preserves all $m$ multiple time translation symmetries of the driving Hamiltonian $H(\vec\theta)$.

If on the other hand $\rho_{\beta}$ is not invariant under  $g_{\vec \tau}$, $g_{\vec \tau} \rho_\beta g_{\vec \tau}^{\dagger} \neq  \rho_\beta$, then the emergent symmetry $g_{\vec \tau}$ is said to be spontaneously broken in the thermal state of $D$. This can, of course, happen for multiple $g_{\vec \tau}$ at the same time.
From \eqnref{eqn:rhotheta}, $\rho(\vec\theta)$ will then have a periodicity different from the original Hamiltonian $H(\vec\theta)$, and consequently $\rho(t) = \rho(\vec\omega t + \vec\theta_0)$ will have a different time-quasiperiodicity than the driving Hamiltonian $H(t) = H(\vec\omega t + \vec\theta_0)$ -- the hallmark of a DTQC phase. The precise connection between the spontaneous breaking of $g_{\vec \tau}$ and of TTS reflects the fact that $g_{\vec{\tau}}$ was a manifestation of the TTS in the first place.

\subsection{Observable consequences}
\label{subsec:tqc_obs}

The spontaneously broken TTSes in quasiperiodically-driven systems manifest themselves most clearly through periodicity changes in the $m$-dimensional $\vec{\theta}$-space. The interplay of multiple time translation symmetries gives a large variation in the number of different symmetry breaking patterns and associated DTQC phases. 
However, there will also be measurable signatures of these patterns in terms of the dependence of the system on physical time $t$,  for example,  in the Fourier spectrum (or power spectrum) of local observables (see also \cite{Dumitrescu_1708}). These are analogous to probing quasicrystalline structures in space through their diffraction patterns~\cite{Levine_1986, Socolar_1986}. 

Consider the regime described above, where the state of the system is described by \eqnref{eqn:rhotheta}. Then, the expectation of a local observable $o(t) := \langle \hat{o}(t) \rangle$ can be written as $o(t) = o(\vec{\omega} t + \vec{\theta}_0)$, where $o(\vec{\theta}) := \Tr(\hat{o} \rho(\vec{\theta}))$ has periodicity that depends on the periodicity of $\rho(\vec{\theta})$. Let $\mathcal{L}_{SSB}$ be the sublattice of $\mathcal{L}$ comprising those $\vec{\tau} \in \mathcal{L}$ such that $\rho(\vec{\theta} + \vec{\tau}) = \rho(\vec{\theta})$ for all $\vec{\theta}$; $\mathcal{L}_{SSB}$ describes the symmetry-breaking pattern in $\vec{\theta}$ space. Then we can expand $o(\vec{\theta})$ as a Fourier series
\begin{equation}
o(\vec{\theta}) = \sum_{\vec{\alpha}\in \mathcal{L}_{SSB}^*} e^{i \vec{\alpha} \cdot \vec{\theta}} o_{\vec{\alpha}},
\end{equation}
where the sum is over the reciprocal lattice vectors $\vec{\alpha} \in \mathcal{L}_{SSB}^*$, which are the vectors $\vec\alpha$ satisfying $e^{i \vec{\alpha} \cdot \vec{\tau}} = 1$ for all $\vec{\tau} \in \mathcal{L}_{SSB}$. 
Consequently, the power spectrum of $o(t)$  
\begin{equation}
\label{eq:powerspectrum}
\mathcal{P}_o(\Omega) = \sum_{\vec{\alpha} \in \mathcal{L}_{SSB}^*} |o_{\vec{\alpha}}|^2 \delta(\Omega - \vec{\alpha} \cdot \vec{\omega}),
\end{equation}
has peaks at frequencies $\Omega_{\vec \alpha} = \vec{\alpha} \cdot \vec{\omega}$. 
Note that for smooth driving, the Fourier coefficient $o_{\vec{\alpha}}$ will decay exponentially with $|\vec{\alpha}|$; furthermore, not every peak will necessarily appear for any choice of observable $\hat{o}$, since it is possible that  $o_{\vec{\alpha}} = 0$ for some $\vec{\alpha}$'s. 

Now, in a DTQC, since the symmetry lattice $\mathcal{L}$ spontaneously breaks to the proper sublattice $\mathcal{L}_{SSB}$, the reciprocal lattice $\mathcal{L}_{SSB}^*$ is also a proper superlattice of $\mathcal{L}^* = \mathbb{Z}^m$. 
This implies that some of the frequencies $\Omega_{\vec\alpha} = \vec\alpha \cdot \vec\omega$ for $\vec{\alpha} \in \mathcal{L}_{SSB}^*$ are not derivable from integer linear combinations of the base driving frequencies $(\omega_1, \cdots, \omega_m)$, i.e.~they do not correspond to the base harmonics $\Omega_{\vec n} = \vec{n} \cdot \vec\omega$ for $\vec{n} \in \mathbb{Z}^m$.  This is the dynamical signature of the spontaneous breaking of the time translation symmetries. 
Of course, the frequencies associated with original drive harmonics $\Omega_{\vec n} = \vec{n} \cdot \vec{\omega}$  are dense on the real line, so they can   lie arbitrarily close to those  frequencies $\Omega_{\vec \alpha}$ reflecting the symmetry-breaking. 
However,  the peaks in the power spectrum at frequencies $\Omega_{\vec \alpha}$ will nevertheless be well resolved from those at $\Omega_{\vec n}$. This is because the weights of peaks at frequencies $\Omega_{\vec n}$ approaching $\Omega_{\vec \alpha}$ become ever more strongly suppressed  as a consequence of the smoothness of the drive, as they involve very large $|\vec{n}|$. Therefore, the presence of well-defined peaks at frequencies $\Omega_{\vec{\alpha}}$ constitutes a sharp dynamical signature of the DTQC phase. 
We discuss this point in greater depth in Section \ref{subsec:how_distinct}, where we also discuss the effect of finite observation time (fundamentally limited by the heating time $t_*$) in resolving these peaks in practice.

Additionally, an important signature of the DTQC phase is that the location of these peaks $\Omega_{\vec\alpha}$ in the power spectrum reflecting the spontaneous symmetry-breaking, is robust against small perturbations to the driving protocol, such as in changing $f(\vec\theta) \mapsto f(\vec\theta) + \epsilon(\vec\theta)$ for small, smooth $\epsilon(\vec\theta)$, or  by adding small time-quasiperiodic terms to the Hamiltonian $V$.

In the following subsection as well as Appendix \ref{appendix:DTQC}, we provide  examples of Hamiltonians that exhibit   DTQC phases.

\subsection{Example Hamiltonian: $\mathbb{Z}_2$ DTQC}
\label{subsec:Z2DTQC}

\begin{figure*}[t]
\includegraphics[width=\textwidth]{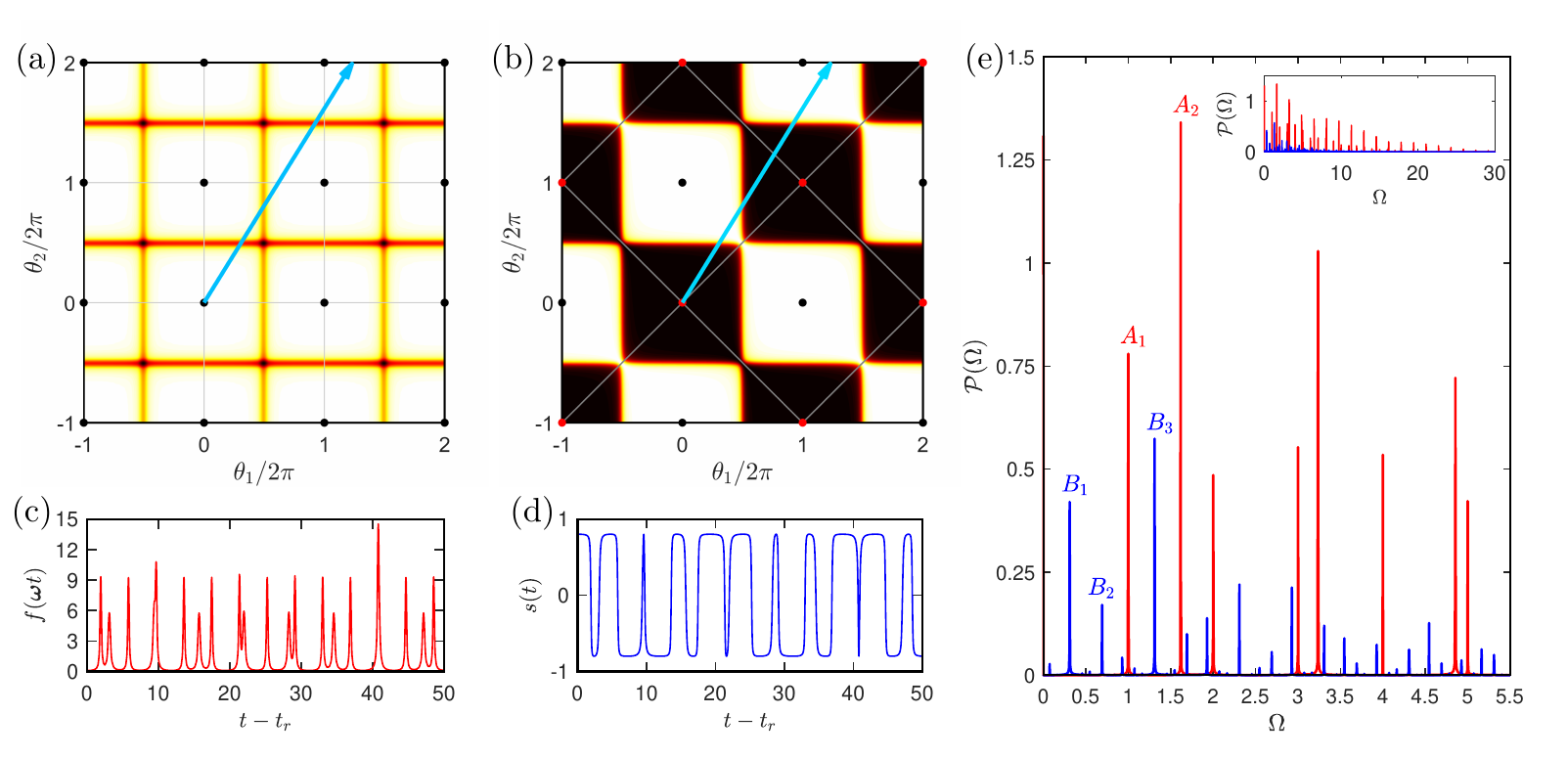}
\caption{\label{fig:tqc_cartoon} 
$\mathbb{Z}_2$ Discrete Time Quasicrystal (DTQC) of Sec.~\ref{subsec:Z2DTQC}: driving function (a,c) and time-quasicrystal response (b,d).  The  spontaneous symmetry breaking of time-translation symmetries as seen in the power spectrum (e) of a local observable.
(a) Driving function $f(\vec\theta)$ in   extended $\vec\theta$-space, with $\Delta(\theta)$ in \eqnref{eqn:driving_profile} replaced by $\Delta_N(\theta)$ in \eqnref{eqn:Fejer} with $N = 20$. The function is symmetric under translations   $\vec\tau \in \mathcal{L} = 2\pi \mathbb{Z}^2$ (black points). Grey boxes represent the unit cell which is the standard torus.  The blue arrow is the single time trajectory $\vec\theta = \vec\omega t$. 
(b) In contrast, observable $s(\vec\theta) = \Tr[\sigma^z_i \rho(\vec\theta)]$ is symmetric under translations by $\vec\tau' = 2\pi(1,\pm 1) \mathbb{Z}$, which defines a symmetry-breaking sublattice $\mathcal{L}_{SSB} = \mathcal{L}'$ (red points). The grey boxes represent the unit cell, which is enlarged and rotated by 45$^{\text{o}}$ with respect to the driving Hamiltonian's unit cell. 
(c) Driving profile $f(t) = f(\vec\omega t)$ is a smooth function in the time domain. Times are measured from the relaxation time $t_r \sim J^{-1}$.
(d) Discrete time quasi-crystal response $s(t) = s(\vec\omega t)$ in the time domain.
(e) Power spectra $\mathcal{P}(\Omega)$ of both the driving function $f(t)$ (red) and the observable $s(t)$ (blue).
The power spectrum of $f(t)$ exhibits peaks at the base frequencies $\Omega = \vec{n}\cdot\vec\omega$, with dominant peaks at $\Omega = \omega_1$ ($A_1$) and  $\Omega = \omega_2$ ($A_2$).
The power  spectrum of $s(t)$ has peaks instead at frequencies shifted from the base frequencies by fractional values, reflecting the spontaneous breaking of the time-translation symmetries of the driving Hamiltonian. 
In particular, the dominant peaks are at $\Omega = (\omega_2 - \omega_1)/2$ ($B_1$), $\Omega = \omega_2 - \omega_1$ ($B_2$), and  $\Omega = (\omega_1 + \omega_2)/2$ ($B_3$).
The inset shows the power spectra over a wider range of frequencies.
}
\end{figure*} 

Consider a system of spin-1/2 degrees of freedom on a lattice, evolving with the $m=2$ time-quasiperiodic Hamiltonian 
\begin{align}
\label{eqn:exampleH_Z2_pre}
H(t)  &= H_0(\vec\omega t) +  \sum_{i, j} J_{ij} \sigma^z_i \sigma^z_j + \sum_i h \sigma^z_i, \\
H_0(\vec\theta) & = \sum_i f(\vec{\theta})  \frac{1}{2} (\sigma^x_i + 1)
\label{eqn:exampleH_Z2}
\end{align}
Here $\sigma^x, \sigma^y, \sigma^z$ are the standard Pauli matrices, and we will choose the couplings $J_{ij}$ to be ferromagnetic so that the Ising Hamiltonian $\sum_{i, j} J_{ij} \sigma_i^z \sigma_j^z$ has an ordered phase at finite temperature. We will also assume that the couplings decay sufficiently fast enough with spatial distance. For example, we can consider short-range nearest-neighbor couplings in two or greater dimensions, or  power-law decaying interactions in one dimension with exponent between one to two; see Sec.~\ref{subsec:longrange} for a discussion of the dynamical consequence of power-law interactions. Note that we could have replaced $\sigma_i^x + 1$ with $\sigma_i^x$ in \eqnref{eqn:exampleH_Z2} without affecting the dynamics, but we chose the former to ensure that $\frac{1}{2}(\sigma_i^x + 1)$ has integer eigenvalues.

The Hamiltonian \eqnref{eqn:exampleH_Z2_pre} falls into the class of Hamiltonians described in Sec.~\ref{sec:twisted_high_freq_QP}.
It comprises of two terms: first, $V$  describes pairwise Ising interactions between spins with amplitude $J_{ij}$, as well as a longitudinal field in the $z$-direction with strength $h$. The couplings are assumed to satisfy $J_{ij}, h \ll |\omega_1 -\omega_2|/2$.
Second, $H_0(\vec\omega t)$  describes a quasiperiodic drive on the system in the $x$-direction,  with frequency vector $\vec\omega = (\omega_1, \omega_2)$.
In experimental platforms where such interactions can be realized, such as with trapped ions or ensembles of nitrogen-vacancy centers in diamond (see Sec.~\ref{subsec:expt_realization}), this can be implemented, for example, by external pulses using lasers or microwaves.

Let us now consider how to choose the driving profile $f(\vec{\theta})$. A natural generalization of models of the DTC previously considered in the Floquet case \cite{Khemani_1508,Else_1603, vonKeyserlingk_1605,
Potirniche_1610} would be to take
\begin{align}
\label{eqn:driving_profile}
f(\theta_1, \theta_2) &= \pi [\omega_1 \Delta(\theta_1) +  \omega_2 \Delta(\theta_2)], \\
\label{eqn:sum_of_deltas}
\Delta(\theta) &= \sum_{n = -\infty}^{\infty} \delta(\theta - \pi + 2\pi n).
\end{align}
If we substitute into \eqnref{eqn:exampleH_Z2_pre} and \eqnref{eqn:exampleH_Z2}, we see that this corresponds to instantaneously applying $\sigma^x$ to all the spins at certain times $t$, namely those for which either $\omega_1 t - \pi$ or $\omega_2 t - \pi$ is an integer multiple of $2\pi$. 
A somewhat similar driving sequence was considered in disordered spin models in Ref.~\cite{Dumitrescu_1708}. However, since the drives considered in that reference as well as in Ref.~\cite{Zhao_1906} are not smooth, slow heating and a stretched-exponentially long-lived prethermal plateau will not be guaranteed by our results.
To circumvent this problem, we will replace the sharp peaks in \eqnref{eqn:sum_of_deltas} with smoothed-out approximations, as we discuss in more detail later.

Let us now observe that $f(\vec{\theta})$ satisfies
$\overline{f} = Q_j \omega_j$, where $Q = \begin{pmatrix} 1/2 & 1/2 \end{pmatrix}$.
Therefore, we can apply the discussion of Sec.~\ref{sec:twisted_high_freq_QP}.
and pass to a frame-twisted high-frequency limit.
Going into the interaction frame of $H_0(\vec\theta)$, we can compute the time-quasiperiodic interaction Hamiltonian $H_\text{int}(\vec\theta) = U_0^\dagger(\vec\theta) V U_0(\vec\theta)$.
The periodicity of $H_\text{int}(\vec\theta)$ is on the lattice $\mathcal{L}'$, generated by the translation vectors $\vec{\tau}_1 = 2\pi(1,1)$ and $\vec{\tau}_2 = 2\pi(-1,1)$. Here $\mathcal{L}'$ is a sublattice of the original lattice $\mathcal{L} = 2\pi \mathbb{Z}^m$ of symmetry vectors of $H(\vec\theta)$.
One also sees that $H_\text{int}(\vec\theta)$  possesses a single nontrivial twisted-TTS corresponding to a translation by $\vec{\tilde{\tau}} = 2\pi(1,0)$ or $\vec{\tilde \tau} = 2\pi(0,1)$:
$ H_\text{int}(\vec\theta + \vec{\tilde{\tau}} ) =g H_\text{int}(\vec\theta) g^\dagger $ with $g =  \prod_i \sigma^x_i$. Indeed, the operator $g$ generates the finite group $\mathcal{G} \cong \mathcal{L}/\mathcal{L}'  = \mathbb{Z}_2$.

\begin{figure}[t]
\includegraphics[width=\columnwidth]{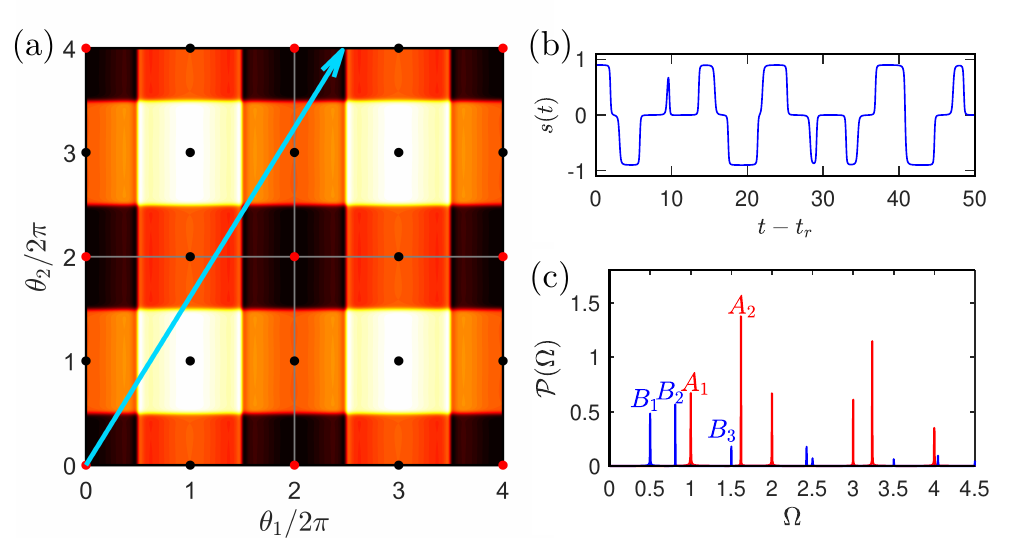}
\caption{\label{fig:prethermalCartoonZ2Z2} 
$\mathbb{Z}_2 \times \mathbb{Z}_2$ DTQC with model and observable detailed in Appendix.~\ref{subsec:Z2Z2DTQC}. (a) Observable $s(\vec\theta)$ in $\vec\theta$-space at leading order in inverse frequency,  displaying the symmetry-breaking pattern lattice $\mathcal{L}_{SSB}$ (red points) as well as the original lattice of translations $\mathcal{L}$ (black points). Grey boxes denote the symmetry-breaking unit cell.
(b) Discrete time quasi-crystal response $s(t) = s(\vec\omega t)$. 
(c) Power spectra $\mathcal{P}(\Omega)$ of both the driving function (red) and the observable $s(t)$ (blue).
Dominant peaks in the driving function correspond to:   $\Omega = \omega_1$ ($A_1$) and  $\Omega = \omega_2$ ($A_2$).
Dominant peaks of the observable $s(t)$ correspond to: $\Omega = \omega_1/2$ ($B_1$), $\Omega = \omega_2/2$ ($B_2$), $\Omega = 3\omega_1/2$ ($B_3$).
Going beyond leading order in inverse frequency, an observable's power spectrum can exhibit peaks at the frequencies $\Omega = \frac{1}{2} \vec{n} \cdot \vec{\omega}$ for $\vec{n} \in \mathbb{Z}^2$.
 }
\end{figure} 

We next construct the effective time-independent Hamiltonian $D$ from the $H_\text{int}$ in a high-frequency expansion, using our approach in Sec.~\ref{sec:proof}.  
However, for our purpose here of understanding its steady states, it suffices to understand the leading order Hamiltonian $D^{(0)}$ in the high-frequency expansion: Since $D$ and $D^{(0)}$ are perturbatively close by construction, the steady states of $D$ and $D^{(0)}$ are in the same universality class.
For our expansion in Sec.~\ref{sec:proof}, the leading order term $D^{(0)}$ is the average of the interaction Hamiltonian
\begin{align}
D^{(0)} = \frac{1}{ 8\pi^2} \int_{0}^{2\pi}\!\! d\theta_2 \int_0^{4\pi} \!\! d\theta_1 \, H_\text{int} (\vec\theta).
\end{align}
We note in particular the bounds on integration, corresponding to a unit cell of $\mathcal{L}'$. We find that, if the driving profile is given by Eqs.~(\ref{eqn:driving_profile}) and (\ref{eqn:sum_of_deltas}), then
\begin{align}
D^{(0)} = \sum_{i < j} J_{ij} \sigma^z_i \sigma^z_j.
\label{eqn:exampleD}
\end{align}
More generally, in order to achieve slow heating, we need to smooth out the driving profile, for example by replacing \eqnref{eqn:sum_of_deltas} with
\begin{align}
\Delta_N(\theta)  = \frac{1}{2\pi} \sum_{\vert n \vert < N} \left(1 - \frac{|n|}{N} \right) e^{-2 {|n|}/{N}} e^{i n (\theta - \pi)}.
\label{eqn:Fejer}
\end{align}
Here $\Delta_N(\theta)$ is a smooth function that approximates the delta function comb $\Delta(\theta)$ increasingly well as $N \to \infty$ and is related to the so-called ``Fejer kernel''. With this replacement, we find
\begin{align}
D^{(0)} = \sum_{i < j} J_{ij} (a(N) \sigma^z_i \sigma^z_j + b(N) \sigma^y_i \sigma^y_j),
\label{eqn:exampleD}
\end{align}
where $a(N), b(N)$ are numerical constants depending on the smoothness parameter $N$.  Furthermore, for large $N$, we note that the ratio $a(N)/b(N)$ will be large. For example, if $N = 20$, then $a(20) = 0.866$ and  $b(20) = 0.134$.  

In accordance with Sec.~\ref{sec:twistedTTSes}, observe that \eqnref{eqn:exampleD} is Ising-symmetric, that is, $[D^{(0)},\prod_i \sigma^x_i] =0$. Furthermore, $D^{(0)}$ is dominated by Ising interactions $\sigma^z_i \sigma^z_j$ along the $z$-direction. This implies that
the steady states of $D$ in the preheating regime are as follows.
Provided the initial state has energy density (measured with respect to $D^{(0)}$) below some critical energy density, the system will prethermalize to a Gibbs state $\rho_\beta$ that spontaneously breaks the Ising symmetry $g = \prod_i \sigma^x_i$. 
This means that, the expectation value of an operator odd under the symmetry, such as the local magnetization along the $z$-direction $\sigma^z_i$, is generically not zero, i.e.~$\Tr[\sigma^z_i \rho_\beta] \neq 0$.

From the discussion in Sec.~\ref{sec:tqc} and Sec.~\ref{subsec:tqc_obs}, we   see that the state $\rho(\vec\theta)$, \eqnref{eqn:rho_rotating}, has periodicity on the lattice $\mathcal{L}_{SSB} = \mathcal{L}'$, and the corresponding reciprocal lattice $\mathcal{L}^*_{SSB}$ is generated by the vectors $\vec\alpha_1 = (1/2,1/2)$ and $\vec\alpha_2 = (1/2,-1/2)$.
Thus, the power spectrum of a   local observable will generically have peaks at the frequencies
\begin{align}
\label{eqn:fractionalOmega}
\Omega = \frac{1}{2}\left[n_1(\omega_1 - \omega_2) + n_2(\omega_1 + \omega_2) \right],
\end{align}
where $n_1, n_2 \in \mathbb{Z}$.

\begin{figure}[t]
\includegraphics[width=\columnwidth]{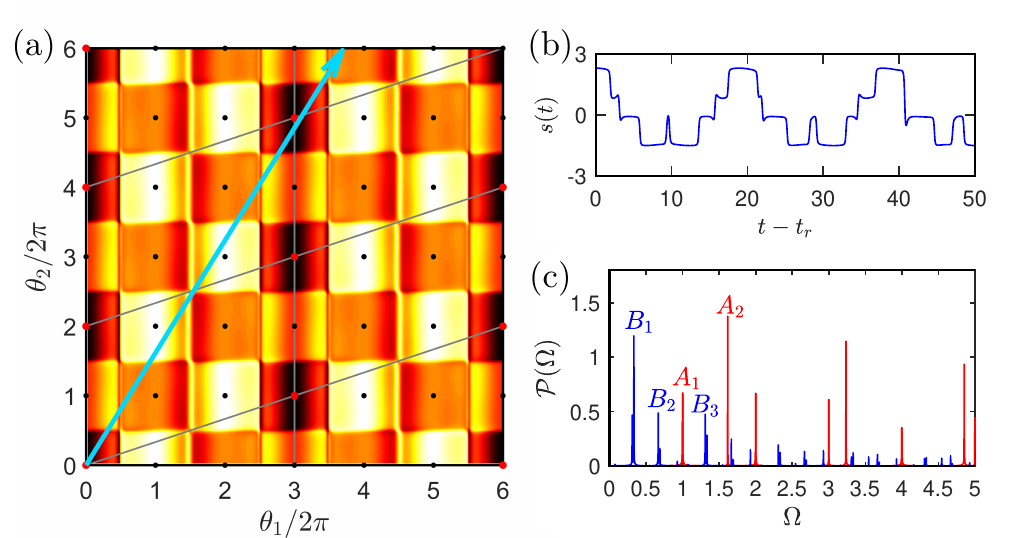}
\caption{\label{fig:prethermalCartoonZ3Z2} 
$\mathbb{Z}_3 \times \mathbb{Z}_2$ DTQC with model and observable detailed in Appendix.~\ref{subsec:Z3Z2DTQC}. (a) Observable $s(\vec\theta)$ in $\vec\theta$-space at leading order in inverse frequency,  displaying the symmetry-breaking pattern lattice $\mathcal{L}_{SSB}$ (red points) as well as the original lattice of translations $\mathcal{L}$ (black points). Grey boxes denote the symmetry-breaking unit cell.
(b) Discrete time quasi-crystal response $s(t) = s(\vec\omega t)$. 
(c) Power spectra $\mathcal{P}(\Omega)$ of both the driving function (red) and the observable $s(t)$ (blue).
Dominant peaks in the driving function correspond to:   $\Omega = \omega_1$ ($A_1$) and  $\Omega = \omega_2$ ($A_2$).
Dominant peaks of the observable $s(t)$ correspond to: $\Omega = \omega_1/3$ ($B_1$), $\Omega = 2 \omega_1/3$ ($B_2$), $\Omega = (\omega_1 + \omega_2)/2$ ($B_3$).
Going beyond leading order in inverse frequency, an observable's power spectrum can exhibit peaks at frequencies $\Omega = n_1 \frac{\omega_1}{3} + n_2 \left( - \frac{\omega_1}{6} + \frac{\omega_2}{2} \right)$ for $n_1, n_2 \in \mathbb{Z}$.
}
\end{figure}

Which peaks dominate, however, depends on the operator measured and its symmetry properties under $g$. 
As a concrete example, suppose we were to measure the local observable $\sigma^z_i$  (which is odd under $g$), whose expectation value in time can be written as 
\begin{align}
s(t) \approx \Tr[ & \sigma^z_i   \mathcal{V}(t) U_0(t) \rho_\beta U_0(t)^\dagger \mathcal{V}(t)^\dagger  ].
\end{align}
Because of the periodicity of $\rho(\vec\theta)$, it can be expressed as $s(t) = s(\vec\omega t + \vec\theta_0)$ for an underlying $s(\vec\theta) = \Tr[\sigma^z_i \rho(\vec\theta)]$ that has periodicity in $\mathcal{L}_{SSB}$ too.
To leading order in inverse frequency (i.e.~treating $\mathcal{V}(t) \approx \mathbb{I}$), the signal can be analytically derived, and is given by
\begin{align}
s(t) \approx & \Tr[ \sigma^z_i \rho_\beta ] \cos g(\vec\omega t) +  \Tr[ \sigma^y_i \rho_\beta ] \sin g(\vec\omega t)
%
\label{eqn:finite_s}
\end{align}
where
\begin{align}
g(\vec{\theta}) &= \pi [\Theta_N(\theta_1) + \Theta_N(\theta_2)], \\
\Theta_N(\theta) &:= \int_0^{\theta} \Delta_N(\theta') d\theta'. \label{eqn:defn_of_Theta}
\end{align}
Note that in the case of delta function driving ($N \to \infty$), this reduces to
\begin{align}
s(t) &\approx \Tr[ \sigma_i^z \rho_\beta] \sigma(\omega_1 t) \sigma(\omega_2 t), \\
\sigma(\theta) &:= (-1)^{\lceil(\theta - \pi)/(2\pi)\rceil}.
\end{align}

In Fig.~\ref{fig:tqc_cartoon}(b,d), we plot both $s(t)$ [\eqnref{eqn:finite_s}] and the underlying function $s(\vec\theta)$ from which it is derived from, assuming $\Tr[\sigma^z_i \rho_\beta] = 0.8$ and $\Tr[\sigma^y_i \rho_\beta] = 0$, and taking $\omega_1 = 1$, $\omega_2 = (1+\sqrt{5})/2$ (the golden ratio).
Fig.~\ref{fig:tqc_cartoon}(e)  illustrates the corresponding  power spectrum $\mathcal{P}(\Omega)$ of the signal $s(t)$.
We see that the spectrum contains peaks at the particular frequencies 
$
\Omega = (n_1+ 1/2 )\omega_1 + (n_2 + 1/2 ) \omega_2$
where $n_1, n_2 \in \mathbb{Z}$; furthermore, those frequencies given by small values of $n_1, n_2$   contribute the most. These frequencies are a subset of the ones described by \eqnref{eqn:fractionalOmega} (in particular, they do not include the harmonics of the original driving frequencies, i.e.\ $\Omega = n_1 \omega_1 + n_2 \omega_2$), but at higher orders in the inverse frequency we expect peaks to occur at all values \eqnref{eqn:fractionalOmega}.
 
While the $\mathbb{Z}_2$ DTQC discussed here is a particularly simple case, we can construct more complicated DTQCs such as ones characterized by the spontaneously broken emergent symmetry groups $\mathcal{G} = \mathbb{Z}_2 \times \mathbb{Z}_2$ (Fig.~\ref{fig:prethermalCartoonZ2Z2}) or $\mathbb{Z}_3 \times \mathbb{Z}_2$ (Fig.~\ref{fig:prethermalCartoonZ3Z2}) in an analogous fashion. We give the explicit form of the systems and drives to realize these in Appendix~\ref{appendix:DTQC}. 

\subsection{The MBL-DTQC}
\label{subsec:mbl_dtqc}

We are not restricted to realizing DTQC phases when the Hamiltonian $D$ is thermalizing. Instead, we can consider the case where $D$ exhibits MBL phenomenology and spontaneously breaks the emergent symmetries $g_{\vec{\tau}}$.
Then, the local integrals of motion $\tau^z_i$ are themselves not invariant under $g_{\vec\tau}$, which is the usual sense of spontaneous symmetry breaking in MBL; see \cite{Huse_1304,Vosk_1307,Pekker_1307,Kjall_1403,Nandkishore_1404}.
For the model we considered in \eqnref{eqn:exampleH_Z2}, we can pick strongly disordered interactions $J_{ij}$, which leads to an ``MBL spin glass'' effective Hamiltonian with Ising symmetry \cite{Huse_1304,Kjall_1403}.  
In the laboratory frame, this gives rise to an MBL DTQC, whose properties we now describe.

As stated in Section \ref{subsec:long_time_preheating}, the defining property of MBL in a quasiperiodically-driven system is that there is a complete set of commuting operators $\widetilde{\tau}^z_i$ (the ``l-bits''), which evolve quasiperiodically under reverse Heisenberg evolution $\widetilde{\tau}_i^z(t)_R = U(t) \widetilde{\tau}_i^z U(t)^{\dagger}$. We can define TTS to be spontaneously broken in a quasiperiodically-driven MBL system if the quasiperiodicity of $\widetilde{\tau}_i^z(t)_R$ is not the one of the drive but a different one.  That is, writing $\widetilde\tau_i^z(t)_R = \widetilde{\tau}_i^z(\vec{\omega} t + \vec{\theta}_0)$, a TTS corresponding to $\vec{\tau} \in \mathcal{L}$ is spontaneously broken if $\widetilde\tau_i^z(\vec{\theta} + \vec{\tau}) \neq \widetilde\tau_i^z(\vec{\theta})$.

For discrete time crystals in periodically driven systems, a key feature was the spectral pairing between eigenstates of the Floquet operator $U_F$.
For example, in the case corresponding to a discrete time crystal with period-doubling, these eigenstates always come in cat state pairs separated in quasienergy by $\pi$. However, generalizing this concept to DTQC is subtle because of a lack of a single time evolution operator like $U_F$.
 
Finally, let us mention what behavior is expected for measurable observables for an MBL DTQC. We know that for MBL systems, local observables relax to a steady state, generically as a power law in time \cite{Serbyn_1408}. Once this steady state has been achieved, the power spectrum of the time dependence of local observables will display the same behavior as discussed in the prethermal case, by the same arguments. At shorter times before the steady state is achieved, one expects analogous behavior to what is seen for the discrete time crystal in Floquet systems: in addition to ``universal'' peaks in the power spectra of observable that persist to infinite times, one also sees other non-universal peaks that are dependent on the precise disorder realization of the system.

\subsection{How sharply distinct is the DTQC phase?}
\label{subsec:how_distinct}
In this subsection, we wish to elaborate on a point we made earlier regarding the DTQC phase. The DTQC clearly looks like a spontaneous symmetry breaking phase in terms of the extended space picture, as seen for example in Figs.~\ref{fig:tqc_cartoon}(b), \ref{fig:prethermalCartoonZ2Z2}(a), and \ref{fig:prethermalCartoonZ3Z2}(a). However, the extended space picture is a purely formal construction, and what we actually measure are observables as a function of a single time. Here, we wish to be precise about the sense in which the DTQC is distinct from the trivial phase in terms of the single time. We   will also discuss the impact of the fact we can only observe the system over a finite time window, due both to practical experimental limitations and the finiteness of the heating time $t_*$.

The main point, as we have seen, is that the sharp order parameter for the DTQC is the existence of ``subharmonic'' peaks, i.e.~the nonzero   amplitude of a peak in the power spectrum at frequency $\Omega_{\vec{\alpha}} = \vec{\alpha} \cdot \vec{\omega}$ where $\vec\alpha \in \mathcal{L}^*_{SSB}$, which is \emph{not} simply a harmonic of the applied frequencies; that is is, $\Omega_{\vec{\alpha}} \neq \vec{n} \cdot \vec{\omega}$ for any integer vector $\vec{n}$. This is a meaningful statement even though for a quasiperiodic drive the harmonics $\Omega_{\vec{n}}$ are dense. To see this, consider for example the DTQC described in Section \ref{subsec:Z2DTQC}, which has a peak at frequency $\Omega_{\vec{\alpha}}$ for $\vec{\alpha} = (1/2,1/2)$. If $\Omega_{\vec{\alpha}} = \Omega_{\vec{n}}$ for some $\vec{n}$, then this would imply that $(\vec{n} - \vec{\alpha}) \cdot \vec{\omega} = 0$, and hence that $\omega_1/\omega_2 = (n_2 - 1/2)/(n_1 - 1/2)$, which is a rational number, contradicting our assumption that $\omega_1/\omega_2$ is irrational.

Of course, the above argument presupposes that we have infinitely good frequency resolution, corresponding to observing the system over an infinitely long time. Let us instead consider the case where we have some finite frequency resolution $\delta$,  corresponding to a finite observation time $\tau \sim 1/\delta$. Suppose that $|\Omega_{\vec{\alpha}} - \Omega_{\vec{n}}| < \delta$ for some $\Omega_{\vec\alpha}$ not expressible as an integer harmonic of the applied frequencies. 
We would like to estimate how large $|\vec{n}|$ has to be in order for this condition to hold. To achieve this, we assume the frequency vector $\vec\omega$ obeys a so-called Diophantine condition which we will introduce in  Sec.~\ref{sec:heating_estimate}, see \eqnref{eqn:Diophantine},  which quantifies precisely how small $\vec{n}\cdot\vec\omega$ can be as a function of $|\vec{n}|$.
Applying this condition with $\gamma = m-1+\epsilon$ [and replacing, for example,  $\vec{n} \to 2\vec{\alpha} - 2\vec{n}$ if $\vec{\alpha} = (1/2,1/2)$], we find that 
\begin{equation}
|\vec{n}| \geq C (\delta/|\vec{\omega}|)^{-1/(m-1+\epsilon)}
\end{equation}
 for some constant $C$ that does not depend on $\delta$ or on the overall frequency scale. Since the amplitude of peaks decays exponentially with $\vec{n}$ due to the smoothness of the drive, we see that, as we increase our frequency resolution, the amplitude of harmonic peaks that would not be resolvable from the subharmonic DTQC peak goes to zero stretched exponentially fast. Another way to say this is that if we fix the highest-order $n_{\mathrm{max}}$ of harmonic peaks that we want to consider, then the resolution $\delta$ required to distinguish the DTQC phase from the trivial phase is determined by imposing
 \begin{equation}
 \label{eq:phases_distinct}
C (\delta/|\vec{\omega}|)^{-1/(m-1+\epsilon)} \geq n_{\mathrm{max}}.
\end{equation}
Equivalently, the observation time $\tau$ has to satisfy
\begin{align}
\tau \geq |\vec\omega|^{-1} (n_\text{max}/C')^{m-1+\epsilon}
\label{eq:obs_time_req}
\end{align}
for some numerical constant $C'$. Formally, the phases are sharply distinct only when $n_{\mathrm{max}} \to \infty$, which requires $\delta \to 0$. However, the phases are more or less distinct ``in practice'' provided \eqnref{eq:phases_distinct} or \eqnref{eq:obs_time_req} are satisfied for some sufficiently large $n_{\mathrm{max}}$. This is similar to the familiar idea that phases of matter are formally sharply distinct only in the thermodynamic limit, but in practice are distinct as long as the system size is much larger than the correlation length.
Note that the scalings of Eqs.~(\ref{eq:phases_distinct}, \ref{eq:obs_time_req}) get worse as $m$ becomes large; that is, as the number of incommensurate frequencies gets large, we need to observe the system over very long times in order to distinguish phases.
 
Let us observe that \eqnref{eq:phases_distinct} has a very appealing physical interpretation \cite{ThankReferee} in terms of the picture of quasiperiodic driving where time sweeps out a path in the phase torus as shown in Figure \ref{fig:TorusCutProj}(b). After time $\tau$ has elasped, the proportion of the phase torus that is within distance $\Delta$ of the orbit scales like $\sim \Delta^{m-1} |\omega| \tau$. Therefore, if we identify the frequency resolution $\delta$ with the reciprocal of the time $\tau$ over which we observe the system, then \eqnref{eq:phases_distinct} (if we neglect $\epsilon$) is the condition for the fraction of the torus within distance $1/n_{\mathrm{max}}$ of the orbit to be $~\sim 1$. This makes sense, because in order for different quasiperiodic phases to be distinct, the system needs to ``know'' that it is quasiperiodic, that is, it needs to explore the phase torus sufficiently densely.

Finally, let us recall that the heating time $t_*$ sets an upper bound on the observation time $\tau$. Thus, the phases are formally sharply distinct only in the infinite frequency limit where $t_* \to \infty$. Nevertheless, since $t_*$ grows stretched exponentially fast with frequency, we can always increase the frequency while keeping all the other parameters of the problem fixed, in order to have $\tau < t_*$ while still satisfying \eqnref{eq:phases_distinct} and \eqnref{eq:obs_time_req} for $n_{\mathrm{max}} \gg 1$, and so that the phases are distinguishable in practice.

\section{Quasiperiodic topological phases}
\label{sec:TopologicalPhase}

\subsection{Eigenstate classification of quasiperiodic topological phases}
\label{subsec:eigenstate_topological}

In addition to spontaneous symmetry-breaking phases like the DTQC, one can also consider topological phases protected by multiple TTSes. In discussing these topological phases, we will focus on situations where the driving Hamiltonian has  sufficiently strong disorder, so that the preheating Hamiltonian $D$ is MBL (see Sec.~\ref{subsec:long_time_preheating}). 
We remark here that it is not completely settled whether MBL for spatial dimensions greater than one can exist in the strict sense (that is, as an infinite time phenomenon), with some arguments suggesting that, at least for uncorrelated disordered local potentials, it will be destabilized by an ``avalanche'' mechanism \cite{DeRoeck_1608, Luitz_1705, DeRoeck_1705, Gopalakrishnan_1901}. 
However, such a mechanism is highly suppressed for sufficiently strong disorder (compared to other local energy scales of the system), such that it is expected that there is a long timescale $t_{\mathrm{avalanche}}$  \cite{Gopalakrishnan_1901} up to which the effect of avalanches can be neglected and the system exhibits MBL phenomenology. Indeed, such localizing behavior has been verified experimentally in disordered systems in two spatial dimensions \cite{Bordia__17,Choi_1604}. 
In our present case, as the heating time $t_*$ sets a fundamental limit on the physics we describe anyway, we only require that MBL is sufficiently `long-lived', in the sense that $t_{\mathrm{avalanche}}$ is comparable to or greater than $t_*$, which can be straightforwardly achieved with appropriate driving and system parameters. 
%

Since MBL eigenstates have properties analogous to the ground state of gapped local Hamiltonians at all energy densities~\cite{Bauer_1306}, the topological classification of such ground states can be applied to each state of $D$. Static MBL phases can therefore be distinguished by the topology of their eigenstates; this is referred to as eigenstate order~\cite{Huse_1304,Bahri_1307,Chandran_1310}.
In driven systems, the concept of eigenstate order is enriched \cite{vonKeyserlingk_1602_a,Else_1602,Potter_1602,Roy_1602,Roy_1610,Else_PhD} as there are additional topological features arising from the time evolution. Consider the Floquet time-evolution operator \eqnref{eqn:Uapprox} in the pre-heating regime, where we temporarily treat the decomposition as exact.
For an eigenstate $\ket{\Psi}$ of $D$ with energy eigenvalue $\epsilon$, define a time evolution
\begin{equation}
\label{eq:Psi_micromotion}
\ket{\Psi(t)} = e^{i \epsilon t} U(t) P(0) \ket{\Psi} = P(t) \ket{\Psi}.
\end{equation}
We refer to this as the micromotion of the eigenstate $\ket{\Psi}$.
Even when the eigenstate $\ket{\Psi}$ itself describes a topologically trivial static phase, its corresponding micromotion $\ket{\Psi(t)}$ could still be nontrivial. More precisely, let $\Omega_d$ denote the space of all possible gapped ground states of quasilocal Hamiltonians in $d$ spatial dimensions (we mod out by global phase factors in the definition of points in $\Omega_d$). In Floquet systems, the micromotion is periodic and defines a loop in $\Omega_d$ -- we say the micromotion is nontrivial if this loop is not contractible to a point.

We can generalize the notion of topological micromotion \eqnref{eq:Psi_micromotion} to the case of quasiperiodically-driven systems, using the decomposition \eqnref{eqn:Uapprox2}. The micromotion $\ket{\Psi(t)} = P(t)\ket{\Psi}$ is now quasi-periodic, and can be expressed as $\ket{\Psi(t)} = \ket{\Psi(\vec{\omega} t + \vec{\theta}_0)}$, where $\ket{\Psi(\vec{\theta})}$ is parametrized by the torus $\mathbb{T}^m$. This evolution now defines a map on $\mathbb{T}^m \to \Omega_d$. If this map cannot be continuously deformed to the constant map, the evolution is nontrivial. 

The question of how to classify maps $\mathbb{T}^m \to \Omega_d$, or more generally maps $\mathcal{X} \to \Omega_d$ for any space $\mathcal{X}$, remains in principle an open problem. There are, however, good reasons \cite{Turaev_9910, Turaev_0005, Kitaev_0506, Lurie_0905, KitaevIPAM, Thorngren_1612, Xiong_1701, Else_1810} to conjecture that the answer to this question is already contained within the frameworks used to classify   stationary topological phases.
See Appendix \ref{appendix:topological} for a more technical discussion.

Here, we will focus on a special class of phases, which are natural generalizations of bosonic symmetry-protected topological (SPT) phases and already display a rich set of behaviors.
Recall that equilibrium bosonic SPT phases with unitary symmetry $G$ in $d$ spatial dimensions are believed  to be partially classified\footnote{Note that this classification is often stated as $\mathcal{H}^{d+1}(G, \mathrm{U}(1))$, which is equivalent for compact groups provided that the cohomology with $\mathrm{U}(1)$ coefficients is defined appropriately. However, here we are dealing with non-compact groups such as translations and must use $\mathcal{H}^{d+2}(G, \mathbb{Z})$.}
 by the group cohomology $\mathcal{H}^{d+2}(G, \mathbb{Z})$ \cite{Dijkgraaf__90,Chen_1008,Schuch_1010,Chen_1106}. 
We can also write group cohomology as the \emph{singular} cohomology of the so-called classifying space  $BG$ of the group $G$: $\mathcal{H}^{d+2}(G, \mathbb{Z}) \cong H^{d+2}(BG, \mathbb{Z})$. The idea is that maps $\mathcal{X} \to \Omega_d$ in the presence of symmetry $G$ should be partially classified by replacing $BG \to \mathcal{X} \times BG$, i.e.~the classification is $H^{d+2}(\mathcal{X} \times BG, \mathbb{Z})$. In fact, under general conditions for both SPT and symmetry-enriched topological (SET) phases of bosons or fermions, the classification can be derived from the classification of equilibrium SPT and SET phases simply by replacing $BG \to \mathcal{X} \times BG$. We give some justification for this in Appendix \ref{appendix:topological}.

The following powerful statement follows, using the fact that $B\mathbb{Z}^m = \mathbb{T}^m$. We find that the classification of maps $\mathbb{T}^m \to \Omega_d$ in the presence of symmetry $G$ is in one-to-one correspondence with the classification of stationary symmetry-protected and symmetry-enriched phases with symmetry $\mathbb{Z}^m \times G$. The interpretation of the additional $\mathbb{Z}^m$ symmetry is that they correspond to the ``multiple time-translation symmetries'' referred to in Sec.~\ref{sec:twistedTTSes}. We call this the \emph{quasi-periodic equivalence principle}. The periodic case, which we could call the ``Floquet equivalence principle'', was discussed in \cite{Else_1602,Else_PhD}; compare also the ``crystalline equivalence principle'' of Ref.~\cite{Thorngren_1612}.

The simplest case of fundamentally nonequilibrium topological phase are ones where eigenstates $\ket{\Psi}$ are themselves in the trivial $G$ SPT phase. In this case, the classification of maps $\mathbb{T}^m \to \Omega_d$ with a $G$ symmetry imposed has the general decomposition
\begin{equation}
\label{eq:torus_decomposition}
\bigoplus_{r=1}^m \bigoplus_{(k_1 \cdots k_r)} \mathcal{C}_{d-r},
\end{equation}
where $\bigoplus_{(k_1 \cdots k_r)}$ indicates a sum over all possible choice of non-repeating numbers $k_1, \cdots ,k_r \in \{ 1, \cdots, m \}$. Here $\mathcal{C}_s$ is the classification of equilibrium SPT phases or invertible topological orders with $G$ symmetry in $s$ spatial dimensions for $s \geq 0$; for $s < 0$ we set $\mathcal{C}_s = \pi_{-s}(\Omega_0)$. For the case of the group cohomology classification $\mathcal{H}^{d+2}(\mathbb{Z}^{\times m} \times G, \mathbb{Z})$, \eqnref{eq:torus_decomposition} can be proven using the K\"unneth formula \cite{Wen_1301},
although the result holds more generally (see Appendix \ref{appendix:topological}).

The formula \eqnref{eq:torus_decomposition} has a simple physical interpretation: the different terms correspond to cases where $r$ time-translations symmetries are `essentially' involved in the definition of the corresponding phases. One term corresponds to those maps $\mathbb{T}^m \to \Omega_d$ that depend on all $m$ incommensurate frequencies, which are classified by $\mathcal{C}_{d-m}$. The remaining terms of \eqnref{eq:torus_decomposition} depend on only $r < m$ frequencies, corresponding to micromotions which vary on $r$-dimensional hyperplanes $\mathbb{T}^r \subseteq \mathbb{T}^m$ and are deformable to constant evolutions in the other directions. 

For periodic drives ($m=1$), \eqnref{eq:torus_decomposition} reduces to just $\mathcal{C}_{d-1}$. 
In this case, we can think of the non-triviality of the micromotion as a pump per Floquet cycle, which nucleates equilibrium $(d-1)$-dimensional SPT phases and transports them onto the boundary of the system~\cite{vonKeyserlingk_1602_a,Else_1602,Potter_1602,Potter_1610}.
For $m \geq 2$, $\mathcal{C}_{d-m}$ terms in \eqnref{eq:torus_decomposition} can be interpreted as a higher-order pump (`pump of pumps'). However, making this notion concrete for observables on the boundary of model systems is beyond the current discussion.

Let us now focus on the micromotions which depend on all $m$ frequencies.
The ideas of Sec.~\ref{sec:emergentsymmetries} allow us to construct non-trivial micromotions in a preheating regime through what we refer to as a ``bootstrap construction''. Indeed, suppose that we choose $Q = ({1}/{N}) \mathbb{I}_{m }$ for integer $N$, so that $\mathcal{G} = \mathbb{Z}_N^{\times m}$. For concreteness we will consider the case $m=2$. Then the effective Hamiltonian $D$ has an enhanced symmetry $G \times \mathcal{G}$, and accordingly it can host bosonic SPT phases protected by $G \times \mathcal{G}$, which are classified by $\mathcal{H}^{d+2}(G \times \mathcal{G}, \mathbb{Z})$. Invoking the K\"unneth formula, we see that this contains a factor
\begin{equation}
\label{finite_Kunneth}
\mathcal{H}^1(\mathbb{Z}_N, \mathcal{H}^1(\mathbb{Z}_N, \mathcal{H}^{d}(G, \mathbb{Z}))).
\end{equation}
Suppose we ensure that the effective Hamiltonian $D$ has eigenstates which are in an SPT phase associated with the emergent $\mathcal{G} \times G$ symmetry corresponding to an element of \eqnref{finite_Kunneth}.
The projection map $\mathbb{Z} \to \mathbb{Z}_N$ induces a map in group cohomology
\begin{align}
& \mathcal{H}^1(\mathbb{Z}_N, \mathcal{H}^1(\mathbb{Z}_N, \mathcal{H}^{d}(G, \mathbb{Z})) \nonumber \\ &\to \mathcal{H}^1(\mathbb{Z}, \mathcal{H}^1(\mathbb{Z}, \mathcal{H}^{d}(G, \mathbb{Z}))) \nonumber \\ &\cong \mathcal{H}^{d}(G,\mathbb{Z}), \label{ZKunneth}
\end{align}
where we have used the fact that $\mathcal{H}^1(\mathbb{Z}, A) \cong A$ for any finite Abelian group $A$. Note that $\mathcal{H}^d(G,\mathbb{Z})$ obtained in \eqnref{ZKunneth} corresponds to the classification of SPT phases in $d-2$ dimensions, i.e.~$\mathcal{C}_{d-2}$ in the above notation. If the image of the element of \eqnref{finite_Kunneth} in \eqnref{ZKunneth} is nontrivial, it means that the micromotions corresponding to the eigenstates of $D$ are nontrivial in a way that relies on both frequencies.

\subsection{A minimal example of a non-trivial quasiperiodic topological phase}
\label{subsec:minimal}
In this subsection, we will specialize the general and somewhat abstract considerations above. We will consider a particular example of a new quasiperiodic topological phase, which cannot occur in either a stationary  or in a periodically driven system. 

The phase occurs in a two-dimensional spin system driven quasiperiodically with $m=2$ independent frequencies and containing a microscopic symmetry $G = \mathbb{Z}_2$. According to the general classification discussed above, there is a $\mathcal{C}_0 = \mathcal{H}^2(\mathbb{Z}_2, \mathbb{Z}) = \mathcal{H}^1(\mathbb{Z}_2, \mathrm{U}(1)) = \mathbb{Z}_2$ classification of quasiperiodic topological phases that rely on both frequencies. Thus, there is a single nontrivial topological phase in this classification. Below, we will give an explicit Hamiltonian construction to realize this phase, but first discuss its universal properties.

Following the general approach of this paper, we can make this phase well defined in the twisted high-frequency limit. There we obtain an effective static Hamiltonian with an emergent $\mathbb{Z}_2 \times \mathbb{Z}_2$ symmetry (coming from the multiple time-translation symmetries) in addition to the microscopic $\mathbb{Z}_2$. For the phase to be topological and stable, the Hamiltonian must realize an MBL phase in the bulk whose eigenstates are SPT states under the overall $\mathbb{Z}_2^{\times 3}$ symmetry. Such SPT states, however, cannot remain MBL at the boundary while preserving all three $\mathbb{Z}_2$ symmetries. Therefore, a dramatic signature of this topological phase, as with SPT phases more generally \cite{Kane_0506,Pollmann_0909,Levin_1202,Khemani_1508}, is the existence of a non-trivial boundary. Depending on the exact nature of the Hamiltonian on the boundary, either there is a boundary DTQC (which occurs when the emergent $\mathbb{Z}_2 \times \mathbb{Z}_2$ is spontaneously broken), or the microscopic $\mathbb{Z}_2$ symmetry will be spontaneously broken on the boundary, or there is topologically induced boundary delocalization.

\begin{figure}
\includegraphics{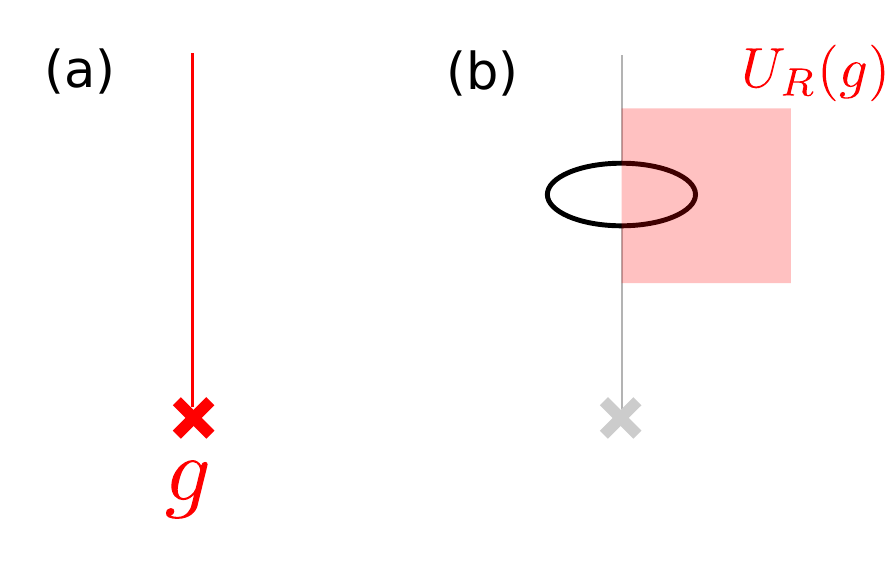}
\caption{\label{symmetryflux}(a) A symmetry flux for a symmetry $g$ consists of a line terminating in a point. (b) Terms of the Hamiltonian (shown as an oval) that cross the line get conjugated by the unitary action, restricted to act only one side of the line [denoted by $U_R(g)$ in the figure].}
\end{figure}

Let us discuss another diagnostic of the nontrivial topology. This diagnostic relates to a ``symmetry twist defect'' in the bulk, which is a standard probe for SPT phases \cite{Barkeshli_1410}. We define such a defect by introducing a line terminating in a point (Figure \ref{symmetryflux}) and conjugating the local terms of the Hamiltonian that straddle the line by the microscopic $\mathbb{Z}_2$ symmetry, restricted to act only on one half of the line.
In the nontrivial SPT-MBL phase we are discussing, the Hilbert space of local states near the termination point carries a projective representation of the emergent $\mathbb{Z}_2 \times \mathbb{Z}_2$ symmetry. From this, combined with the fact that the effective static Hamiltonian $D$ must commute with the emergent $\mathbb{Z}_2 \times \mathbb{Z}_2$, we can deduce the existence of a topologically protected qubit; that is, there exist localized operators $\tau^z$ and $\tau^x$ near the defect such that $\tau_z^2 = \tau_x^2 = 1$, $\tau_z \tau_x =  - \tau_x \tau_z$, and $[\tau^x,D] = [\tau^z,D] = 0$. This is saying there is an effective qubit degree of freedom which does not couple to the rest of the system under time evolution in the rotating frame. In the lab frame, this effective qubit degree of freedom becomes time dependent due to the rotating frame transformation, but it comes back to arbitarily closely even after very long times, due to the recurrences of the quasiperiodic in time rotating frame transformation. Note that this is a stronger condition from just MBL. In MBL there are $l$-bit operators $\tau_i^z$ which preserve the memory of the initial state forever, but normally their conjugate operators $\tau_i^x$ would decohere, unlike what happens here.
We can imagine probing this effect numerically if we assume that $\tau^z$ and $\tau^x$ have some nonzero overlap with some local spin operators $\sigma_i^z$ and $\sigma_i^x$. Then the unequal time correlator $\langle \sigma_i^\alpha(t) \sigma_i^\alpha \rangle$, $\alpha = x,z$, will fail to decay to zero as $t \to \infty$. Finally, we note that we will find similar behavior if, instead of considering a symmetry twist defect, we introduce a boundary for the system, spontaneously break the microscopic $\mathbb{Z}_2$ symmetry on the boundary, and then examine the properties of a domain wall.

It is an important point to consider to what extent the signatures discussed here depend on the presence of an emergent $\mathbb{Z}_2 \times \mathbb{Z}_2$ symmetry. We know that such an emergent symmetry is always present when we stabilize the phase under consideration through the twisted high-frequency limit, which is the main focus of this paper. Nevertheless, one can imagine, as we have alluded to previously, that, in the presence of strong disorder leading to MBL, quasiperiodically driven phases of matter can be stabilized even in a regime not captured by the twisted high-frequency limit. In such a regime, the concept of eigenstate micromotions giving rise to a phase classification, as discussed in generality in the previous subsection, will still apply, but it is not clear if we still expect an emergent $\mathbb{Z}_2 \times \mathbb{Z}_2$ symmetry to be present.
On the other hand, the general discussion of the previous subsection shows that with respect to the classification of phases, we can still treat the system as having a $\mathbb{Z} \times \mathbb{Z}$ symmetry.

One can check that the phase we are discussing, which is originally a $\mathbb{Z}_2 \times \mathbb{Z}_2 \times \mathbb{Z}_2$ SPT (i.e.~protected by the combination of the emergent $\mathbb{Z}_2 \times \mathbb{Z}_2$ and the microscopic $\mathbb{Z}_2$ symmetry) remains nontrivial if we relax the symmetry group to $\mathbb{Z} \times \mathbb{Z} \times \mathbb{Z}_2$. What may be somewhat in doubt, however, is whether the topologically protected qubit associated with the symmetry twist, as discussed above, remains robust, for reasons that we discuss in more detail in Appendix \ref{appendix:relaxing}. However, the property of the boundary not being localizable while preserving the microscopic $\mathbb{Z}_2$ symmetry and the quasiperiodicity of the drive should still hold. Moreover, in Appendix \ref{appendix:relaxing} we also argue that a more robust feature of the symmetry twist defect should be that it hosts a ``topological pump'' of energy between the two incommensurate frequencies (see Ref.~\cite{Martin_1612} and also Section \ref{subsec:relation} below) that is \emph{half-quantized}.

\subsubsection{Example Hamiltonian}
\label{subsec:topological_example}

\begin{figure}[t]
\includegraphics[width=1\columnwidth]{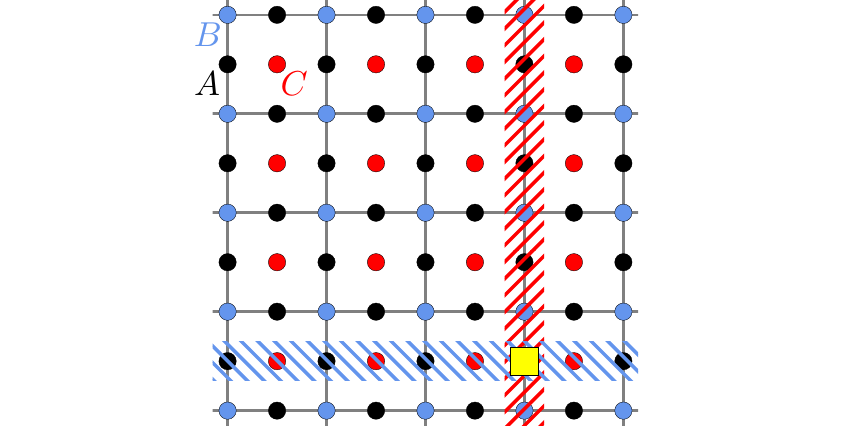}
\caption{\label{fig:squarelattice} Lattice of spins used in the construction of a minimal Hamiltonian displaying the quasiperiodic symmetry protected topological phase introduced in Sec.~\ref{subsec:topological_example}.
Spins are placed on the links ($A$, black),  vertices ($B$, blue), and plaquettes ($C$, red) of a 2d square lattice.  The ground state is a superposition of fluctuating domain walls of Ising orders of the $B$ spins (blue hatched region) and $C$ spins (red hatched region). In the topological phase, the intersection of these domain wall (yellow square) binds an Ising charge of $A$.
} 
\end{figure}

Here we give an explicit construction of a Hamiltonian on the lattice realizing the phase discussed above. Consider a square lattice in two spatial dimensions with spin-1/2 particles on the links ($A$), on the sites ($B$), and on the plaquettes ($C$), as shown in Fig.~\ref{fig:squarelattice}. Assume that the system has a microscopic $\mathbb{Z}_2$ symmetry generated by $X^A := \prod_{l \in A} \sigma_{l}^x$, where the  product is over spins on the links. We then define $S_B = \frac{1}{2}\sum_{v \in B} (\sigma_{v}^x+1)$, where the sum is over all spins on the vertices, and $S_C = \frac{1}{2}\sum_{p\in C} (\sigma_p^x+1)$, where the sum is over spins on the plaquettes. Furthermore, $X^B = \exp(-i\pi S^B) = \prod_{v \in B} \sigma_v^x$ and $X^C = \exp(-i\pi S^C) = \prod_{p \in C} \sigma_p^x$.
We consider the time-quasiperiodic Hamiltonian
\begin{equation}
\label{eq:topexample}
H(t) = f_1(t) S_B + f_2(t) S_C + D_0 + V(t).
\end{equation}
Here $D_0$ is a time-independent Hamiltonian that commutes with $X^A, X^B$ and $X^C$ that we will specify explicitly below and $V(t)$ is some generic perturbation which must commute with the microscopic symmetry $X^A$. All the time-dependent quantities $f_1(t)$, $f_2(t)$ and $V(t)$ are quasiperiodic with frequency vector $\vec{\omega} = (\omega_1, \omega_2)$ assumed to be large compared to the local energy scales of $D_0, V(t)$. 

We now choose the driving functions $f_1(t), f_2(t)$ to be
\begin{align}
f_i(t) = \pi \omega_i \Delta_N(\omega_i t) 
\end{align}
for $i = 1,2$. Here $\Delta_N(\theta)$ is the smooth  approximation to the Dirac Delta comb, as utilized in the DTQC example of Sec.~\ref{subsec:Z2DTQC}, given in \eqnref{eqn:Fejer}. The drive corresponds to one in which $Q = ({1}/{2}) \mathbb{I}_2$ following the discussion of the frame-twisted high frequency limit of Sec.~\ref{sec:twisted_high_freq_QP}.
This choice of driving function strikes a balance between having a smooth drive, necessarily to achieve a long-lived preheating regime, and the property that the interaction Hamiltonian $H_\text{int}(\vec\theta) = U_0(\vec\theta)^\dagger (D + V(\vec\theta)) U_0(\vec\theta)$ (see Sec.~\ref{sec:twisted_high_freq_QP}) has the   term $U_0(\vec\theta)^\dagger D_0 U_0(\vec\theta) $ satisfying, roughly, $U_0(\vec\theta)^\dagger D_0 U_0(\vec\theta) \approx D_0$. 
 Accordingly, upon taking the frame-twisted high-frequency limit, we find that there is an effective Hamiltonian $D = D_0 + \delta D$ in a long-lived preheating regime at high driving frequencies, where $\delta D$ commutes with $X^A$, $X^B$, and $X^C$, and where the local strength of $\delta D$ is on the order of that of $V(t)$.  

Suppose $D_0$ is an MBL Hamiltonian whose eigenstates are in the SPT phase for symmetry $\mathbb{Z}_2^{\times 3}$ corresponding to the non-trivial element of $\mathbb{Z}_2 \cong \mathcal{H}^1(\mathbb{Z}_2, \mathcal{H}^1(\mathbb{Z}_2, \mathcal{H}^2(\mathbb{Z}_2, \mathbb{Z}))) \leq \mathcal{H}^4(\mathbb{Z}_2^{\times 3}, \mathbb{Z})$.
 Indeed, such Hamiltonians can be realized via a decorated domain wall construction \cite{Chen_1303},
 where the ground state is a superposition of fluctuating domain walls. The intersections of $X^B$ domain walls and $X^C$ domain walls (which precisely occur at $A$ spins) carry a $-1$ charge of $X^A$; see Fig.~\ref{fig:squarelattice}. Concretely, we can write $D_0 = \mathcal{U}(\sum_v h_v \sigma_x^v + \sum_p h_p \sigma_x^p + \sum_l h_l \sigma_x^l) \mathcal{U}^{\dagger}$, where the coefficients $h_v$,$h_l$,$h_p$ are chosen from some random distribution, and
\begin{widetext}
\begin{equation}
\mathcal{U} = \sum_{\substack{ \{ \sigma_v \},  \{ \sigma_p  \},\\  \{ \sigma_l  \}  = \pm 1 }} \prod_{   l \in A    } (-1)^{ \frac{1}{8} ( 1 - \sigma_{p_1(l)} \sigma_{p_2(l)} )(1 - \sigma_{v_1(l)} \sigma_{v_2(l)}) (1-\sigma_{l})} \ket{\{ \sigma_v \}, \{ \sigma_p \}, \{ \sigma_l \}}
\bra{\{ \sigma_v \}, \{ \sigma_p \}, \{ \sigma_l \}},
\end{equation}
\end{widetext}
where $p_1(l), p_2(l)$ are the two plaquettes adjacent to the link $l$; $v_1(l), v_2(l)$ are the two vertices connected by the link $l$; and $\ket{ \{ \sigma_v\}, \{ \sigma_p \}, \{ \sigma_l \}}$ is a basis state labelled by the eigenvalues of $\sigma_v^z$, $\sigma_p^z$, $\sigma_l^v$ for all vertices $v$, plaquettes $p$, and links $l$.
Note that if the symmetric term $\delta D$ is sufficiently weak, the effective Hamiltonian $D$ 's eigenstates will also belong to the same SPT phase as the MBL Hamiltonian $D_0$.  

Following the general arguments from before, one finds that the eigenstates of $D$ indeed undergo nontrivial micromotion, classified by the non-trivial element of $\mathcal{H}^1(\mathbb{Z}_2, \mathrm{U}(1)) = \mathbb{Z}_2$, which also classifies SPT phases with $\mathbb{Z}_2$ symmetry in 0 spatial dimensions (i.e.~$\mathbb{Z}_2$ charges).

\subsection{Relation to previous works}
\label{subsec:relation}
Some aspects of topology in quasiperiodically driven systems have previously been studied, primarily in the context  of non-interacting systems. We now briefly mention how these phenomena are related to our classification and construction. 

Ref.~\cite{Martin_1612} discussed ``quantized energy pumping'' between $m=2$ different frequencies in $d=0$ spatial dimensions. Let us interpret the topological invariant that was found. 
In Ref.~\cite{Martin_1612}, the limit of very low frequencies was studied, so that the adiabatic theorem ensured that the system is always in the ground state of the instantaneous Hamiltonian. The latter varies quasiperiodically, $H(t) = H(\vec{\omega} t + \vec{\theta}_0)$ for some continuous function $H(\vec{\theta})$ defined on the 2-torus $\mathbb{T}^2$. Therefore, we can treat the projector onto the ground state of $H(\vec{\theta})$ as a function of $\vec{\theta}$, which defines a micromotion that can be classified according to the general framework discussed above.
We find that
\begin{equation}
H^2(\mathbb{T}^2, \mathbb{Z}) = \mathcal{H}^2(\mathbb{Z} \times \mathbb{Z}, \mathbb{Z}) = \mathbb{Z}.
\end{equation}
This $\mathbb{Z}$ invariant is the Chern number over the torus $\mathbb{T}^2$ of the Berry connection of the ground state.
Note that the adiabatic limit considered is very different from the high-frequency one we considered above -- in particular, there cannot be a decomposition of the time-evolution operator of the form \eqnref{eqn:Uapprox2} here \cite{Jauslin_1991}. However, this does not affect the classification.
The topological invariants discussed in Ref.~\cite{Crowley_1808} can also be understood in a similar fashion.

One might think that the nontrivial invariant of this kind could also be realized in the frame-twisted high-frequency limit through a bootstrap construction as in Sec.~\ref{subsec:topological_example}. However, this is not the case. The reason is that in the frame-twisted high-frequency limit, the emergent symmetry is always a finite group $\mathcal{G} = \mathbb{Z}^{\times 2}/\mathcal{L}'$. One can check that although $H^2(\mathbb{Z} \times \mathbb{Z}, \mathbb{Z}) = \mathbb{Z}$, the image in this group under the map $\mathcal{H}^2(\mathcal{G}, \mathbb{Z}) \to \mathcal{H}^2(\mathbb{Z} \times \mathbb{Z}, \mathbb{Z})$ induced by the projection map $\mathbb{Z} \times \mathbb{Z} \to \mathcal{G}$, is always trivial for any finite quotient $\mathcal{G}$.

In Ref.~\cite{Peng_1805}, ``Majorana multiplexing'' was introduced, where a system in one spatial dimension which is quasiperiodically driven with $m$ incommensurate frequencies, may host $2^m$ different kinds of boundary Majorana zero modes that can occur simultaneously, while being protected from coupling to each other. This is consistent with  our general framework, since the Majorana modes are distinguished by the charge ($\pm 1$) they carry under each of the $m$ generators of the $\mathbb{Z}^{\times m}$ ``time-translation symmetry'' group. One can readily show that any combination of such boundary Majorana zero modes can be realized in a frame-twisted high-frequency limit in a bootstrap construction of the general form described above. In particular, this demonstrates that the ``Majorana multiplexing'' of Ref.~\cite{Peng_1805} can be made stable to interactions, at least up to the parametrically long heating time $t_*$, in the presence of strong disorder and driving at high frequencies. 

\section{Estimating the heating time in quasiperiodically-driven systems}
\label{sec:heating_estimate}

We return to the important question of dynamics and slow heating in quasiperiodically-driven systems which underpin the existence of the nonequilibrium phases of matter discussed in the preceding sections.
We first give schematic arguments within linear response that can be used to understand  the stretched exponential scaling of the heating time $t_*$ of \eqnref{eqn:heating_t}, by extending the discussion of Sec.~\ref{subsec:longlifetimes}. This analysis follows that of Ref.~\cite{Abanin_1507}. The scaling of $t_*$ that we obtain through these arguments will be borne out in the rigorous bounds on heating of Sec.~\ref{sec:proof}. 

\subsection{Linear response arguments}

As discussed in Sec.~\ref{subsec:longlifetimes}, the heating rate in a quasiperiodically driven system is governed by the competition between the decay of Fourier series coefficients of the driving term $V_{\vec{n}}$ with $|\vec{n}|$, and the fact that $| \vec{\omega} \cdot \vec{n}|/|\vec{\omega}|$ could become ever smaller as $|\vec{n}|$ increases. In order to estimate the heating time, we need to quantify exactly how small $|\vec{\omega} \cdot \vec{n}|$ can get as a function of $|\vec{n}|$. The key mathematical tool is the following Diophantine condition
\begin{align}
\frac{|\vec\omega \cdot \vec n|}{|\vec\omega|} \geq \frac{c}{| \vec n|^{\gamma}},
\label{eqn:Diophantine}
\end{align}
for all integer vectors $\vec{n} \neq \vec{0}$, where $c$ is a constant depending on the ratios of the frequencies $\omega_i/\omega_j$, but not on the overall frequency scale $|\omega|$. 
It can be rigorously shown (see Appendix \ref{sec:rarity_resonances}) that all choices of frequency vectors except a set of measure zero obey such a condition for any $\gamma > m-1$ (with a constant $c$ depending on $\gamma$). 
In particular, \eqnref{eqn:Diophantine} holds (again, for any $\gamma > m-1$, and with $c$ depending on $\gamma$) when the ratios $\omega_i/\omega_j$ are all irrational algebraic numbers -- that is, they are each roots of some polynomial equation with integer coefficients, in which case it is known as the  \emph{subspace theorem}~\cite{Schmidt__72,SchmidtBook}.

In our analysis we always assume that the frequency vectors $\vec\omega$ obey \eqnref{eqn:Diophantine}.
This will allow us to derive lower bounds on the scaling of the heating time at high frequencies (recall that this high frequency limit corresponds to taking $\omega := |\vec{\omega}| \to \infty$ while keeping the ratios $\omega_i/\omega_j$ fixed).
For the periodically-driven case, we can use \eqnref{eqn:Diophantine} with $\gamma = 0$.

In linear response, a term $V_{\vec{n}}$ can drive transitions between energy levels of the average Hamiltonian separated by energy $\Delta E_{\vec{n}} = \vec{\omega} \cdot \vec{n}$. However, such processes require a rearrangement of at least $\sim \Delta E_{\vec n}/J$ sites, where $J$ is the local strength of the average Hamiltonian, and hence the amplitudes for such processes are suppressed by a factor  
\begin{equation}
\label{eq:quasiperiodic_factor}
e^{- \kappa (\Delta E_{\vec n} /J)} \leq e^{-\kappa (\omega/J) |\vec{n}|^{-\gamma}}
\end{equation}
for some constant $\kappa > 0$, using the Diophantine condition \eqnref{eqn:Diophantine}.
 Since only $|\vec{n}| \neq 0$ processes contribute to heating, the smallest value of $|\vec{n}|$ is $1$, in which case this term evaluates to $e^{-\kappa (\omega/J)}$, an exponentially small factor at high frequencies. For $\gamma = 0$ (the periodic case), this is the whole story since there is no $\vec{n}$ dependence on the right-hand side of \eqnref{eq:quasiperiodic_factor}, and we recover the usual exponential scaling of heating time with frequency. For quasiperiodically driven systems, on the other hand, $\gamma > 0$ and then the right-hand side of \eqnref{eq:quasiperiodic_factor} goes to $1$ at large $|\vec{n}|$. This would suggest that the heating rate is not suppressed even at high frequencies.
 
To derive stronger results, the key is to use the smoothness of the driving. 
If we assume the local strength of $ V_{\vec{n}} $ is bounded by some function $g(|\vec{n}|)$, then the rate of a heating process governed by a given $V_{\vec{n}}$ is controlled by both \eqnref{eq:quasiperiodic_factor} and $g(|\vec{n}|)$, which gives a suppression factor
 \begin{equation}
 \label{eq:corrected_scaling}
 g(|\vec{n}|)^2 e^{-\kappa (\omega/J) |\vec{n}|^{-\gamma}}.
 \end{equation}
 To obtain the full heating rate, we sum over all such processes. Asymptotically, it is governed by the fastest heating rate:  For a given form of $g(|\vec{n}|)$, there will be a value of $|\vec{n}|$ that maximizes \eqnref{eq:corrected_scaling}, and the heating time then scales like the inverse of the maximum value of \eqnref{eq:corrected_scaling}. 
  If $g(|\vec{n}|)$ decays exponentially in $|\vec{n}|$, which is the case we consider throughout our paper, then we find $t_* \sim \exp[{\text{const.} (\omega/J)^{1/(\gamma+1)}}]$, which is a similar scaling to the bound in \eqnref{eqn:heating_t}. One sees that unlike the Floquet case, for quasiperiodic driving ($\gamma > 0$), it is crucial to use the smoothness of the driving. Indeed, if $g(|\vec{n}|)$ does not decay with $|\vec{n}|$, we see that the heating time does not grow with frequency at all if $\gamma > 0$. Treating other forms of drives, where $g(|\vec{n}|)$ has different asymptotic forms at large $\vert \vec{n} \vert$, is an interesting direction to develop; see Sec.~\ref{subsec:nonsmooth}.

\subsection{Stability to varying the frequencies}

Our results have assumed keeping the ratios $\omega_i/\omega_j$ constant, and only changing the overall frequency scale $\vert \vec\omega \vert$ to reach the high-frequency limit. 
In practice, the physical time evolution of the system should not depend too sensitively on the precise ratios $\omega_i/\omega_j$ and associated constant $c$. We expect there to be a long timescale before the system can resolve the distinction between a drive with driving frequencies $\vec\omega$ and one with nearby driving frequencies $\vec\omega + \delta \vec\omega$ (for example, a close rational approximation, for which our results technically do not apply).

We can estimate this timescale  at the level of the linear response arguments. Consider a vector $\vec{\omega}$ which satisfies the conditions for our results to hold, \eqnref{eqn:Diophantine}, 
and consider perturbing the frequency $\vec{\omega} \to \vec{\omega}' = \vec{\omega} + \delta \vec{\omega}$. By rescaling we can ensure that $|\vec{\omega}'| = |\vec{\omega}|$.
From the triangle inequality,
\begin{align}
\frac{|\vec{\omega}' \cdot \vec{n}| }{|\vec{\omega}|}  \geq c|\vec{n}|^{-\gamma} - \frac{|\delta \vec{\omega}| |\vec{n}| }{|\vec{\omega}|} 
\geq \frac{c}{2} |\vec{n}|^{-\gamma}, 
\end{align}
provided that
\begin{equation}
\label{eq:n_condn}
|\vec{n}| \leq n_{\mathrm{max}} = \left( \frac{c|\vec{\omega}|}{2|\delta \vec{\omega}|}  \right)^{{1}/{(\gamma+1)}}.
\end{equation}
That is, the new frequency vector still satisfies the required approximation condition (with a constant $c$ that is independent of $\delta \vec{\omega}$), but only when $|\vec{n}| \leq n_{\mathrm{max}}$.

In this context, the effect of processes with $|\vec{n}| \geq n_{\mathrm{max}}$ will never be felt if we truncate the Fourier expansion $V_{\vec{n}}$ of the driving so that $V_{\vec{n}} = 0$ for $|\vec{n}| > n_{\mathrm{max}}$. Assuming smooth driving, so that $\| V_{\vec{n}} \|$ decays exponentially with $|\vec{n}|$, then this truncation will not substantially affect the dynamics until a time
\begin{equation} 
t_{\mathrm{perturb}} \sim e^{Cn_{\mathrm{max}}} = \exp\left(C \left[\frac{c|\vec{\omega}|}{2|\delta \vec{\omega}|}\right]^{1/(\gamma+1)}\right),
\end{equation}
for some constant $C$.  After $t_{\mathrm{perturb}}$, the evolution of the system will be governed by a different regime.

Therefore, the worst that could possibly happen when perturbing the frequency vector is that the heating time $t_*$ will grow with frequency until it reaches $t_{\mathrm{perturb}}$ and then stop growing because the system immediately heats at time $t_{\mathrm{perturb}}$. However, this heating scenario is not inevitable, because the dynamics after $t_{\mathrm{perturb}}$ could themselves have a much longer heating time $t_*'$ that grows with frequency. For example, if $\vec{\omega}'$ is a rationally related frequency vector, then the system is really periodic and one can invoke Floquet slow heating results.

 \section{Long-lived, preheating regime  in quasiperiodically-driven systems: rigorous results and sketch of proof} 
 \label{sec:proof}
 
 In this section we finally formalize the preceding discussions on slow heating and emergent symmetries, into  a rigorous theorem on preheating in quasiperiodically-driven systems. We do this by explicitly constructing the effective Hamiltonian $D$ and providing bounds on its validity.
The construction employed also manifestly allows for emergent symmetries to appear in $D$, should the driving Hamiltonian have twisted-TTSes as described in Sec.~\ref{sec:twistedTTSes}.
 The exact formulation of the proof and heating bounds is somewhat involved and we leave the technical details to  Appendix \ref{appendix:proof}. Here we shall provide an accessible statement of the theorem, as well as an outline of the proof.

\subsection{Conditions and setup}
 
 Recall that we are considering quasiperiodically driven systems defined on   lattices in $d$-spatial dimensions with locally bounded Hilbert spaces, i.e.~spins or fermions, with a driving frequency vector $\vec\omega$ that is rationally independent, $\vec\omega \cdot \vec{n} \neq \vec{0}$ for any nonzero integer vector $\vec{n} \in \mathbb{Z}^m$. We furthermore assume that the time-quasiperiodic Hamiltonian is quasilocal, with the drive performed at high-frequencies and is furthermore smooth in time.
What we mean by the high-frequency condition is that each driving frequency $\omega_i$ is large compared to the local energy scales $J$ of the Hamiltonian $H(\vec\theta)$, 
and what we mean by the smoothness of drive condition is the imposition of the condition that the Fourier modes $H_{\vec{n}}$ decay exponentially fast with $|\vec{n}|$, i.e.~$\|H_{\vec{n}}\| = O(e^{- \kappa' |\vec n|})$ for some $\kappa' > 0$.
 Additionally, let us remind the reader that we will assume that the frequency vector $\vec\omega$ obeys the Diophantine condition \eqnref{eqn:Diophantine}, which holds for all choices of rationally-independent frequency vectors except for a set of measure zero.

 More precisely, our construction and theorem makes use of a notion of a local norm $\|O\|_{\kappa}$ parameterized by a constant $\kappa > 0$, appropriate for a many-body operator $O(\vec\theta)$ parameterized on the torus $\mathbb{T}^m$ and which acts on an infinite lattice with bounded local Hilbert space dimension.
To define this norm, let us first write (non-uniquely) a many-body operator $O(\vec\theta)$ in terms of a sum of local `potentials' $O_{Z}(\vec\theta)$ on $\mathbb{T}^m$, where $Z$ is a finite subset of the lattice  and where $O_{Z}(\vec\theta)$ only acts non-trivially on sites $x \in Z$ for all $\vec\theta \in \mathbb{T}^m$.
We can furthermore decompose the local potential into its Fourier modes $O_Z(\vec\theta) = \sum_{\vec n} O_{Z,\vec{n}\in \mathbb{Z}^m} e^{i \vec{n} \cdot \vec{\theta}}$ so that 
\begin{align}
O(\vec\theta) = \sum_{Z} O_{Z}(\vec\theta) = \sum_{Z,\vec{n}} O_{Z,\vec{n}} e^{i \vec{n} \cdot \vec\theta}.
\label{eqn:O}
\end{align}
We then define  $\| O \|_{\kappa}$ to be 
\begin{equation}
\label{eq:colored_potential_norm}
\| O \|_{\kappa} = \sup_x \sum_{\vec{n},{Z} \ni x} e^{\kappa (|Z| + |\vec{n}|)} \| O_{Z,\vec{n}} \|,
\end{equation}
for some constant $\kappa > 0$, where the supremum is over sites $x$ on the lattice, and the norm appearing in the sum is the standard operator norm.

The norm  measures the strength of local terms making up $O(\vec\theta)$, specifically taking into account the decay of the strength of interactions in both spatial extent and Fourier space. Terms corresponding to larger spatial support and higher Fourier modes are weighted more, parameterized by $\kappa$ which can be understood as the decay constant.
Note that the norm is only really useful if it is finite. Thus, the local terms appearing in the sum in \eqnref{eqn:O} have to be decaying at least exponentially fast both in real and Fourier space. The quasilocal and time-smooth nature of the driving $H(\vec\theta)$ we imposed ensures that there is some $\kappa > 0$ for which  $\| H \|_{\kappa} < \infty$. 

\subsection{Theorem and statement of results: long-lived, preheating description of dynamics and emergent symmetries}
\label{subsec:thmstatement}

With this setup, we   turn to the object of interest: the unitary time-evolution operator 
\begin{align}
U(t) = \mathcal{T} \exp\left(-i \int_{0}^{t} H(t') dt' \right)
\end{align}
 generated by a  quasiperiodically-driven   Hamiltonian $H(t) = H(\vec\omega t + \vec\theta_0)$ for some $\vec\theta_0 \in \mathbb{T}^m$. We assume there is some decay constant $\kappa_0 > 0$ such that $H(\vec\theta)$ has a local norm $J := \| H \|_{\kappa_0}$ which is small enough compared to the driving frequency: $J \leq K \omega$. Here $K$ is some numerical constant that does not depend on the Hamiltonian, geometry of the lattice or driving frequencies, which we give explicitly in Appendix \ref{appendix:proof}.

We then have the following statements.\\
 {\bf (A) Existence of a long-lived preheating dynamical description.} It is possible to find a decomposition of the unitary time-evolution operator as:
\begin{align}
U(t) = P(t) \mathcal{T} \exp \left( -i \int_0^t \left[D + V(t') \right] dt' \right) P^\dagger(0), 
\label{eqn:Uproof}
\end{align}
where $P(t) = P(\vec{\omega} t + \vec{\theta}_0)$ is a time-quasiperiodic quasilocal unitary,
$D$ is a time-independent, quasilocal Hamiltonian, and $V(t) = V(\vec{\omega} t + \vec{\theta}_0)$ is a time-quasiperiodic, quasilocal Hamiltonian. The functions $P(\vec{\theta}), V(\vec{\theta})$ are smooth on $\mathbb{T}^m$. 

Define the decay constant $\kappa = \kappa_0/4$. The Hamiltonian $D$ is close to the time-averaged Hamiltonian 
\begin{align}
\langle H \rangle = \int_{\mathbb{T}^m} \frac{d^m\vec\theta}{(2\pi)^m} H(\vec\theta),
\label{eqn:D}
\end{align}
in the sense $\|D - \langle {H} \rangle \|_{\kappa} \leq C J \left(\frac{J}{\omega} \right)$.
The term $V(\vec\theta$) is small, in the sense that
\begin{align}
\|V\|_{\kappa} \leq J 2^{-q_*}.
\label{eqn:V}
\end{align}
where
\begin{align}
{q_*} = \Big\lfloor K' \left(\frac{\omega}{J} \right)^{\frac{1}{\gamma+1} }  \Big\rfloor.
\label{eqn:nstar}
\end{align}
Here $C, K'$  are numerical constants which we compute in Appendix \ref{appendix:proof}, that importantly do not depend on the Hamiltonian or driving frequencies.

We now spell out the dynamical consequences.
\begin{itemize}
\item Slow heating. 
The time-averaged Hamiltonian $\langle H \rangle$ is an almost conserved energy operator, up to perturbative corrections of order $J/\omega$, captured precisely by
\begin{align}
\frac{1}{\text{Vol.}} \| U^\dagger(t) &\langle H \rangle U(t) - \langle H \rangle \| 
  \leq J \left( \tilde{K} 2^{-q_*} t + C \left( J/\omega \right) \right)
\label{eqn:conserved_energy2}
\end{align}
for all $t \geq 0$. In the above, the standard operator norm is used, and we have divided by the volume ``$\text{Vol.}$'' of the system. In other words, the normalized energy density as measured by $\langle H \rangle/J$ grows very slowly, apart from a small quantity of order $J/\omega$. Here $ \tilde{K}, C$ are finite constants that do not depend on the Hamiltonian, geometry of the lattice, or driving frequencies. An analogous statement holds replacing $\langle H \rangle \mapsto D$ as they only differ by small terms on the order of $J/\omega$.

From this we can thus  derive that the heating time $t_*$, defined as the time beyond which  the quantity $\frac{1}{\text{Vol.}} \| U^\dagger(t) \langle H \rangle U(t) - \langle H \rangle \|$ starts growing appreciably, obeys a lower bound 
\begin{align}
t_* \geq \tilde{K}^{-1} 2^{q_*} \equiv \frac{C'}{J} \exp\left[C \left(\frac{\omega}{J}\right)^{1/(m + \epsilon )}\right],
\end{align}
for appropriately defined numerical constants $C',C$.
Note that this is only a lower bound for $t_*$, because \eqnref{eqn:conserved_energy2} is only an upper bound on how fast the energy density can change.

\item Effective   description of dynamics.  If we define
\begin{equation}
\hat{U}(t) = P(t) e^{-iDt} P^{\dagger}(0),
\end{equation}
then for any local operator $O$, we have that
\begin{align}
\label{eq:heisenberg_difference}
& \| \hat{U}^{\dagger}(t) O \hat{U}(t) - U^{\dagger}(t) O U(t) \| \nonumber \\
& \leq K''(O) 2^{-{q_*}} Jt(1 + Jt)^d,
\end{align}
where $K''(O)$ is a numerical constant not depending on the Hamiltonian or driving frequencies.
Thus, the difference \eqnref{eq:heisenberg_difference} is very small for times less than the heating time  $t_* \sim 2^{{q_*}}$.

\end{itemize}

 {\bf (B) Existence of emergent symmetries in effective Hamilton $D$.} 
If in addition the original driving Hamiltonian $H(\vec\theta)$ has   twisted-TTSes  generated by ${g}_{\vec{\tilde\tau}}$ (for some set of $\vec{\tilde\tau}$s), that is,
\begin{align}
H(\vec\theta + \vec{\tilde\tau}) = g_{\vec{\tilde\tau}} H(\vec\theta) g_{\vec{\tilde\tau}}^\dagger,
\end{align}
 as described in Sec.~\ref{sec:twistedTTSes},  then $V(\vec\theta), P(\vec\theta)$ in the above decomposition of the unitary also have the same twisted-TTSes. The effective Hamiltonian $D$ obeys
\begin{align}
[D, {g}_{\vec{\tilde\tau}}] = 0.
\end{align}

The above statements form the rigorous basis on which the results of  this paper rests on. They show that (i) the system does not heat until the long time $t_*$, that (ii) in the preheating regime the decomposition \eqnref{eqn:Uapprox2} holds up to small corrections that can be ignored until $t_*$, in the sense made precise by Eqs.~(\ref{eqn:conserved_energy2}, \ref{eq:heisenberg_difference}), and that (iii) the effective static Hamiltonian $D$ possesses  emergent symmetries as a consequence of   twisted-TTSes of the drive.

\subsection{Main ideas and sketch of proof}
\label{subsec:proofsketch}
Let us provide here the ideas underlying our technical procedure that allows us to derive the previous assertions, as well as sketch   the proof.
To obtain the dynamical statement (A), we employ an iterative procedure `renormalizing' the initial driving Hamiltonian $H(\vec\theta)$ through a series of small rotations that sequentially reduce the norm of time-dependent pieces. This is possible under conditions of high frequencies where the parameter $J/\omega$, the ratio of the local energy scales to the driving frequency $\omega$, sets a natural small parameter.
 Upon stopping at some optimal order, this will eventually give us  the effective static Hamiltonian $D$ as well as a remnant small time-dependent piece $V(\vec\theta)$. Note that it is expected that the procedure generically cannot be carried out \emph{ad infinitum} as this would imply that a driven interacting system has a static local Hamiltonian description. This would go against the unbounded heating we expect to occur at long times.

Such a logic is behind the rigorous prethermalization works of Refs.~\cite{Abanin_1509, Abanin_1510}, and indeed our manipulations largely follow closely that of Ref.~\cite{Abanin_1509} but with a number of technical extensions to handle the quasiperiodically-driven scenario.
Our main contribution, however, and the biggest departure from the earlier works, is that we will employ a renormalization procedure specifically tailored for preserving a twisted-TTS at all stages, which then allows us to obtain our statement (B) on emergent symmetries in $D$.  
 
\emph{Setting up the iteration.} --- The renormalization process is nothing but a sequence of well-chosen rotating frame transformations, effected by the unitaries $P^{(q)}(\vec\theta) = P^{(q)}(\vec\theta + \vec\tau)$ where $\vec\tau \in \mathcal{L} \in 2\pi \mathbb{Z}^m$ and $q = 0,1,\cdots$ up to some cut-off $q_*$ to be determined.
We start the process by defining the original Hamiltonian as the zeroth-level renormalized Hamiltonian $H^{(0)}(\vec\theta) \equiv H(\vec\theta)$. 
The rotation $P^{(q)}(\vec\theta)$ defines the $q$-th level renormalized Hamiltonian at the next level,
\begin{align*}
H^{(q)}(\vec\theta) \mapsto H^{(q+1)}(\vec\theta) := P^{(q), \dagger} (\vec\theta)(H^{(q)}(\vec\theta) - i \vec\omega \cdot \partial_{\vec\theta} ) P^{(q)}(\vec\theta).
\end{align*}
Here $H^{(q)}(\vec{\theta})$ is the resulting Hamiltonian in the new frame of reference.
Let us also define the unitaries
\begin{align}
U^{(q)}(t) := \mathcal{T} \exp\left( -\int_0^t H^{(q)}(\vec\omega t + \vec\theta_0) \right),
\end{align}
in particular, $U^{(0)}(t) = U(t)$, the unitary time-evolution operator of interest. 
In terms of dynamics, the iterative procedure is
\begin{align}
U^{(q)}(t) = P^{(q) }(\vec\omega t + \vec\theta_0) U^{(q+1)}(t) P^{(q),\dagger}(\vec\theta_0).
\end{align}

As mentioned, the aim is to reduce the time-dependent terms in $H^{(q)}(\vec\theta)$ at each level. 
To that end, let us  define the time-averaging operation
\begin{align}
\langle H^{(q)} \rangle := \int_{\mathbb{T}^m} \frac{ d^m \vec\theta}{(2\pi)^m}  H^{(q)}(\vec\theta) = H^{(q)}_{\vec 0},
\end{align}
so that we can write $H^{(q)}(\vec\theta)  = D^{(q)} + V^{(q)}(\vec\theta)$, where
\begin{align}
D^{(q)} &= \langle H^{(q)} \rangle, \nonumber \\
V^{(q)}(\vec\theta) &= H^{(q)}(\vec\theta) - D^{(q)}.
\end{align}
We see that we need to eliminate, or at least reduce, the contributions of $V^{(q)}(\vec\theta)$ in the next level Hamiltonian $H^{(q+1)}(\vec\theta)$.

The high-frequency assumption allows us to choose appropriate rotations $P^{(q)}(\vec\theta)$, that are close to the identify (`small').
To gain some intuition, we write $P^{(q)}(\vec\theta) = e^{A^{(q)}(\vec\theta)}$ for some antihermitian operator $A^{(q)}(\vec\theta)$ assumed to be smaller than $V^{(q)}$ by a factor of $1/\omega$. Then, using the Duhamel formula to expand $H^{(q+1)}(\vec\theta)$, we have
\begin{align}
H^{(q+1)}(\vec\theta) & = D^{(q)} + V^{(q)}(\vec\theta) - i \vec\omega \cdot \partial_{\vec\theta} A^{(q)}(\vec\theta) \nonumber \\
 & - [A^{(q)}(\vec\theta), H^{(q)}(\vec\theta)] - \frac{i}{2} [A^{(q)}(\vec\theta), \vec\omega \cdot \partial_{\vec\theta} A^{(q)}(\vec\theta) ] \nonumber \\
&  + \cdots,
\label{eqn:BCH}
\end{align}
and we see schematically that all terms beyond the first line are smaller than $V^{(q)}$ by a factor of $J/\omega$ or less. 

Let us therefore demand that in our iterative procedure, the generators $A^{(q)}(\vec\theta)$ are chosen to satisfy
\begin{align}
V^{(q)}(\vec\theta) - i (\vec\omega \cdot \partial_{\vec\theta}) A^{(q)}(\vec\theta) = 0. 
\label{eqn:VA}
\end{align}
Note that there is a freedom of choice in the solution of this partial differential equation, as the initial condition has not  been specified.

So far, the manipulations have been formally identical to that of Ref.~\cite{Abanin_1509} upon reduction to the Floquet case.
 However (and here is the crucial technical difference), we  solve the above equation for $A^{(q)}(\vec\theta) = \sum_{\vec n} A^{(q)}_{\vec n} e^{i \vec{n} \cdot \vec\theta}$  in terms of Fourier modes of $V^{(q)}(\vec\theta) = \sum_{\vec n} V^{(q)}_{\vec n} e^{i \vec{n} \cdot \vec\theta}$, with specific choice
\begin{align}
A^{(q)}_{\vec n} = \begin{cases}
-\frac{1}{\vec \omega \cdot \vec n} V^{(q)}_{\vec n} , \qquad & \vec n \neq \vec 0 \\
0, \qquad & \vec n = \vec 0
\end{cases}.
\label{eqn:A}
\end{align}
This fixes the initial condition by setting $\langle A^{(q)}\rangle = 0$. By contrast, the immediate generalization of Ref.~\cite{Abanin_1509} would be to set $A^{(q)}(\vec{\theta}_0) = 0$. Our different choice of initial condition is what allows our iterative procedure to preserve twisted-TTSes, as we will see later.

\eqnref{eqn:A} highlights the fact that the solution only makes sense should the Fourier modes $V^{(q)}_{\vec n}$ decay fast enough so that $A^{(q)}_{\vec n}$ can be written as a convergent Fourier series and $\|A\|_{\kappa}$ is at least well-defined for some values of $\kappa$. This is the technical reason for the imposition of the `smoothness' of drive conditions. 

As a consistency check, one can see from \eqnref{eqn:A} that  $A^{(q)}$,  loosely speaking, differs from $V^{(q)}$ by the `small' factor of $\sim 1/\omega$. More precisely, the relative size of $A^{(q)}$ to that of $V^{(q)}$ should be measured by their local norms as given by Eq.~\eqref{eq:colored_potential_norm}, with slightly different decay constants $\kappa$ (see Appendix \ref{appendix:proof} for   details).
Note that the choice \eqnref{eqn:A} also preserves the quasilocality: if $V^{(q)}(\vec\theta)$ is quasilocal, then so will be $A^{(q)}(\vec\theta)$.

\emph{Estimating the optimal order ${q_*}$}. --- Having set up the iteration, let us now present the logic behind bounding how far this iterative procedure can be carried out to. 
While satisfying the relation \eqnref{eqn:VA} makes $H^{(q)}(\vec\theta)$ ever less dependent on $\vec{\theta}$, there is a price to pay: the extra terms generated at each level in the renormalized Hamiltonians are of ever longer range.
To account for this and in order to meaningfully estimate their local strength, we allow for some sequence of strictly decreasing decay constants $\kappa_0 > \kappa_{1} > \kappa_{2} \cdots > 0$ which we have to pick judiciously, and measure the Hamiltonian $H^{(q)}(\vec\theta)$ at the $q$-level through its norm $\| H^{(q)}(\vec\theta)\|_{\kappa_q}$.
However, at some stage, the smallness of $\kappa_q$ will impede our ability to bound the Hamiltonian $H^{(q+1)}(\vec\theta)$ at the next level; this is when the renormalization procedure stops.
The aim is to choose a suitable set of decay constants $\kappa_q$, which allows for the inductive process to be carried out to as high an order as possible, rendering the resulting $V^{({q_*})}(\vec\theta)$ optimally small and giving the effective Hamiltonian $D := D^{({q_*})}$ at the stopping order $q_*$.

Our proof  indeed follows such a procedure. We skip the heavily technical details in setting up various inductive bounds, as well as choosing the appropriate decay constants $\kappa_q$, but merely state that we end up with the optimal level of truncation ${q_*}$ as given in Eq.~\eqref{eqn:nstar}.
This also yields the claimed bounds on $D,V(\vec\theta)$,  defined as the optimal $D^{({q_*})}, V^{({q_*})}(\vec\theta)$, respectively.
The result \eqref{eqn:Uproof} then follows, with $P(\vec\theta) := \prod_{l=1}^{{q_*}} P^{(l)}(\vec\theta)$.
We refer the reader to the Appendix \ref{appendix:proof} for the full details of our manipulations.

\subsection{Emergent symmetries in effective Hamiltonian $D$}

Thus far, the above discussion was purely dynamical, as related to statement (A). 
Now let us consider the statement (B) on the emergent symmetries in $D$.

Our choice of solutions Eq.~(\ref{eqn:A}) explicitly preserves the twisted-TTSes of the original driving Hamiltonian $H(\vec\theta)$ at every level of the renormalization procedure. 
This in turn allows for them to be manifested as unitary operators that commute with the effective Hamiltonian $D$.
To see this, recall that a  twisted-TTS acts in Fourier space on a Hamiltonian according to  ${g}_{\vec{\tilde\tau}} {H}_{\vec n} {g}_{\vec{\tilde\tau}}^{\dagger} = e^{i \vec{n} \cdot \vec{\tilde\tau}} {H}_{\vec n}$, so in particular $\langle H \rangle = H_{\vec 0} = {g}_{\vec{\tilde\tau}} H_{\vec 0} {g}_{\vec{\tilde\tau}}^\dagger$ is symmetric.
Now if the $q$-th level Hamiltonian $H^{(q)}(\vec\theta)$ has a twisted-TTS, then $V^{(q)}(\vec\theta)$ has the same symmetry, while $D^{(q)} = \langle H^{(q)} \rangle$ commutes with ${g}_{\vec{\tilde\tau}}$.
But from Eq.~\eqref{eqn:A}, this immediately implies that that $A^{(q)}(\vec\theta)$ and hence $P^{(q)}(\vec\theta)$ will also inherit the twisted-TTS of $V^{(q)}(\vec\theta)$, and therefore so does $H^{(q+1)}(\vec\theta)$.
As this is true for every $q$, we end up with the statement that the effective Hamiltonian $D$ obtained at the optimal order ${q_*}$ has emergent symmetries ${g}_{\vec{\tilde\tau}}$. 

\section{Extensions \& Future Directions}  
\label{sec:extensions}

Having demonstrated how one can achieve novel nonequilibrium phases of matter in quasiperiodically-driven systems, we now consider extensions and future directions arising from our work. Many of these are directly inspired by the recent development in nonequilibrium Floquet systems.

\subsection{Long-range interactions}
\label{subsec:longrange}

One of the assumptions which allowed us to bound heating rates and establish the pre-heating regime, was that the many-body systems we considered had quasilocal interactions -- that is, the amplitudes of the interaction terms in the Hamiltonian decay at least exponentially fast with space. 
The study of prethermalization in Floquet systems suggests that this restriction may be lifted to encompass long range interactions. Refs.~\cite{Tran_1908, Machado_1908}  have demonstrated slow heating for periodic driving in the presence of two-body power-law interactions, provided that the interactions decay with distance as $\sim 1/r^{\alpha}$, with $\alpha > d$ where $d$ is the spatial dimension, see also Ref.~\cite{Ho_1706}. The existence of an effective Hamiltonian approximately generating the dynamics in the preheating regime, however, can only be proven for $\alpha > 2d$. We can immediately combine our proof with the approach of \cite{Machado_1908} to derive similar results of slow heating and emergent symmetries in quasiperiodically driven systems with power-law interactions. These extensions are particularly valuable for realizing quasiperiodic phases in trapped ion systems, which naturally have long range interactions.

\subsection{Time-independent systems: continuous time quasi-crystal}

Although we have focused on applications to quasiperiodically-driven systems, our theorem also has an intriguing implication for systems with time-independent Hamiltonians. In particular, let $\Gamma_1, \cdots, \Gamma_m$ be a set of commuting operators, each of which has integer eigenvalues, and let $\vec{\omega} = (\omega_1, \cdots, \omega_m)$ be some rationally-independent numbers. Consider the time-independent Hamiltonian
\begin{equation}
H = H_0 + V,
\end{equation}
where
\begin{equation}
\label{eq:static_H0}
H_0 = \sum_i \omega_i \Gamma_i.
\end{equation}
Observe that the time evolution generated by $H_0$ is quasiperiodic $U_0(t) = \exp(-iH_0 t) = \exp(-it \omega_i \Gamma_i)$, since $U_0(t) = U_0(\vec{\omega} t)$, where $U_0(\vec{\theta}) = \exp(-i\theta_i \Gamma_i)$. If we define the Hamiltonian in the rotating frame generated by $H_0$, i.e.
\begin{align}
H_{\mathrm{int}}(t) = U_0^{\dagger}(t) V U_0(t),
\label{eqn:interactingH2}
\end{align}
then it also has quasi-periodic time dependence $H_{\mathrm{int}}(t) = H_{\mathrm{int}}(\vec{\omega} t)$, with
\begin{equation}
H_{\mathrm{int}}(\vec{\theta}) = U_0^{\dagger}(\vec{\theta}) V U_0(\vec{\theta}).
\end{equation}

We invoke our theorem to construct a rotating frame transformation $P(\vec{\theta})$ that generates a time-independent quasi-local effective Hamiltonian $D$, up to corrections that can be ignored until a time $t_*$ that scales like a stretched exponential in $|\vec{\omega}|$. This might not seem very useful, since the original Hamiltonian $H$ was already a time-independent Hamiltonian. However, crucially,  $H_{\mathrm{int}}$ has a twisted TTS in the sense discussed in Sec.~\ref{sec:twistedTTSes}. Indeed, we have that
\begin{equation}
H_{\mathrm{int}}(\vec{\theta} + \vec{\tau}) = g_{\vec{\tau}} H_{\mathrm{int}}(\vec{\theta}) g_{\vec{\tau}}^\dagger, \quad g_{\vec{\tau}} = \exp(i\tau_i \Gamma_i) = U_0^{\dagger}(\vec{\tau}),
\end{equation}
but now for \emph{any} vector $\vec{\tau} \in \mathbb{R}^m$ (thus, we have a continuous twisted TTS rather than a discrete twisted TTS as considered previously). Therefore, our theorem ensures (see \eqnref{eqn:Uproof}) that the time-evolution operator generated by $H_{\mathrm{int}}(t)$ can be decomposed (ignoring corrections that only become important at times $t \gtrsim t_*$) as
\begin{equation}
\label{eq:Udecompcontinuous}
U_{\mathrm{int}}(t) \approx P(t) \exp(-iDt) P^{\dagger}(0),
\end{equation}
where $D$ has the emergent symmetries $[D,g_{\vec{\tau}}] = 0$ for all vectors $\vec{\tau}$, and hence (by taking the limit $\vec{\tau} \to 0)$ $[D,\Gamma_i] = 0$ for all $i$. Moreover, $P(t) = P(\vec{\omega} t)$, where $P(\vec{\theta})$ obeys the twisted TTS, i.e.~$P(\vec{\theta} + \vec{\tau}) = g_{\vec{\tau}} P(\vec{\theta}) g_{\vec{\tau}}^\dagger$. Hence, using the form of $g_{\vec{\tau}}$, we find that $P(\vec{\theta}) = U_0^{\dagger}(\vec{\theta}) P U_0(\vec{\theta})$, where $P := P(0)$.
 Substituting into \eqnref{eq:Udecompcontinuous}, we find
\begin{align}
U_{\mathrm{int}}(t) &\approx U_0^{\dagger}(t) P U_0(t) \exp(-iDt) P^{\dagger} \\
&= U_0^{\dagger}(t) P \exp(-i[D + H_0]t) P^{\dagger}.
\end{align}

On the other hand, we also know that, by definition, $U_{\mathrm{int}}(t) = U_0^{\dagger}(t) e^{-iHt}$. Therefore, we find that
\begin{equation}
\label{eq:continuousHemergent}
H \approx P(D + H_0)P^{\dagger}.
\end{equation}
The right-hand side of \eqnref{eq:continuousHemergent} commutes with $P \Gamma_i P^{\dagger}$ for $i = 1, \cdots, m$. That is, the Hamiltonian $H$ has an approximate (because it only holds up to time $t_*$) emergent $\mathrm{U}(1)^{\times m}$ symmetry. The $m=1$ case was already proven in Ref.~\cite{Else_1607} (though the connection with twisted TTS was not identified).

Let us discuss three applications of this result. Firstly, we can imagine that the emergent $\mathrm{U}(1)^{\times m}$ symmetry is spontaneously broken. In that case, we obtain a continuous time quasicrystal -- a time-independent Hamiltonian which spontaneously develops a quasiperiodic response with frequencies $\omega_1, \cdots, \omega_m$ until the long time $t_*$. This is a generalization of the prethermal continuous time crystal discussed in Ref.~\cite{Else_1607}.

Another application is the topological protection of quantum information. Ref.~\cite{Else_1704} used the $m=1$ version of this result to argue that the decoherence time for a qubit encoded in a Majorana zero mode on the boundary of a one-dimensional topological superconductor could be made exponentially long (in a parameter of the Hamiltonian) at arbitrary energy density, not just in the ground state, even without disorder.  Ref.~\cite{Else_1704} also considered a two dimensional planar code, but they were not able to show long lifetime in the limit of large system size, because the $m=1$ result does not allow for the separate conservation of the number of $e$-type and $m$-type excitations separately. With our new $m > 1$ result, we can now ensure that these numbers are separately conserved, which does lead to a (stretched) exponentially long lifetime for the encoded qubit, even as the system size goes to infinity.

Finally, let us discuss an application to a lattice model of charged particles in a strong electric field.
Consider charged fermions hopping on a $d$-dimensional lattice, and suppose we set $\Gamma_i = \sum_{\vec{x}} x_i c_{\vec{x}}^{\dagger} c_{\vec{x}}$ for $i = 1, \cdots, d$, where the sum is over the positions $\vec{x}$ of the lattice sites (which we take to be integer vectors), $x_i$ is the $i$-th coordinate of $\vec{x}$, and $c_{\vec{x}}^{\dagger}$ is the fermion creation operator at position $\vec{x}$. Then the Hamiltonian $H_0$ of \eqnref{eq:static_H0} corresponds to applying an electric field $\vec{E} = \vec{\omega}$ to the system. If $|\vec{E}|$ is much larger than the other local scales in the Hamiltonian, then we can apply our theorem, and we find that there are $i$ emergent symmetries that are conserved up to time $t_*$; we can interpret these as dressed versions of the components of electric dipole moment. The fact that these dipole moments are conserved prevents particles from hopping from one site to another, which we can think of as a many-body, higher-dimensional version of Stark localization \cite{vanNieuwenburg_1808,Schulz_1808,Khemani_1910}.
The ability to engineer such conserved quantities in a robust manner can lead to a variety of interesting dynamical phenomena, including the recently discussed idea of ``Hilbert space shattering'' of Refs.~\cite{Sala_1904, Khemani_1904,Khemani_1910,Rakovszky_1910}, where the Hilbert space fractures into exponentially many subsectors.

\subsection{Non-smooth drives}
\label{subsec:nonsmooth}

In this paper, we have focused on drives that are smooth in time. It would be interesting to further develop our discussions for cases where the Fourier modes of the drive $H_{\vec{n}}$ decay slower than exponentially with $\vert\vec{n}\vert$. For example, discrete step-drives decay as some power-law with $|\vec n|$. Within the simple linear response estimates of Sec.~\ref{sec:heating_estimate}, it seems natural to expect a heating time  $t_*$ that scales like a power-law with frequency. However, a more careful calculation of the heating processes along the lines of Sec.~\ref{sec:proof} would be valuable. 

Ref.~\cite{Dumitrescu_1708} considered driving an MBL system with a discrete Fibonacci step-drive and found two slow heating regimes, {one governed by an effective Hamiltonian generated from a high-frequency expansion whose description lasts for power law times, and a subsequent one where there is a slow logarithmic decay of observables that lasts up to an exponentially long timescale.}
Some of the phenomenology considered is very different from the one discussed in Sec.~\ref{subsec:long_time_preheating}; establishing the relationship to the work here remains an open question. 

On a technical level, we note that extending our results of Sec.~\ref{sec:proof} to the case of power-law decay of Fourier modes is more complex than applying them to the case of power law decay of interactions with distance. Considering our results and those of \cite{Machado_1908}, we see that the analog of power-law dependence of Fourier modes are interactions that decay as a power-law in the number of sites contained within the support of the interaction. By contrast, Ref.~\cite{Machado_1908} established results only for interactions which decay as a power-law in the diameter of the support, while still decaying exponentially in the number of sites.

\subsection{Quasiperiocally-driven topological phases beyond eigenstate micromotion}

In Sec.~\ref{sec:TopologicalPhase}, we focused on eigenstate micromotions as a diagnostic of nontrivial topological phases in periodically or quasiperiodically-driven MBL systems. However, it is known that there are topological Floquet phases which cannot be diagnosed in this way, many of which have particularly fascinating phenomenology. For example, the phases discussed in Refs.~\cite{Po_1609,Harper_1609,Po_1701,Fidkowski_1703}  are characterized by a chiral pumping of quantum information along the one-dimensional boundary of a two-dimensional bulk. One expects analogous phenomena in quasiperiodically-driven systems, but we leave such developments for future work.

\subsection{Experimental realizations}
\label{subsec:expt_realization}
Lastly let us briefly mention the possibility of experimentally realizing the quasiperiodically-driven phases of matter we have discussed, in particular, the discrete time quasi-crystal (DTQC) phase described in Sec.~\ref{sec:tqc}.

The driving protocol and conditions required, described in \ref{sec:twisted_high_freq_QP}, can easily be realized in setups such as in synthetic quantum systems of trapped ions, or in solid state systems like nitrogen-vacancy (NV) defects in diamond. Both cases give rise to ensembles of coherently interacting effective spin degrees of freedom which are well isolated from the environment.
Appropriate sequences of laser or microwave pulses can be engineered to realize particular drives in these systems. 
Indeed, these were utilized to realize sequences~\cite{Else_1603, Yao_1608} that led to signatures of discrete time crystal phases  observed in experiments of a 1d chain of trapped ions~\cite{Zhang_1609} and in a dense, disordered 3d ensemble of NV centers in diamond~\cite{Choi_1610}.

Extending these protocols to realize time-quasiperiodic phases is relatively straightforward, although the time-dependence of the drives needs to be smooth in time to avoid fast heating. The different platforms offer different comparative advantages. In trapped ions, the ability to tune the range of long-range interactions between spins implies that it should be possible to realize a DTQC in a prethermal regime even in 1d, provided that the power law exponent satisfies $1 < \alpha < 2$.  On the other hand, one can utilize the fact that NV centers in diamond are grouped by their orientation with respect to the crystallographic axes of diamond. Using interactions within and between multiple groups, one can naturally realize more complicated TQCs, such as the example Hamiltonian in Appendix~\ref{subsec:Z3Z2DTQC} giving rise to a $\mathbb{Z}_2 \times \mathbb{Z}_2$ DTQC, which utilizes two collection of spins. 
Note that if in 3d, the dipolar interacting nature of the spins precludes localization, but thermalization is nonetheless slow in heavily disordered samples (precisely, critically slow, see~\cite{Kucsko_1609, Ho_1703, Choi_1806}). Thus, a long-lived version of a DTQC might still be realizable, protected by critically slow thermalization dynamics.
  
The dynamical signatures that one would look out for, would be the appearance of peaks at fractional harmonics of the input drive frequencies in the power spectrum of a local observable. Additionally, these signatures should be robust to small perturbations to the drive protocol which still preserve its time-quasiperiodicity. This would signal the spontaneous breaking of multiple time-translation symmetries  of the original drive and is the defining characteristic of the DTQC phase (see Sec.~\ref{sec:tqc}).
The exact pattern that is manifested in the Fourier harmonics, however, depends on the exact symmetry breaking pattern, realized using different driving protocols. 
As there are multiple time-translation symmetries in quasiperiodic drives, there are myriad complex symmetry breaking patterns, which could be observed even for a given experimental platform.

\section{Conclusion}
\label{sec:conclusion}

In this paper, we have shown how interacting quantum many-body systems that are quasiperiodically-driven can realize a panoply of long-lived nonequilibrium phases of matter under a general set of driving conditions. These phases are protected by  multiple time-translation symmetries, which arise as the driving Hamiltonian derives from a function $H(\vec\theta)$ defined on a higher dimensional torus. 
They are fundamentally different from phases realizable in static or Floquet systems and add to the richness of the landscape of possible phases of matter realizable in nonequilibrium settings. 
As an exemplar, we described the phase of matter obtained by the spontaneous breaking of some or all of the multiple time-translational symmetries --- a time quasi-crystal. 
We also gave a classification of the symmetry-protected topological phases of matter achieved in the quasiperiodic setting.

Key to our results was our ability to identify a dynamical regime, a so-called preheating regime, in which the deleterious effects of heating in driven systems was controlled for a parametrically long time, as well as our analysis of how multiple-TTSes play out in this regime.
We provided a class of general driving conditions that realize prethermal or quasiperiodic many-body localized phases in this regime. 
We emphasize that their dynamical signatures are universal and robust against small perturbations to the driving protocol, as long as they respect the time-quasiperiodicity of the drive: they are thus genuine nonequilibrium phases of matter in interacting many-body quantum systems.

Our results open up exciting new lines of experimental and theoretical research.
Indeed, the nonequilibrium phases of matter we have discussed, in particular, time quasi-crystals, are immediately directly accessible in experimental platforms of today.
Theoretically, our work establishes new universal structure present in quantum many-body systems in highly out-of-equilibrium settings.  What other possible fundamentally nonequilibrium phases of matter remain to be discovered? 

\begin{acknowledgments}
W.W.H.\ would like to thank S.~Choi, D.~Sels and P.J.D.~Crowley for helpful discussions. D.V.E.\ thanks V.~Khemani and J.~Lin for helpful discussions. We thank M.\ Serbyn, A.C.\ Potter, and R.M.~Nandkishore for helpful comments on our manuscript.
This work was initiated at the Kavli Institute for Theoretical Physics, which is supported by the National Science Foundation under Grant No.\ NSF PHY-1748958. 
The Flatiron Institute is a division of the Simons Foundation.
W.W.H.\ and D.V.E.\ were each supported by the Gordon and Betty Moore Foundation’s EPiQS Initiative, Grant No.\ GBMF4306 (W.W.H.) and GBMF4303 (D.V.E.). W.W.H.~acknowledges support from the NUS Development Grant AY2019/2020.
\end{acknowledgments}
\appendix

\section{Generalized Floquet-Magnus expansion for quasiperiodically-driven systems}
\label{appendix:FloquetMagnus}
Here we provide the derivation of a generalized Floquet-Magnus expansion for quasiperiodically-driven systems driven at high frequencies, introduced and used in the main text in Sec.~\ref{sec:preheating_dynamics} for illustrative purposes on how twisted time-translation symmetries enter into the static Hamiltonian $D$ governing effectively the dynamics of the system.

We note two aspects of the following discussion:
1) This expansion is \emph{not} the naive direct generalization of the Floquet-Magnus expansion to quasiperiodically-driven systems, a point we shall highlight in the derivation.
2) We do not prove any bounds on the validity of the expansion. More precisely, what this means is as follows. As with   the usual Magnus expansion in Floquet systems, we do not expect the series that will write down to be convergent for a many-body system -- instead, it is an asymptotic series. For the Floquet case, it is possible to analyze the optimal order to which the series should be truncated to \cite{Kuwahara_1508,Mori_1509, Abanin_1510}, which gives an optimal effective Hamiltonian $D$ and an estimate on the lifetime that the system can be viewed as having dynamics under $D$. However, we do not give such an estimate, but merely generate the formal series. That said, the analysis of this series  does suggest an interesting alternative and complementary derivation of our results, which we reserve for future work.
 
The logic behind the generalized Floquet-Magnus expansion, or for that matter, the usual Floquet-Magnus expansion in periodically-driven systems, simply involves moving into a suitable rotating frame of reference and deriving the rotating frame Hamiltonian. At high frequencies $\omega = |\vec \omega|$, it is possible to organize this transformation as a perturbative series in $1/\omega$, and use the freedom endowed by the frame transformation to cancel time-dependent pieces in the rotating frame Hamiltonian, order by order. This is also the same logic underlying our technical procedure that yields a different effective Hamiltonian $D$; thus, the following expansion should be viewed as a different derivation of an effective description of the system.

To start, we consider a quasiperiodically-driven Hamiltonian $H(t) = H(\vec \omega t + \vec\theta_0)$, where $H(\vec\theta)$ is a periodic Hamiltonian on the torus. The unitary time-evolution operator obeys the following equation
\begin{align}
i \partial_t U(t) = H(t) U(t).
\end{align}
We move into a new frame of reference effected by the unitary $P(t) = P(\vec \omega t + \vec\theta_0)$, for some periodic $P(\vec\theta)$ on the torus. In other words, we decompose the time evolution operator as $U(t) = P(t) \tilde{U}(t) P(0)^\dagger$, where $\tilde{U}(t)$ obeys
\begin{align}
i \partial_t \tilde{U}(t) & = \left( P^\dagger(t) H(t) P(t) - i P^\dagger(t) \partial_t  P(t) \right) \tilde{U}(t) \nonumber \\
& \equiv {D}(t) \tilde{U}(t).
\end{align}
We see that time-evolution is now generated by the Hamiltonian ${D}(t)$, which   is time-quasiperiodic as ${D}(t) = {D}(\vec\omega t + \vec\theta_0)$ for a periodic ${D}(\vec\theta)$ which is given explicitly as 
\begin{align}
{D}(\vec\theta) =  P^\dagger(\vec\theta) H(\vec\theta) P(\vec\theta) - i P^\dagger(\vec\theta) \vec{\omega} \cdot \partial_{\vec \theta}  P(\vec\theta).
\end{align}
We will show how $P(\vec\theta)$ can be chosen to make $D(t)$ static, i.e.~$D(t) = D$.

To that end we write $P(\vec\theta)$ as the exponential of a sum of antihermitian operators,
\begin{align}
 P(\vec\theta) = \exp\left( \Omega(\vec\theta) \right), \qquad \Omega(\vec\theta) =  \sum_{q = 1}^\infty \Omega_q(\vec\theta),
\end{align}
where we have organized $\Omega$ as a series with terms labeled by $q$. $q$ will turn out to track the order of the high frequency expansion. Now we can write, using Duhamel's formula,
\begin{align}
{D}(\vec\theta) = e^{-\text{ad}_{\Omega(\vec\theta)}} H(\vec\theta) - i \frac{1 - e^{- \text{ad}_{\Omega(\vec\theta)} }}{\text{ad}_{\Omega(\vec\theta)}} \vec\omega \cdot \partial_{\vec \theta}  {\Omega(\vec\theta)},
\end{align}
where $\text{ad}_A B = [A,B] $. Expanding, we get
\begin{align}
& {D}(\vec\theta) = \sum_{q=1}^\infty {D}^{(q)}(\vec\theta), \nonumber \\
&  {D}^{(q)}(\vec\theta) = G^{(q)}(\vec\theta) - i \vec\omega \cdot \partial_{\vec\theta} \Omega_{q+1}(\vec\theta), \nonumber \\
& G^{(q)}(\vec\theta)  = \sum_{k=1}^q \frac{(-1)^k}{k!} \sum_{\stackrel{1 \leq i_1, \cdots, i_k \leq q}{i_1 + \cdots + i_k = q}} \text{ad}_{\Omega_{i_1}} \cdots \text{ad}_{\Omega_{i_k}} H(\vec\theta) \nonumber \\
& + i \sum_{m=1}^q \sum_{k=1}^{q+1-m} \frac{(-1)^{k+1}}{(k+1)!}   \nonumber \\
& \sum_{\stackrel{1 \leq i_1, \cdots, i_k \leq q+1-m}{i_1 + \cdots + i_k = q+1-m}} \text{ad}_{\Omega_{i_1}} \cdots \text{ad}_{\Omega_{i_k}} \vec\omega \cdot \partial_{\vec \theta} \Omega_m(\vec\theta).
\label{eqn:big_expansion}
\end{align}
 
The question now is how best to reduce the $\theta$-dependent pieces of ${D}(\vec\theta)$. One natural way is to assume that $\Omega_q(\vec\theta)$ has norm $\sim 1/\omega^q$, i.e.~treat it as a high frequency expansion. Therefore, we need to pick the generators $\Omega_{q+1}(\vec\theta)$ in such a way that this assumption holds at every step. Indeed, if we write $G^{(q)}(\vec\theta)$ in terms of a Fourier series
\begin{align}
G^{(q)}(\vec\theta) = G^{(q)}_{\vec 0} + \sum_{\vec n \neq \vec 0} G^{(q)}_{\vec n} e^{i \vec{n} \cdot \vec\theta},
\end{align}
then we see that by imposing the condition
\begin{align}
\sum_{\vec n \neq \vec 0} G^{(q)}_{\vec n} e^{i \vec{n} \cdot \vec{\theta}} - i \vec\omega \cdot \partial_{\vec \theta} \Omega_{q+1}(\vec\theta) = 0 
\label{eqn:FM_condition}
\end{align}
that     the norm of $\Omega_{q+1}$ is smaller than $\Omega_q$ by $1/\omega$, since   $\|\Omega_{q+1}\| \sim \|G^{(q)}\|/|\vec\omega|$ and $\|G^{(q)}\| \sim 1/\omega^q$, schematically.

The imposition of  \eqnref{eqn:FM_condition} thus defines a family of expansions, differing by a choice of initial condition.
Indeed, the immediate, direct generalization of the Floquet-Magnus expansion to quasiperiodically-driven systems corresponds to one particular choice:
One solves \eqnref{eqn:FM_condition} with the condition that $\Omega_{q}(\vec\omega t + \vec\theta_0)|_{t=0} = 0$, so that $P(t)|_{t=0} = P(\vec\omega t + \vec\theta_0)|_{t= 0} = \mathbb{I}$.  This gives the solution
\begin{align}
\Omega_{q+1}(\vec\theta) = 
- \sum_{\vec n \neq \vec 0} \frac{G_{\vec n}^{(q)}}{\vec{n} \cdot \vec{\omega} } \left( e^{i \vec{n} \cdot \vec \theta} - e^{i \vec{n} \cdot \vec \theta_0} \right).
\label{eqn:MagnusSoln}
\end{align}
To see why this is the `standard' expansion, for a Floquet system we get  from the above $P(0) = P(T) = \mathbb{I}$ and so at stroboscopic times $t = NT$ for $N \in \mathbb{Z}$ we have the familiar relation $U(NT) = e^{-i D N T}$, ignoring questions of convergence.

Instead, we will choose to solve \eqnref{eqn:FM_condition} as
\begin{align}
\Omega_{q+1}(\vec\theta) = 
- \sum_{\vec n \neq \vec 0} \frac{G_{\vec n}^{(q)}}{\vec{n} \cdot \vec{\omega} } e^{i \vec{n} \cdot \vec \theta}.
\label{eqn:OurSoln}
\end{align}
This   uniquely defines our `generalized Floquet-Magnus' expansion.
\\

To see what terms emerge in our expansion, consider $q = 0$. We get from \eqnref{eqn:big_expansion} and \eqnref{eqn:FM_condition}
\begin{align}
D^{(0)}(\vec\theta) & = H_{\vec 0} + \sum_{\vec n \neq \vec 0} H_{\vec n} e^{i \vec{n} \cdot \vec\theta} - i \vec{\omega} \cdot \partial_{\vec \theta} \Omega_1 (\vec\theta) \nonumber \\
& = H_{\vec 0}
\end{align}
since $\Omega_1(\vec\theta)$ was chosen to eliminate $\theta$-dependent pieces, given explicitly by
$
\Omega_1(\vec\theta) = - \sum_{\vec n \neq \vec 0} \frac{H_{\vec n}^{(q)}}{\vec{n} \cdot \vec{\omega} } e^{i \vec{n} \cdot \vec \theta}.
$
For $q = 1$, we get
\begin{align}
D^{(1)}(\vec\theta) & = - \text{ad}_{\Omega_1} H(\vec\theta) + \frac{i}{2} \text{ad}_{\Omega_1} \vec\omega \cdot \partial_{\vec \theta} \Omega_1(\vec\theta) -  i \vec{\omega} \cdot \partial_{\vec \theta} \Omega_2 (\vec\theta) \nonumber \\
& = \frac{1}{2} \sum_{\vec n} \frac{ [H_{\vec n}, H_{-\vec n}] }{\vec n \cdot \vec \omega} 
\end{align}
since $\Omega_2(\vec\theta)$ was similarly chosen to eliminate the $\theta$-dependent pieces.  In similar fashion, $D^{(q)}(\vec\theta)$ for any $q$ can be derived.
\\

Lastly, let us explain how our choice of solution \eqnref{eqn:OurSoln} giving our generalized Floquet-Magnus expansion manifestly allows for twisted time-translation symmetries to be preserved:  it is such that if the original driving Hamiltonian $H(\vec\theta)$ had a twisted time translational symmetry, that is, if there exists a unitary operator $g$ and integer $N$ satisfying $g^N = \mathbb{I}$, as well as some vector $\vec{\tilde \tau}$ on the torus such that $H(\vec\theta + \vec{\tilde \tau}) = g H(\vec\theta) g^\dagger$,  then $D(\vec\theta)$ also possesses the same twisted time-translation symmetry.
To see this, recall that on the Fourier modes a twisted TTS acts as $gH_{\vec n} g^\dagger = e^{i \vec{n} \cdot \vec{\tilde\tau}} H_{\vec n}$.
Now, if we assume $G^{(q)}(\vec\theta)$ has a twisted-TTS for some $q$, then our choice of $\Omega_{q+1}(\vec\theta)$, \eqnref{eqn:OurSoln}, will also inherit the same twisted-TTS (one just sees that the Fourier modes of $\Omega_{q+1}(\vec\theta)$ obey the required transformation). Consequently, $G^{(q+1)}(\vec\theta)$, given by sums of conjugations of $H(\vec\theta)$ with various $\Omega_{k}(\vec\theta)$ with $k < q$, etc.~(see \eqnref{eqn:big_expansion}) will also have the same twisted-TTS. Note the standard Floquet-Magnus expansion corresponding to a choice \eqnref{eqn:MagnusSoln},  does not have this property.

\section{Periodicity of $U_0(\vec\theta)$ and related details}
\label{appendix:gi_details}

Here we will prove the claim we made in Sec.~\ref{sec:twisted_high_freq_QP}, namely if we write $H_0(t) = H_0(\vec{\theta}_0 + \vec{\omega t})$, $H_0(\vec{\theta}) = f_i(\vec{\theta}) \Gamma_i$, where $\overline{f}_i = Q_{ij} \omega_j$, then it follows that the time evolution operator $U_0(t)$ generated by $H_0(t)$ can be written of the form $U_0(t) = U_0(\vec{\theta}_0 + \vec{\omega} t)$, where

\begin{equation}
U_0(\vec{\theta}) = \exp(-i h_i(\vec{\theta}) \Gamma_i),
\end{equation}
for some functions $h_i(\vec{\theta})$ to be determined, which will satisfy certain symmetry properties. Firstly, we observe that this amounts to the statement that
\begin{equation}
U_0^{\dagger}(\vec{\theta}) (\vec{\omega} \cdot \partial_{\vec{\theta}}) U_0(\vec{\theta}) = -i H_0(\vec{\theta}),
\end{equation}
which in turn is equivalent to
\begin{equation}
\label{eq:gi_eqn}
(\vec{\omega} \cdot \partial_{\vec{\theta}}) h_i(\vec{\theta}) = f_i(\vec{\theta}),
\end{equation}
If we define $\eta_i(\vec{\theta}) = f_i - \overline{f}_i$ and $\Delta_i(\vec{\theta}) = h_{i}(\vec{\theta}) - Q_{ij} \theta_j$, then \eqnref{eq:gi_eqn} reduces to
\begin{equation}
(\vec{\omega} \cdot \partial_{\vec{\theta}}) \Delta_i(\vec{\theta}) = \eta_i(\vec{\theta}),
\end{equation}
The solution to this equation can be obtained in terms of the Fourier transform; if we write $\eta_i(\vec{\theta}) = \sum_{\vec{n} \neq 0} e^{i\vec{n} \cdot \vec{\theta}} \eta_i(\vec{n})$ and similarly for $\Delta_i$, then we can set $\Delta_i(\vec{n}) = (i\vec{n} \cdot \vec{\omega})^{-1} \eta_i(\vec{n})$ for $\vec{n} \neq 0$. In conclusion, we have found that
\begin{equation}
h_i(\vec{\theta}) = Q_{ij} \theta_j + \Delta_i(\vec{\theta}),
\end{equation}
where $\Delta_i(\vec{\theta})$ satisfies $\Delta_i(\vec{\theta} + \vec{\tau}) = \Delta_i(\vec{\theta})$ for all $\vec{\tau} \in \mathcal{L}$. These functions $h_i(\vec{\theta})$ indeed satisfy the symmetry property claimed in the main text.

\section{Properties of rational matrices}
\label{appendix:rational}
Let $Q$ be some $r \times m$ matrix with rational entries.
Let $\mathcal{L}$ be the lattice comprising all vectors of the form $2\pi \vec{n}$ for some integer vector $\vec{n}$. We are interested in determining the sublattice $\mathcal{L}' \leq \mathcal{L}$ comprising those vectors $\vec{\tau}' \in \mathcal{L}$ such that $e^{i Q_{ij} \tau_j'} = 1$.

To solve this problem, we invoke the well-known fact that any integer matrix $Z$ has a Smith normal form decomposition $Z = V D W$ where $V$ and $W$ are unimodular integer matrices (that is, they are invertible integer matrices whose inverses are also integer matrices) and $D$ is a diagonal integer matrix. Since any rational matrix can be converted to an integer matrix simply by multiplying by some integer, we conclude that there is also a Smith normal form for rational matrices: $Q = V D W$ where $V$ and $W$ are still unimodular integer matrices, and $D$ is a diagonal rational matrix.

The Smith decomposition allows us to reduce the problem to the case where $Q$ is diagonal. Suppose that if $Q = D = \mathrm{diag}(p_1/q_1, \cdots, p_{\mathrm{min}(r,m)}/q_{\mathrm{min}(r,m)})$, where each $p_i$,$q_i$ are coprime integers (with $q_i$ positive and $p_i$ non-negative). We can choose to set $q_i = 1$ if $p_i = 0$, and we also define $q_i = 1$ for $r < i \leq m$. Then, we see that $\mathcal{L}'$ comprises integer linear combination of the $m$ vectors spanned by $\vec{\tau}^{(i)} = q_i \vec{e}_i$ (no summation), for $(\vec{e}^{(i)})_j = \delta_{ij}$. More generally, if $Q = V D W$ then $\mathcal{L}'$ comprises integer linear combinations of the vectors $W^{-1} \vec{\tau}^{(i)}$. We also find that $\mathcal{L} / 
\mathcal{L}' \cong \mathbb{Z}_{q_1} \times \cdots \times \mathbb{Z}_{q_m}$.

Note that a particularly simple case is where $r=m$ and
$Q = Z^{-1}$ for some (not necessarily unimodular) integer matrix $Z$. In that case, one finds that all the $p_i$'s are $1$, the $q_i$'s are the diagonal entries of the Smith decomposition of $Z$, and $\mathcal{L}'$ simply comprises integer linear combinations of the columns of $Z$.

\section{More examples of DTQC}
\label{appendix:DTQC}

We provide here the  quasiperiodically-driven Hamiltonians realizing the $\mathbb{Z}_2\times \mathbb{Z}_2$ and $\mathbb{Z}_3 \times \mathbb{Z}_2$ DTQCs shown in Figs.~\ref{fig:prethermalCartoonZ2Z2}, \ref{fig:prethermalCartoonZ3Z2} of Sec.~\ref{sec:tqc}.

 \subsection{$\mathbb{Z}_2 \times \mathbb{Z}_2$ DTQC}
 \label{subsec:Z2Z2DTQC}
 The model in consideration is a direct extension of the one considered in Sec.~\ref{subsec:Z2DTQC}, but now involves two groups of spins.
 Concretely, suppose we have a system comprised of two subsystems $A,B$. Each subsystem consists of spin-1/2 degrees of freedom placed along  either a 1d  chain or  a 2d square lattice, and we assume that the two subsystems $A,B$ are stacked on top of each other (so that site $i$ of subsystem $A$ is directly above site $i$ of subsystem $B$).

  Consider   the following $m=2$  quasiperiodic Hamiltonian
  \begin{align}
  H(t) = H_0(\vec\omega t) + V,
  \end{align}
  where now
\begin{align}
V &= \sum_{ \stackrel{ i< j} { \alpha = A,B}} J_{ij} \sigma^z_{i,\alpha} \sigma^z_{j,\alpha} + \sum_{i\leq j} J'_{ij} \sigma^z_{i,A} \sigma^z_{j,B} \nonumber \\
& + \sum_{i, \alpha = A,B} h  \sigma^z_{i,\alpha}.
\label{eqn:exampleH_DTQC}
\end{align} 
Similarly to the Hamiltonian in \eqnref{eqn:exampleH_Z2}, \eqnref{eqn:exampleH_DTQC} describes pairwise Ising interactions of spins between and within subsystems, and a longitudinal field in the $z$-direction with strength $h \ll  \omega_i/2$.
The interactions have strengths $J_{ij}, J_{ij}'$, 
and we consider the  $2d$ case where $J_{ij} = -J \delta_{\langle i,j \rangle}$ where $\langle i,j\rangle$ represent  nearest-neighbor pairs of sites on one subsystem, and $J_{ij}' = -J \delta_{i,j}$.  We take $0 < J < \omega_i/2$. 

The driving Hamiltonian is taken to be 
\begin{align}
H_0(\vec\theta) = \sum_{i, \alpha = A,B} \frac{1}{2} ( \sigma^x_{i,\alpha} + 1 ) f_{\alpha}(\vec\theta),
\end{align}
and we choose 
\begin{align}
f_A(\vec\theta) &= \pi \omega_1 \Delta_N(\theta_1), \nonumber \\
f_B(\vec\theta) & = \pi \omega_2 \Delta_N(\theta_2),
\end{align}
where $\Delta_N$ is the approximation to the Dirac delta comb introduced in \eqnref{eqn:Fejer}. These functions satisfy $\overline{f}_i = Q_{ij} \omega_j$, with  $
Q = \frac{1}{2} \begin{pmatrix}
1 & 0 \\
0 & 1
\end{pmatrix}$.
One can   understand this driving protocol as two  Floquet drives, performed at frequencies $\omega_1, \omega_2$, acting on different halves of the system.
Indeed, if the  two subsystems were disjoint, this protocol would simply result in two independent prethermal discrete time crystals. 
In the present case,  the system has interactions between the different subsystems -- nevertheless, we will show that a stable DTQC phase will be realized.

To see this, one can easily work out that the interaction Hamiltonian $H_\text{int}(\vec\theta)$ in the frame-twisted high-frequency limit has periodicity on the lattice $\mathcal{L}'$ generated by the translation vectors $\vec\tau_1 = 2\pi(2,0)$ and $\vec\tau_2 = 2\pi(0,2)$. In other words, $\mathcal{L}' = 2 \mathcal{L} = 4 \pi \mathbb{Z}^m$. Thus, the group of symmetries of the effective Hamiltonian $D$ is $\mathcal{G} = \mathbb{Z}_2 \times \mathbb{Z}_2$.
Explicitly, the unitary symmetries realized are
 $g_{\vec{\tilde \tau}} = \prod_{i} \sigma^x_{i,A}$ for  $\vec{\tilde \tau} = 2\pi(1,0)$, and  $g_{\vec{\tilde \tau}} = \prod_{i} \sigma^x_{i,B}$  for $\vec{\tilde \tau} = 2\pi(0,1)$.

Repeating the same analysis as in Sec.~\ref{subsec:Z2DTQC}, the leading order effective Hamiltonian is
\begin{align}
D^{(0)} = \sum_{ \stackrel{ i< j}{ \alpha = A,B}} \!\!\!\! J_{ij} \left( a_{\alpha}(N)  \sigma^z_{i,\alpha} \sigma^z_{j,\alpha}  + b_\alpha(N) \sigma^y_{i,\alpha} \sigma^y_{j,\alpha} \right),
\end{align}
where $a_{\alpha}(N), b_{\alpha}(N)$ are numerical factors which result from the averaging of the interaction Hamiltonian over the larger unit cell.
For $N=20$, $a_{\alpha}(20) = 0.928$ and $b_{\alpha}(20) = 0.072$ for $\alpha = A,B$.
We observe the following salient features: First, $D^{(0)}$   is manifestly $\mathbb{Z}_2 \times \mathbb{Z}_2$ symmetric. 
In particular, the   $\sigma^z_{i,A} \sigma^z_{i,B}$ interactions between the two subsystems, which are odd under both $\mathbb{Z}_2$ groups, have been eliminated. In fact, any interactions that remain between the subsystems, even at higher orders, \emph{must} be $\mathbb{Z}_2 \times \mathbb{Z}_2$ symmetric.
Second, the Ising interactions in the $z$-direction dominate, in both subsystems. 

Thus, an initial state that has sufficiently low energy density with respect to $D^{(0)}$ will  equilibrate to a thermal state $\rho_\beta$ which spontaneously breaks both  $\mathbb{Z}_2 \times \mathbb{Z}_2$ Ising symmetries. 
$\rho(\vec\theta)$ will then have periodicity on the lattice $\mathcal{L}_{SSB} = \mathcal{L}'$, whose reciprocal lattice $\mathcal{L}^*_{SSB}$ is generated by the vectors $\vec\alpha_1 = (1/2,0)$ and $\vec\alpha_2 = (0, 1/2)$.
The power spectrum of the expectation value in time of local operators measured in this state, will then exhibit peaks at frequencies
\begin{align}
\Omega = \frac{1}{2} \vec{n}\cdot \vec\omega, \qquad \vec{n} \in \mathbb{Z}^2,
\end{align}
which is the manifestation of the spontaneous breaking of the TTSes of the driving Hamiltonian.

 In the high frequency limit, we can treat $\mathcal{V}(t) = \mathbb{I}$ and derive analytic expressions for   the measurement of an observable in time $s(\vec\theta) = \Tr[\hat{s} \rho(\vec\theta)]$. For example, for the observable $\hat{s} = \frac{1}{\text{Vol.}} \sum_{i} \sigma^z_{i, A} + \frac{1}{\text{Vol.}} \sum_i \sigma^z_{i,B}$, this works out to be
 \begin{align}
 s(\vec\theta) &   \approx   \frac{1}{\text{Vol.}} \sum_{i}  \Tr[ \sigma^z_{i,A} \rho_\beta ] \cos\left[ \pi \Theta_N(\theta_1) \right] \nonumber \\
& +   \frac{1}{\text{Vol.}} \sum_{i}\Tr[ \sigma^y_{i,A} \rho_\beta ] \sin\left[  \pi \Theta_N(\theta_1) \right]  \nonumber \\
& +   \frac{1}{\text{Vol.}} \sum_{i}  \Tr[ \sigma^z_{i,B} \rho_\beta ] \cos\left[ \pi \Theta_N(\theta_2) \right]  \nonumber \\
& +   \frac{1}{\text{Vol.}} \sum_{i}\Tr[ \sigma^y_{i,B} \rho_\beta ] \sin\left[ \pi \Theta_N(\theta_2) \right],
 \end{align}
 where the function $\Theta_N(\theta)$ is given by \eqnref{eqn:defn_of_Theta}.
 We plot this observable in Fig.~\ref{fig:prethermalCartoonZ2Z2} using $\Tr[\sigma^z_{i,\alpha} \rho_\beta] = 0.9$ and $\Tr[\sigma^y_{i,\alpha} \rho_\beta] = 0$.
 
\subsection{$\mathbb{Z}_3 \times \mathbb{Z}_2$ DTQC} 
\label{subsec:Z3Z2DTQC} 

Next we present the Hamiltonian for the $\mathbb{Z}_3 \times \mathbb{Z}_2$ DTQC.
Consider, like in the previous example, a system comprised of two subsystems  $A,B$ of square lattices in $d=2$,   stacked so that sites lie on top of each other. 
Assume now however the local degrees of freedom of subsystem $A$ (which reside on the sites) to be comprised of three levels (i.e.~a `spin-1') and that of subsystem $B$ to be comprised of two levels (i.e.~a `spin-1/2').
Take the $m=2$ quasiperiodic Hamiltonian to be 
\begin{align}
H = H_0(\vec\omega t ) + V,
\end{align}
with
\begin{align*}
V & = -\sum_{\langle i j\rangle} J^A \left( \mu^\dagger_i \mu_j + \text{h.c.}  \right) -  \sum_{\langle i j\rangle} J^B \sigma^z_i \sigma^z_j  \nonumber \\
& + \sum_i J^{AB} \left( \mu_i \sigma^z_i + \text{h.c.} \right) + \sum_i h \sigma^z_i + h' (\mu_i + \mu_i^\dagger),
\end{align*}
where the operators $\mu, \eta$ which act locally on $A$ are given explicitly by 
\begin{align}
\mu = \begin{pmatrix}
1 & 0 & 0 \\ 
0 & e^{i 2\pi/3} & 0 \\ 
0 & 0 & e^{i 4\pi/3}
\end{pmatrix},
\qquad
\eta = \begin{pmatrix}
0 & 1 & 0 \\
0 & 0 & 1 \\
1 & 0 & 0 
\end{pmatrix}.
\end{align}
These are the so-called `clock operators' that enter into quantum   clock models: $\eta$   increments the state of a `clock' (the levels which the matrix $\mu$ is diagonal in) one step in a cyclic fashion.
Thus, $V$ describes ferromagnetic Ising-like interactions between degrees of freedom of $A$, and also ferromagnetic Ising interactions between degrees of freedom of $B$.
We take the driving Hamiltonian to be 
\begin{align}
H_0 & = \sum_i   \frac{1}{3} \left(  [1- e^{i \frac{2\pi}{3}}] \eta + \text{h.c.}    \right) f_A(\vec{\theta}) \nonumber \\
& + \sum_i   \frac{1}{2} \left(\sigma_i^x + 1 \right)f_B(\vec{\theta}),
\end{align}
with
\begin{align}
f_A(\vec\theta) &=  \frac{2\pi \omega_1}{3} 
\Delta_N(\theta_1), \nonumber \\
f_B(\vec\theta) & = \pi \omega_1 \Delta_N(\theta_1) + \pi \omega_2 \Delta_N(\theta_2),
\end{align}
which satisfy $\overline{f}_i = Q_{ij} \omega_j$, with
 $Q =  \begin{pmatrix}
\frac{1}{3} &  0 \\
\frac{1}{2} &  \frac{1}{2},
\end{pmatrix}
$.
We assume the driving frequencies are sufficiently large compared to all local couplings.

We now take the frame-twisted high-frequency limit.
The interaction Hamiltonian $H_{\text{int}}(\vec\theta)$  has periodicity on the lattice $\mathcal{L}'$ generated by the translation vectors $\vec\tau_1 = 2\pi(3,1)$ and $\vec\tau_2 = 2\pi(0,2)$. Therefore, the symmetries that the effective Hamiltonian $D$ has,  belong to the group $\mathcal{G} = \mathbb{Z}_6 \cong \mathbb{Z}_3 \times \mathbb{Z}_2$. Viewing $\mathcal{G}$ as $\mathbb{Z}_3 \times \mathbb{Z}_2$, it can be checked that group is explicitly generated by the unitary symmetries $g_{\vec{\tilde\tau}} = \prod_i \eta_i$ for $\vec{\tilde \tau} = 2\pi(1,1)$ satisfying $(\prod_i \eta_i)^3 = \mathbb{I}$ and $g_{\vec{\tilde\tau}} = \prod_i \sigma^x_i$ for $\vec{\tilde\tau} = 2\pi(0,1)$ satisfying $(\prod_i \sigma^x_i)^2 = \mathbb{I}$.

Repeating the same analysis as before (with $N=20$), the leading order effective Hamiltonian $D^{(0)}$ can be computed to be
\begin{align}
D^{(0)} = & -\sum_{\langle ij \rangle} 0.9280 J^A (\mu^\dagger_i \mu_j + \text{h.c.}) \nonumber \\
& - \sum_{\langle ij \rangle} J^B ( 0.866 \sigma^z_i \sigma^z_j + 0.134 \sigma^y_i \sigma^y_j ) + \delta
\end{align}
where $\delta$ is a   small term (compared to the dominant interactions above), given explicitly by
\begin{align}
& \delta =J^A \big[-0.018 (\mu_i  (\mu^\dagger_j \eta_j) + (\mu^\dagger_i \eta_i) \mu_j + \text{h.c.} ) \nonumber \\
& + 0.036 ( (\mu_i \eta_i) (\eta^\dagger_j \mu^\dagger_j) + (\eta^\dagger_i \mu_i) (\mu^\dagger_j \eta_j) + \text{h.c.}  ) \nonumber \\
&  - (0.018 -0.03i) ( (\mu_i \eta_i) (\mu^\dagger_j \eta_j) + (\mu^\dagger_i \eta_i) (\mu_j \eta_j)) + \text{h.c.} \nonumber \\
 & + (0.009 - 0.016i) ( \mu^\dagger_i (\mu_j \eta_j) + (\mu_i \eta_i) \mu^\dagger_j ) + \text{h.c.}   \big].
\end{align}
It can be explicitly checked that $D^{(0)}$ is $\mathbb{Z}_3 \times \mathbb{Z}_2$ symmetric, and that the dominant terms are the interactions $\mu_i^\dagger \mu_j$ on subsystem $A$ and Ising terms $\sigma^z_i \sigma^z_j$ on subsystem $B$. Therefore, we expect that the steady state $\rho_\beta$ of this system, in the preheating regime, should exhibit spontaneous breaking of both the $\mathbb{Z}_3 \times \mathbb{Z}_2$ symmetries. $\rho(\vec\theta)$ then has periodicity on the lattice $\mathcal{L}_{SSB} = \mathcal{L}'$, whose reciprocal lattice $\mathcal{L}^*_{SSB}$ is generated by the vectors $\vec\alpha_1 = (1/3,0)$ and $\vec\alpha_2 = (-1/6,1/2)$. The power spectrum of the measurement in time of local operators at long and late times, will then exhibit peaks at frequencies
\begin{align}
\Omega = n_1 \frac{\omega_1}{3} + n_2 \left( - \frac{\omega_1}{6} + \frac{\omega_2}{2} \right)
\end{align}
with $n_1, n_2 \in \mathbb{Z}$.

 In the high frequency limit, we can treat $\mathcal{V}(t) = \mathbb{I}$ and derive analytic expressions for   the measurement of an observable in time $s(\vec\theta) = \Tr[\hat{s} \rho(\vec\theta)]$. For example, for the local observable $\hat{s} = \mu_i + \mu^\dagger_i + \sigma^z_i$, 
 \begin{align}
 s(\vec\theta) &   \approx     \Tr[ \sigma^z_{i} \rho_\beta ] \cos\left[ \pi \Theta_N(\theta_1) + \pi \Theta_N(\theta_2) \right]  \nonumber \\
 & +  \Tr[ \sigma^y_{i} \rho_\beta ] \sin\left[ \pi \Theta_N(\theta_1) + \pi \Theta_N(\theta_2) \right] \nonumber \\
& +   \Big( \Tr[\mu_i^\dagger \rho_\beta] \left[ \frac{2}{3} e^{i (2\pi/3) \Theta_N(\theta_1) } + \frac{1}{3} e^{-i (4\pi/3) \Theta_N(\theta_1)}\right] \nonumber \\
& +  \text{h.c.} \Big).
 \end{align}
 We plot this observable in Fig.~\ref{fig:prethermalCartoonZ3Z2} using $\Tr[\sigma^z_{i} \rho_\beta] = 0.9$ and $\Tr[\mu_i \rho_\beta] = \Tr[\mu_i^\dagger \rho_\beta] = 0.8$.

\section{Topological phases}
\label{appendix:topological}
As we mentioned in Sec.~\ref{sec:TopologicalPhase}, the idea behind classifying quasiperiodic topological phases is to look at the non-trivial micromotion of the eigenstates. First, we make the following observation \cite{Turaev_9910, Turaev_0005, Kitaev_0506, Lurie_0905, KitaevIPAM, Thorngren_1612, Xiong_1701, Else_1810}: all the existing classifications of topological phases in $d$ spatial dimensions  with unitary symmetry $G$ (for example, Refs.~\cite{Chen_1106, Kapustin_1403,Kapustin_1406,Barkeshli_1410,Freed_1604})
can be expressed in terms of classifying the homotopy classes of maps $BG \to \Theta_d$ for some space $\Theta_d$. Here $BG$ is the classifying space for the group $G$, which is defined as $BG = EG/G$, where $EG$ is any contractible space on which $G$ acts freely ($BG$ is unique up to homotopy equivalence). In the case that $\Theta_\bullet$ forms a so-called ``$\Omega$-spectrum'', this amounts to the assumption that the classification is given by some generalized cohomology theory evaluated on $BG$ \cite{KitaevIPAM, Xiong_1701},
which is the case for all known classifications of invertible topological phases. However, our observation is more general than that, and applies also to classifications of non-invertible phases, for example the $G$-crossed modular tensor category classification of symmetry-enriched topological phases in two spatial dimensions \cite{Turaev_0005, Kirillov_0401,Etingof_0909,Barkeshli_1410}.

Now in all these classifications, $\Theta_d$ is some abstract space that can be introduced in a formal mathematical way. However, what we want to posit is that this space actually has physical meaning. Specifically, we want to say that $\Theta_d$ is actually homotopy-equivalent to $\Omega_d$, the space of all gapped ground states, introduced in Sec.~\ref{sec:TopologicalPhase}. Indeed, Kitaev (Appendix F of \cite{Kitaev_0506}) has shown at the microscopic level that any $G$-symmetric gapped ground state in $d$ spatial dimensions gives rise to a map $BG \to \Omega_d$, which would explain the origins of the maps $BG \to \Theta_d$ described above. 

If these conjectures hold, then it is clear how to classify maps $\mathcal{X} \to \Omega_d$ in the presence of symmetry $G$. First of all, we use the approach of Kitaev to turn this into a map $\mathcal{X} \times BG \to \Omega_d$, which is equivalent to a map $\mathcal{X} \times BG \to \Theta_d$. In other words, we simply replace $BG \to \mathcal{X} \times BG$ in the formulation of the classification.

We are now in a position to show how homotopy classes of maps $\mathbb{T}^m \to \Omega_d$ are related to SPT/SET phases with symmetry $\mathbb{Z}^{\times m}$. Indeed, this follows from the mathematical fact that
\begin{equation}
B(\mathbb{Z}^m) \cong \mathbb{T}^m,
\end{equation}
as can be seen by setting $E(\mathbb{Z}^m) = \mathbb{R}^m$, with $\mathbb{Z}^m$ acting freely by discrete translations. More generally, in the case of systems with symmetry (other than time-translation symmetry) $G$, then we want to classify maps $\mathbb{T}^m \times BG \to \Omega_d$, and then these are in one-to-one correspondence with topological phases with symmetry $G \times \mathbb{Z}^m$, as can be seen from the equation
\begin{equation}
B(\mathbb{Z}^m \times G) \cong \mathbb{T}^m \times BG.
\end{equation}

Now let us consider cases where the states under consideration are invertible states; that is, they do not admit fractionalized excitations such as anyons. In that case, as we have already mentioned, it is believed that the spaces $\Theta_d$ for different dimensions are related in a particular way. Specifically, the idea is that, if we define for any space $\mathcal{X}$, $h^q(\mathcal{X})$ to be the homotopy classes of maps $\mathcal{X} \to \Theta_q$, then $h^{\bullet}(-)$ should define a generalized cohomology theory, which obeys some set of axioms \cite{HatcherBook}. For any generalized cohomology theory, one can show, using the suspension and wedge axioms, and the fact that \cite{HatcherBook} $\Sigma \mathbb{T}^d = \Sigma (S^1 \vee \mathbb{T}^{m-1} \vee \Sigma \mathbb{T}^{m-1})$ (here $\Sigma$ denotes suspension and $\vee$ denotes wedge sum),  
\begin{align}
h^q(\mathbb{T}^m) &= h^{q}(S^1) \times h^q(\mathbb{T}^{m-1}) \times h^{q-1}(\mathbb{T}^{m-1}).
\end{align}
By applying this formula recursively, we derive \eqnref{eq:torus_decomposition}.

As we have described, the emergent $\mathbb{Z}_2 \times \mathbb{Z}_2$ symmetry arises as a symmetry of the effective static Hamiltonian in the twisted high-frequency limit. Nevertheless, we ultimately want to interpret the topological phases that we construct as protected by multiple time-translation symmetries, i.e.\ as being protected by $\mathbb{Z} \times \mathbb{Z}$. Naturally, any SPT protected via $\mathbb{Z}_2 \times \mathbb{Z}_2$ can also be interpreted as an SPT protected by $\mathbb{Z} \times \mathbb{Z}$, since there is a homomorphism from $\mathbb{Z} \times \mathbb{Z} \to \mathbb{Z}_2 \times \mathbb{Z}_2$. On the other hand, if we go \emph{beyond} the twisted high-frequency limit (for example, if we assume that quasiperiodically driven systems can be MBL without needing to appeal to a twisted-high frequency limit), we can still apply the general analysis of Section \ref{subsec:eigenstate_topological}, but it is not clear that the emergent $\mathbb{Z}_2 \times \mathbb{Z}_2$ is expected to survive even if we remain within the same quasiperiodic topological phase. Therefore, the signatures of the phase that we discuss might need to be modified.

\section{Signatures of quasiperiodically driven topological phases beyond the twisted high-frequency limit}
\label{appendix:relaxing}
As we mentioned in Section \ref{subsec:minimal}, it is not automatically clear that any signatures of a quasiperiodically driven topological phases that are based on the presence of the emergent finite internal symmetry group $\mathcal{G}$ will necessarily be robust in some stabilization of the phase that is beyond the twisted high-frequency limit. Let us discuss this in more detail in the phase that we were disucssing in Section \ref{subsec:minimal}, which had emergent symmetry $\mathcal{G} = \mathbb{Z}_2 \times \mathbb{Z}_2$ and microscopic symmetry $G = \mathbb{Z}_2$.

The homomorphism $\mathbb{Z} \times \mathbb{Z} \times G \to \mathbb{Z}_2 \times \mathbb{Z}_2 \times \mathbb{Z}_2$ induces a corresponding pull-back map $\mathcal{H}^4(\mathbb{Z}_2 \times \mathbb{Z}_2 \times \mathbb{Z}_2, \mathbb{Z}) \to \mathcal{H}^4(\mathbb{Z} \times \mathbb{Z} \times \mathbb{Z}_2, \mathbb{Z}) \cong \mathcal{H}^2(\mathbb{Z}_2, \mathbb{Z}) \cong \mathcal{H}^1(\mathbb{Z}_2, \mathrm{U}(1)) \cong \mathbb{Z}_2$. The phase discussed in Section \ref{subsec:minimal} remains nontrivial under this pull-back map. However, the signature of the projective representation of the emergent $\mathbb{Z}_2 \times \mathbb{Z}_2$ on a symmetry flux of the microscopic $\mathbb{Z}_2$ is \emph{not} robust. This is because the projective representation of the emergent symmetry corresponds to the non-trivial class in $\mathcal{H}^3(\mathbb{Z}_2 \times \mathbb{Z}_2, \mathbb{Z}) \cong \mathcal{H}^2(\mathbb{Z}_2 \times \mathbb{Z}_2, \mathrm{U}(1))$, and it maps to the \emph{trivial} class under the pull-back map $\mathcal{H}^3(
\mathbb{Z}_2 \times \mathbb{Z}_2, \mathbb{Z}) \to \mathcal{H}^3(\mathbb{Z} \times \mathbb{Z}, \mathbb{Z})$.

We can understand this better if we imagine computing the $\mathcal{H}^4$ classifications using the K\"unneth formula. Viewed as a $\mathbb{Z}_2^{\times 3}$ SPT, our SPT lives in the factor
\begin{equation}
\mathbb{Z}_2 \cong \mathcal{H}^1(\mathbb{Z}_2, \mathcal{H}^3(\mathbb{Z}_2 \times \mathbb{Z}_2, \mathbb{Z})) \leq \mathcal{H}^4(\mathbb{Z}_2^{\times 3}, \mathbb{Z})\end{equation}
which precisely corresponds to the statement that a $\mathbb{Z}_2$ symmetry flux corresponds to a projective representation of $\mathbb{Z}_2 \times \mathbb{Z}_2$. On the other hand, viewed as a $\mathbb{Z} \times \mathbb{Z} \times \mathbb{Z}_2$ SPT, our SPT lives in the factor
\begin{equation}
\label{eq:newkunneth}
\mathbb{Z}_2 \cong \mathcal{H}^2(\mathbb{Z}_2, \mathcal{H}^2(\mathbb{Z}\times \mathbb{Z}, \mathbb{Z})) \leq \mathcal{H}^4(\mathbb{Z} \times \mathbb{Z} \times \mathbb{Z}_2, \mathbb{Z})\end{equation}
In other words, relaxing the $\mathbb{Z}_2 \times \mathbb{Z}_2$ symmetry to a $\mathbb{Z} \times \mathbb{Z}$ symmetry causes our phase to jump into a different factor of the K\"unneth formula, which in turns suggests that we should expect different physical signatures. Indeed,
according to the usual physical interpretation of the K\"unneth formula, \eqnref{eq:newkunneth} seems to be telling us that if we annihilate two symmetry fluxes of the microscopic $\mathbb{Z}_2$ symmetry, then we leave behind a nontrivial 0-dimensional SPT classified by $\mathcal{H}^2(\mathbb{Z} \times \mathbb{Z}, \mathbb{Z}) \cong \mathbb{Z}$. The latter corresponds to quantized energy pumping between the two incommensurate frequencies (see Section \ref{subsec:relation}), so we can interpret this as the statement that a symmetry flux of the microscopic $\mathbb{Z}_2$ symmetry carries a \emph{half-quantized} frequency pump.
We leave a more systematic understanding of this effect to future work.

\section{Rarity of resonances in quasiperiodically driven systems}
\label{sec:rarity_resonances}
The crucial mathematical property of quasiperiodic driving that we employ in our paper is that, although for any given frequency vector $\vec{\omega} = (\omega_1, \cdots, \omega_m)$, the spectrum of possible multi-photon emission and absorption processes, i.e.~$\vec{\omega} \cdot \vec{n}$ for integer vectors $\vec{n}$, is dense everywhere (and in particular at zero), if we restrict to low-order processes (i.e.~$|\vec{n}|$ small), the possible values of $|\vec{\omega} \cdot \vec{n}|$ are still somewhat sparse. To quantify this, let us introduce the following definition.

\begin{defn}
Let $\vec{\omega} \in \mathbb{R}^n$ be a nonzero vector. Then the exponent $\sigma(\vec{\omega})$ is the smallest number $\sigma$ such that for all $\epsilon > 0$, there exists a constant $c > 0$ (depending on $\vec{\omega}$ and $\epsilon$) such that
\begin{equation}
|\vec{\omega} \cdot \vec{n}| \geq \frac{c|\vec{\omega}|}{|\vec{n}|^{\sigma + \epsilon}}
\end{equation}
for all nonzero integer vectors $|\vec{n}|$.
\end{defn}
Note that this exponent is invariant under rescaling the frequencies, i.e.~$\sigma(a \vec{\omega}) = \sigma(\vec{\omega})$ for any real $a \neq 0$.

Clearly, if the frequencies are rationally related, i.e.\ there exists a nonzero integer vector $\vec{n}$ such that $\vec{n} \cdot \vec{\omega} = 0$, then $\sigma(\vec{\omega}) = \infty$. On the other hand, we have the lower bound
\begin{proposition}
For any $\vec{\omega} \in \mathbb{R}^m$, $\sigma(\vec{\omega}) \geq m-1$.
\begin{proof}
This is a higher-dimensional version of the famous Dirichlet approximation theorem~\cite{SchmidtBook}. Like that theorem, it can be proven using Minkowski's theorem~\cite{SchmidtBook}, which states that if $C$ is some bounded convex subset of $\mathbb{R}^m$, whose volume satisfies $\mu(C) > 2^m$, then $C \cap \mathbb{Z}^m$ contains at least one point other than $\vec{0}$.

 Indeed, using the rescaling invariance of $\sigma$ to set $|\vec{\omega}| = 1$, let us define $C_{R,c} = \{ \vec{x} \in \mathbb{R}^m : |\vec{x} \cdot \vec{\omega}| \leq c, |\vec{x}| \leq R \}$, then $\mu(C_{R,c}) \geq c v_m R^{m-1}$ for some constant $v_m$ which only depends on $m$. Hence, we can satisfy the condition to apply Minkowski's theorem to $C_{R,c}$ provided that $c v_m R^{m-1} > 2^m$. This means that for any $R > 0$, there exists $\vec{n} \in \mathbb{Z}^d$ nonzero with $|\vec{n}| \leq R$ such that $|\vec{\omega} \cdot \vec{n}| \leq  2^{m+1} v_m^{-1} R^{-(m-1)}$. This indeed implies that we must have $\sigma(\vec{\omega}) \geq m-1$.
\end{proof}
\end{proposition}
In fact, this lower bound is generically saturated, in the following precise sense.
\begin{proposition}
Let $W$ be the set of all $\vec{\omega} \in \mathbb{R}^m$ such that $\sigma(\vec{\omega}) > m-1$. Then its Lebesgue measure $\mu(W) = 0$. That is, almost all $\vec{\omega}$'s have $\sigma(\vec{\omega}) = m-1$.
\begin{proof}
We extend the proof given for the case $m=2$ in Ref.~\cite{BugeaudBook}.
By the rescaling invariance of $\sigma(\vec{\omega})$, it is sufficient to prove that $\mu(S) = 0$, where $S$ is the intersection of $W$ with $B$, the unit ball in $\mathbb{R}^m$.

Define $S_{\vec{n},c,\gamma}$ to be the set of all vectors $\vec{\omega} \in B$ such that
\begin{equation}
|\vec{\omega} \cdot \vec{n}| < c|\vec{n}|^{-\gamma}
\end{equation}
Observe that $\mu(S_{\vec{n},c,\gamma}) \leq cu_m |\vec{n}|^{-(\gamma + 1)}$, for some constant $u_m$ that only depends on $m$.
Now define $S_{\vec{c,\gamma}} = \cup_{\vec{n} \in \mathbb{Z}^m} S_{\vec{n}, c,\gamma}$. Then we see that
\begin{align}
\mu(S_{\vec{c,\gamma}}) \leq cu_m \sum_{\vec{n} \in \mathbb{Z}^m} |\vec{n}|^{-(\gamma+1)} := c \Sigma_\gamma < \infty,
\end{align}
provided that $\gamma > m-1$.
Now define $S_\gamma = \cap_{c > 0} S_{c,\gamma}$. We see that $\mu(S_\gamma) \leq c \Sigma_\gamma$ for all $c > 0$, from which we conclude that $\mu(S_\gamma) = 0$. Finally, since we can write $S = \cup_{\gamma > m-1} S_\gamma$, we conclude that $\mu(S) = 0$.
\end{proof}
\end{proposition}
Finally, let us give a concrete example of a frequency vector $\vec{\omega}$ with $\sigma(\vec{\omega}) = m-1$. Indeed, we have
\begin{proposition}
Let $\vec{\omega}$ be a vector of frequencies that are not rationally related and suppose that all ratios $\omega_i/\omega_j$ are algebraic numbers, i.e.~each one is a root of some polynomial equation with integer coefficients. Then $\sigma(\vec{\omega}) = m-1$.
\begin{proof}
This is known as the \emph{subspace theorem}~\cite{SchmidtBook,Schmidt__72}.
\end{proof}
\end{proposition}

\section{Proof of long-lived preheating regime and emergent symmetries from twisted time-translation symmetries}
\label{appendix:proof}

In this section, we provide the full, rigorous formulation of our statements on
a long-lived, preheating dynamical description of quasiperiodically-driven
systems. Our starting point is the iterative procedure described in Sec.~\ref{sec:proof}, which we will carry out up to some optimal iteration order ${q_*}$, to be computed. We will develop here the technical machinery needed to prove rigorous bounds on this iteration, and thereby prove our theorem.

\subsection{Some definitions}
We adopt the notion of a potential from Refs.~\cite{Abanin_1509}, which is a formal sum
\begin{equation}
\Phi = \sum_Z \Phi_Z,
\end{equation}
where the sum is over subsets $Z$ of the lattice. In order to analyze quasi-periodic driving, we want to consider potentials that are parameterized by the torus $\mathbb{T}^m$, that is, $\Phi(\vec{\theta})$ where $\vec{\theta} \in \mathbb{T}^m$.

It will be convenient to analyze the $\vec{\theta}$ dependence in Fourier space. Accordingly, we will define a ``colored potential'' to be a formal sum
\begin{equation}
\Phi(\vec{\theta}) = \sum_{Z,\vec{n}} \Phi_{Z,\vec{n}} e^{i\vec{\theta} \cdot \vec{n}}.
\end{equation}
where $\vec{n} \in \mathbb{Z}^m$ represents a Fourier mode.

We can define the formal commutator of colored potentials according to
\begin{equation}
[\Phi(\vec{\theta}), \Theta(\vec{\theta})] = \sum_{Z,\vec{n}, Z', \vec{n}'} [\Phi_{Z,\vec{n}}, \Theta_{Z',\vec{n}'}] e^{i\vec{\theta} \cdot (\vec{n} + \vec{n'})},
\end{equation}
where we take the commutator on the right-hand side to be supported on $Z \cup Z'$ whenever $Z$ and $Z'$ are non-disjoint (otherwise the commutator is zero). In other words,
\begin{equation}
\label{eq:colored_commutator_long}
[\Phi,\Theta]_{Z,\vec{n}} = \!\!\!\!\!\!\!\!\! \!\!\! \!\!\!   \sum_{\substack{Z_1, Z_2 : Z_1 \cup Z_2 = Z, \\\text{$Z_1$ and $Z_2$ non-disjoint}}} \sum_{\vec{n}_1, \vec{n}_2 : \vec{n}_1 + \vec{n}_2 = \vec{n}} [\Phi_{Z_1, \vec{n}_1}, \Theta_{Z_2, \vec{n}_2}].
\end{equation}
We can  simplify our notation by defining a \emph{colored set} to be a pair $\vec{Z} = (Z,\vec{n})$. We define a ``$\cup$'' operator on colored sets according to $(Z_1,\vec{n}_1) \cup (Z_2, \vec{n}_2) = (Z_1 \cup Z_2, \vec{n}_1 + \vec{n}_2)$, and we declare that two colored sets $(Z_1,\vec{n}_1)$ and $(Z_2, \vec{n}_2)$ are ``disjoint'' if and only if $Z_1$ and $Z_2$ are disjoint. Then,
 \eqnref{eq:colored_commutator_long} simply becomes
\begin{equation}
[\Phi,\Theta]_{\vec{Z}} = \sum_{\substack{\vec{Z}_1, \vec{Z}_2 : \vec{Z}_1 \cup \vec{Z}_2 = \vec{Z} \\ \text{$\vec{Z}_1$ and $\vec{Z}_2$ non-disjoint}}} [\Phi_{\vec{Z}_1}, \Theta_{\vec{Z}_2}],
\end{equation}
which is formally identical to the uncolored case. As in the uncolored case, we define the exponentiated action of one potential $\Theta$ on another $\Phi$ according to
\begin{equation}
e^\Theta \Phi e^{-\Theta} := \sum_{n=0}^{\infty} \frac{1}{n!} \mathrm{ad}_\Theta^n \Phi,
\end{equation}
where $\mathrm{ad}_\Theta \Phi := [\Theta,\Phi]$.

We also can similarly write the norm \eqnref{eq:colored_potential_norm} in this succinct notation as
\begin{equation}
\| \Phi \|_{\kappa} = \sup_x \sum_{\mathbf{Z} \ni x} e^{\kappa |\mathbf{Z}|} \| \Phi_{\mathbf{Z}} \|,
\label{eqn:appendix_norm}
\end{equation}
where the supremum is over sites $x$ on the lattice, and where we made the following (purely formal) definitions: $|(Z,\vec{n})| = |Z| + |\vec{n}|$, and $x \in (Z, \vec{n})$ if and only if $x \in Z$.

\subsection{Statement of the theorem}
Our starting point is the iterative procedure described in Sec.~\ref{subsec:proofsketch}. We can perform this iterative procedure formally at the level of colored potentials. The goal is to prove bounds on the iteration, which are encapsulated in the following theorem. We start the iteration from some colored potential $H^{(0)}(\vec\theta) = D^{(0)} + V^{(0)}(\vec\theta)$, where $D^{(0)}$ is constant (that is, $D^{(0)}_{Z,\vec{n}} = 0$ for $\vec{n} \neq \vec{0}$) and $V^{(0)}(\vec\theta)$ has zero constant component (that is, $V^{(0)}_{Z,\vec{0}} = 0$). We perform the iteration described in Sec.~\ref{subsec:proofsketch}, giving at the $q$-th iteration a colored potential $H^{(q)}(\vec\theta) = D^{(q)} + V^{(q)}(\vec\theta)$, where $D^{(q)}$ is constant and $V^{(q)}(\vec\theta)$ has zero average on $\mathbb{T}^m$. Then the theorem is as follows.
\begin{theorem}
\label{thm:mainthm}
Suppose that the driving frequency vector $\vec{\omega}$ obeys a Diophantine condition $|\vec{n} \cdot \vec{\omega}| \leq c |\vec{n}|^{-\gamma} |\vec{\omega}|$ for all integer vectors $\vec{n}$  where $\gamma > m - 1$,  and  $c$ is a constant depending on the ratios $\omega_i/\omega_j$ but not on the overall scale of $\omega = |\vec\omega|$.

We assume also that the norm of the driving frequency ${\omega}$ is large enough compared to some local energy scale $\omega_0$,  namely there exists a decay rate $\kappa_0 > 0$, such that 
\begin{align}
\omega \geq K \omega_0, \qquad {q_*} \geq 1,
\end{align}
where
\begin{align}
{q_*} := \left\lfloor K' \left(\frac{\omega}{\omega_0} \right)^{\mu} \right\rfloor
\end{align}
is the maximum iteration order. Here 
\begin{align}
\omega_0 := \max \{ \| D^{(0)} \|_{\kappa_0}, \| V^{(0)} \|_{\kappa_0} \}, \qquad \mu = \frac{1}{\gamma+1}, 
\end{align}
where the definition of the norm $\| \cdot \|_{\kappa_0} $ is given in \eqnref{eqn:appendix_norm} and $K, K'$ numerical constants that do not depend on $D^{(0)}$, $V^{(0)}$,   the geometry of the lattice, or the driving frequencies, and are given by
\begin{align}
\label{eq:muandcprime_defns}
K = \frac{2^{2+\gamma}}{\kappa_0^{2+\gamma}} a, \qquad K' =2^{-4 -2\mu + \frac{1}{1+\mu}} (\mu+1)^{-1} \kappa_0^{\mu+1} a^{-\mu},
\end{align}
with $a = c^{-1} \gamma^\gamma e^{-\gamma} \times \max\{9 \times 2^{6+\gamma}, 3 \times 2^{1+\gamma} \kappa_0 \}$.

At the ${q_*}$-th iteration define $D := D^{(q_*)}$ and $V:= V^{(q_*)}$. Then 
 we have
\begin{align}
\| V^{{(q_*)}} \|_{\kappa_{{q_*}}} &\leq \omega_0 \left(\frac{1}{2}
\right)^{{q_*}}, \\
\| D^{{(q_*)}} \|_{\kappa_{{q_*}}} &\leq 2\omega_0, \\
\| D^{{(q_*)}} - D^{(0)} \|_{\kappa_{{q_*}}} & \leq C \omega_0 \left( \frac{\omega_0}{\omega }\right),
\end{align}
for some numerical constant $C$.
Here the decay constant $\kappa_{q_*}$ is given by $\kappa_{q_*}=(\kappa_0/2)^{\mu+1} - (\mu+1) (a 2^{1/(\mu+1)})^{\mu} \lambda^{\mu} (q_*-1) \geq \kappa_0/4$.
Additionally, for any potential $Z$ we   have
\begin{align}
 \| e^{\text{ad}_{A^{(q_*-1)} }}\cdots e^{\text{ad}_{A^{(0)}   } }  Z - Z \|_{\kappa_{{q_*}}} \leq C \|Z\|_{\kappa_0} (\omega_0/\omega).
\end{align}
\end{theorem}

Note that in all bounds involving the decay constant $\kappa_{q_*}$, we can simply replace it by $\kappa_0/4$. Also, we can identify the local energy scale $J$ in the main text  to be $\omega_0$. Then Theorem \ref{thm:mainthm} immediately yields \eqnref{eqn:D} and \eqnref{eqn:V} in Sec.~\ref{subsec:thmstatement} of the main text.
To prove \eqnref{eq:heisenberg_difference}, one invokes Lieb-Robinson bounds, following an identical argument to Appendix D of \cite{Machado_1908} (but using the Lieb-Robinson bound for short-range interactions).
 \eqnref{eqn:conserved_energy2} follows immediately from the arguments in~\cite{Abanin_1509} and our results above.

\subsection{The proof}

Our key technical tool is the following:
\begin{lemma}
\label{lem:keylemma}
Let $\Theta$, $\Phi$ be colored potentials and assume that $3\|\Theta\|_\kappa \leq \kappa - \kappa'$, with $0 < \kappa' < \kappa$. Then
\begin{equation}
\| e^{\Theta} \Phi e^{-\Theta} - \Phi \|_{\kappa'} \leq \frac{18}{(\kappa - \kappa')\kappa'} \| \Theta \|_\kappa \| \Phi \|_\kappa
\end{equation}
and
\begin{equation}
\| e^\Theta \Phi e^{-\Theta} \|_{\kappa'} \leq \left(1 + \frac{18}{(\kappa - \kappa')\kappa'} \| \Theta \|_{\kappa} \right) \| \Phi \|_{\kappa}
\end{equation}
\begin{proof}
This is the colored potential version of Lemma 4.1 from~\cite{Abanin_1509}, and the proof follows line by line identically to the proof of that Lemma.
\end{proof}
\end{lemma}

Armed with this Lemma, we can now proceed to analyze the iteration in a similar manner to what Ref.~\cite{Abanin_1509} did for the periodic case.
The key new aspect compared to the periodic case is that the formula for $A^{(q)}$ in terms of $V^{(q)}$, \eqnref{eqn:A}, involves a denominator which could potentially become very small. However, the Diophantine condition assumption in the theorem allows us to bound this denominator. This gives the following lemma.

\begin{lemma}
For $0 < \kappa' < \kappa$, and $\vec\omega$ obeying the Diophantine condition $|\vec{n} \cdot \vec{\omega}| \leq c |\vec{n}|^{-\gamma} |\vec{\omega}|$ with $\gamma > m-1$ for all integer vectors $\vec{n} \in \mathbb{Z}^m$ and a constant $c$, we have
\begin{align}
\| A^{(q)} \|_{\kappa'} &\leq \frac{c'}{(\kappa - \kappa')^{\gamma}} \frac{\| V^{(q)} \|_\kappa}{\omega},
\end{align}
where $c' = c^{-1} \gamma^{\gamma} e^{-\gamma}$.
\end{lemma}
\begin{proof}
\begin{align}
\| A^{(q)} \|_{\kappa'} &= \sup_x \sum_{Z \ni x,\va} e^{\kappa'(|Z| + |\va|)} \| A^{(q)}_{Z,\va} \|  \nonumber \\
&= \sup_x \sum_{Z \ni x, \va} e^{\kappa(|Z| + |\va|)} \frac{1}{|\vw \cdot \va|} \| V^{(q)}_{Z,{\va}} \|  \nonumber \\  
&\leq \frac{\omega^{-1} }{c } \sup_x \sum_{Z \ni x, \va} e^{\kappa'(|Z| + |\va|)} |\va|^{\gamma} \| V^{(q)}_{Z,{\va}} \|  \nonumber \\
&\leq \frac{c'\omega^{-1}}{(\kappa - \kappa')^{\gamma}} \sup_x \sum_{Z \ni x, \va} e^{\kappa'(|Z| + |\va|)} e^{(\kappa - \kappa')|\va|} \| V^{(q)}_{Z,\va} \|   \nonumber \\
&\leq \frac{c'}{(\kappa - \kappa')^{\gamma}} \frac{\| V^{(q)} \|_\kappa}{\omega},  
\label{eq:finalAbound}
\end{align}
 and we used the inequality
\begin{equation}
y^{\gamma} e^{-\epsilon y} \leq (\gamma/\epsilon)^{\gamma} e^{-\gamma}
\end{equation}
for $y, \epsilon > 0$.
\end{proof}

Now we proceed as in Ref.~\cite{Abanin_1509}. Specifically, we introduce
\begin{equation}
\gamma_q(O) = e^{-A^{(q)}} O e^{A^{(q)}},
\end{equation}
and
\begin{equation}
\alpha_q(O) = \int_0^1 ds e^{-sA^{(q)}} O e^{sA^{(q)}},
\end{equation}
the latter of which satisfies $e^{-A^{(q)}} (\vw \cdot \partial_\vt) e^{A^{(q)}} = \alpha_q( \vw \cdot \partial_\vt A^{(q)})$.
Then we see that
\begin{equation}
H^{(q+1)} = \gamma_q(D^{(q)}) + (\gamma_q(V^{(q)}) - V^{(q)}) + (\alpha_q(V^{(q)}) - V^{(q)}),
\end{equation}
from which we obtain
\begin{align}
D^{(q+1)}&= D^{(q)} + \langle W^{(q)}\rangle, \\
V^{(q+1)} &= W^{(q)} - \langle W^{(q)} \rangle,
\end{align}
(recalling $\langle \cdot \rangle$ is the time-averaging operation 
$\langle O \rangle := \frac{ \int d^m \vec\theta}{(2\pi)^m} O(\vec\theta) = O_{\vec 0}$), where
 \begin{align}
W^{(q)} & = (\gamma_q(D^{(q)}) - D^{(q)}) + (\gamma_q(V^{(q)}) - V^{(q)}) \nonumber \\
& + (\alpha_q(V^{(q)}) - V^{(q)}).
\end{align}
Let us suppose there is a sequence of strictly decreasing decay constants $\kappa_0 > \kappa_1 > \kappa_2 > \cdots >0$ (we will define this later), we then have
\begin{align}
\| D^{(q+1)}\|_{\kappa_{q+1}} &\leq \| D^{(q)} \|_{\kappa_q} + w_q/2, \\
\| V^{(q+1)} \|_{\kappa_{q+1}} &\leq w_q, \\
\|D^{(q+1)}- D^{(q)} \|_{\kappa_{q+1}} & \leq w_q/2,
\end{align}
where $w_q = 2\| W^{(q)} \|_{\kappa_{q+1}}$.
Using  Lemma \ref{lem:keylemma} and Lemma 2, we see that, provided that
\begin{equation}
\label{eq:anothercondn}
\frac{3c'}{(\kappa_q - \kappa_{q}' )^{\gamma}} \frac{\| V^{(q)} \|_{\kappa_q}}{\omega} \leq \kappa_q' - \kappa_{q+1},
\end{equation}
then we obtain

\begin{equation}
w_q \leq  \frac{36 c'(\| D^{(q)} \|_{\kappa^{(q)}} + 2 \| V^{(q)} \|_{\kappa_q})}{\kappa_{q+1} (\kappa_q' - \kappa_{q+1})(\kappa_q - \kappa_q')^{\gamma}}  \frac{\| V^{(q)} \|_{\kappa_q}}{\omega}
\end{equation}
where we have chosen $0 < \kappa_{q+1} < \kappa_q' < \kappa_q$. In particular, we   choose to let $\kappa_q' = (\kappa_q + \kappa_{q+1})/2$  and so we obtain
\begin{equation}
w_q \leq \frac{(36 \times  2^{1+\gamma)} c' (\| D^{(q)} \|_{\kappa_q} + 2 \| V^{(q)} \|_{\kappa_q}) }{\kappa_{q+1} (\kappa_q - \kappa_{q+1})^{\gamma+1}} \frac{\| V^{(q)} \|_{\kappa_q}}{\omega},
\end{equation}
while the requirement \eqnref{eq:anothercondn} becomes
\begin{equation}
\label{eq:anothercondn_simpler}
(3 \times 2^{1+\gamma}) c' \frac{\|  V^{(q)} \|_{\kappa_q}}{\omega} \leq (\kappa_q - \kappa_{q+1})^{\gamma+1}.
\end{equation}

Now we turn to the question of how we should choose the decay constants $\kappa_q$. A good way to do this is as follows (generalizing the approach of Ref.~\cite{Machado_1908}, which was an improvement on the original analysis of Ref.~\cite{Abanin_1509}). 
For notational brevity in what follows, we define $\lambda = \omega_0/\omega$, the small parameter. 

 First, let us make the induction hypothesis that for some $q$, $ \| D^{(q)} \|_{\kappa_q} \leq 2\omega_0 $, $\| V^{(q)} \|_{\kappa_q} \leq \omega_0 (1/2)^{q}$. What is the choice of $\kappa_{q+1}$ that ensures we also have $\| V_{q+1} \|_{\kappa_{q+1}} \leq \omega_0 (1/2)^{q+1}$ along with \eqnref{eq:anothercondn_simpler}. Indeed, this is so provided that
\begin{equation}
\label{eq:kappa_inequality}
\frac{a \lambda}{\kappa_{q+1}(\kappa_q - \kappa_{q+1})^{\gamma+1}} \leq 1,
\end{equation}
with $a = c' \max\{C_1, C_2 \kappa_0\}$, 
where $C_1 = (144 \times 2^{2 + \gamma})$ and $C_2 = (3 \times 2^{1+ \gamma})$. 
We would also like for $w_0$ to be `small' in the factor $1/\omega$; this suggests that we ought to set the  difference $\kappa_0 - \kappa_1$ to be independent of $\omega$. We achieve this by choosing to let 
\begin{align}
\kappa_0/\kappa_1 = 2.
\end{align}
The condition \eqnref{eq:kappa_inequality} that needs to be satisfied for $q = 0$ then reads
\begin{align}
\lambda \leq \frac{\kappa_0^{2+\gamma}}{2^{2+\gamma}} \frac{1}{a}.
\end{align}


To choose the subsequent constants $\kappa_{q \geq 2}$, intuitively, we can imagine that for small $\lambda$, $\kappa_{q \geq 1}$ varies very slowly as a function of $n$, so that we can treat $q$ as a continuous variable  so that $\kappa_{q} \mapsto k(q)$ and \eqnref{eq:kappa_inequality} becomes\\
\begin{equation}
\label{eq:kappa_inequality_diff}
\frac{a \lambda}{k(q) (-dk(q)/dq)^{\gamma+1}} \leq 1.
\end{equation}
If we choose \eqnref{eq:kappa_inequality_diff} to be saturated, this gives the differential equation
\begin{equation}
\frac{dk(q)}{dq} = -a^{\mu} \left(\frac{\lambda}{k(q)}\right)^\mu, \quad \mu = \frac{1}{\gamma + 1}
\end{equation}
with solution
\begin{equation}
\label{kappan}
k(q)^{\mu+1} = k(1)^{\mu+1} - (\mu+1) a^{\mu} \lambda^\mu (q-1).
\end{equation}

Armed with this intuition, we \emph{define} the decay rates $\kappa_q$  we shall employ, as the values obtained from evaluating the following generalized function $\kappa(q)$   at integer values $q \geq 1$
\begin{align}
& \kappa_q := \kappa(q) \text{ for } q= 1,2,3, \cdots \text{ where }\nonumber \\
& \kappa(q)^{\mu+1}:= \kappa_1^{\mu+1} - (\mu+1) b^{\mu} \lambda^{\mu} (q-1) \text{ for } q \in \mathbb{R}.
\end{align}
We remark that the decay rate is only physically meaningful and mathematically useful when $\kappa_q > 0$, and we will find conditions later on that restrict us to this range. 
Note also the  different constant $b$ from $a$ of \eqnref{kappan}, which will allow us to correct for the fact that we actually want to satisfy the finite difference equation \eqnref{eq:kappa_inequality} rather than the differential equation \eqnref{eq:kappa_inequality_diff}. We then find that $\kappa(q)$ satisfies the differential equation with $a$ replaced with $b$, i.e.
\begin{equation}
\frac{d\kappa}{dq} = -b^{\mu} \left(\frac{\lambda}{\kappa}\right)^\mu, \quad \mu = \frac{1}{\gamma + 1}.
\end{equation}
Since $\kappa(q)$ is a concave down function, we have that
\begin{equation}
\label{eq:diffineq}
\kappa_q - \kappa_{q+1} \geq -\kappa'(q) = b^{\mu} (\lambda/\kappa_q)^{\mu},
\end{equation}
and moreover
\begin{equation}
\left(\kappa_{q+1}/\kappa_q\right)^{\mu+1} = 1 - \frac{(\mu+1) b^{\mu} \lambda^{\mu}}{\kappa_q^{\mu + 1}} \geq 1 - \frac{2^{\mu+1} (\mu+1) b^{\mu} \lambda^{\mu}}{\kappa_1^{\mu + 1}},
\end{equation}
assuming that we continue the iteration only while $\kappa_q \geq \kappa_1/2 = \kappa_0/4$.
Therefore, so long as we impose the condition
\begin{equation}
\label{eq:somecondn}
\frac{ 2^{\mu+1}  (\mu+1) b^\mu \lambda^{\mu}}{\kappa_1^{\mu+1}} \leq 1/2,
\end{equation}
then we have
\begin{equation}
\label{eq:ratioineq}
(\kappa_{q+1}/\kappa_q)^{\mu + 1} \geq 1/2.
\end{equation}
Combining \eqnref{eq:diffineq} and \eqnref{eq:ratioineq}, we obtain
\begin{equation}
\frac{b \lambda}{\kappa_{q+1}(\kappa_q - \kappa_{q+1})^{\gamma+1}} \leq 2^{1/(\mu+1)}
\end{equation}
and hence ensure that \eqnref{eq:kappa_inequality} is satisfied, if we choose $b = a 2^{1/(\mu+1)}$.

Thus, if we define the maximum iteration order  $q^*$ as 
\begin{align}
{q_*} := \lfloor 2^{-4 - 2\mu + \frac{1}{1 + \mu} } (\mu+1)^{-1}  \kappa_0^{\mu+1} a^{-\mu}  \lambda^{-\mu}  \rfloor
\end{align}
and we demand that $\lambda$ (the  parameter relating  the local energy scale to the driving frequency) is small enough so that
\begin{align}
{q_*} \geq 1,
\end{align}
(otherwise the theorem is trivial anyway), 
then \eqnref{eq:somecondn} is automatically satisfied.

Note also that the definition of the maximum iteration order ${q_*}$   indeed ensures that the iteration continues only while $\kappa_q > 0$ (in fact, $ \kappa_q^{\mu+1} \geq (3/4)\kappa_1^{\mu+1}$ for $q \leq {q_*}$ so in particular $\kappa_q \geq \kappa_1/2 = \kappa_0/4$ in agreement with our assumption).

To ensure that the induction hypothesis on $\| D^{(q)} \|_{\kappa_q}$ remains satisfied, let us note that \eqnref{eq:kappa_inequality} also ensures that 
\begin{align}
& \| D^{(q+1)}\|_{\kappa_{q+1}} \leq \| D^{(q)} \|_{\kappa_q} + \frac{1}{2} \omega_0 \left(\frac{1}{2}\right)^q 
\end{align}
and summing over $q$ ensures $\| D^{(q)} \|_{\kappa_q} \leq 2 \omega_0$.
\\

Let us now prove the bounds on $\|D^{(q_*)}- D^{(0)}\|_{\kappa_{{q_*}}}$ and on the potential.
We first use that
\begin{align}
 \| D^{(q+1)}- D^{(q)} \|_{\kappa_{{q_*}}} & \leq \| D^{(q+1)} - D^{(q)} \|_{\kappa_{q}} \leq \frac{1}{2} \left(\frac{1}{2} \right)^{q} w_0  \nonumber \\  
& \leq C  e^{-c q} \omega_0 \left(\frac{\omega_0}{\omega}\right).
\end{align}
Summing the LHS up to ${q_*}$ then gives the desired result, redefining the constant $C$.

Next, 
\begin{align}
\| e^{\text{ad}_{A^{(q)}}} \cdots e^{\text{ad}_{A^{(0)}}} Z \|_{\kappa_{q+1}} & \leq \| Z \|_{\kappa_0} \prod_{j=0}^q \left(1 + C \left(\frac{\omega_0}{\omega} \right) e^{-c j} \right) \nonumber \\
& \leq C \| Z \|_{\kappa_0},
\end{align}
with redefined $C$,
using  Lemma 1 and Lemma 2  and that 
\begin{align}
& \left( 1 + \frac{ 2^{1 + \gamma} \times 18 c' }{ \kappa_{q+1} (\kappa_{q-1} - \kappa_q)^{1+\gamma} } \frac{\| V^{(q)} \|_{\kappa_q}}{\omega} \right) \nonumber \\
&  \leq \left( 1 + C' \frac{\kappa_q}{\kappa_{q+1} } \frac{\|V^{(q)}\|_{\kappa_q} }{b \lambda  \omega } \right) \nonumber \\ 
& \leq \left( 1 + C ({\omega_0}/{\omega} ) e^{-cq} \right),
\end{align}
with $(\kappa_q/\kappa_{q+1}) <  2^{1/(\mu+1)}$  and $\|V^{(q)}\|_{\kappa_q} \leq w_0 (1/2)^{q-1} \leq d e^{-c q} \omega_0 ({\omega_0}/{\omega} ) $.

Then, repeating the proof of Theorem 2.1 in~\cite{Abanin_1509}, defining 
\begin{align}
E_{\kappa_{q+1}} := e^{\text{ad}_{A^{(q)}}} \cdots e^{\text{ad}_{A^{(0)}}} Z - Z,
\end{align}
we have
\begin{align}
\| E_{\kappa_{q+1}} - E_{\kappa_q} \| \leq C (\omega_0/\omega )e^{-c q} 
\end{align}
which after summing yields
\begin{align}
& \| e^{\text{ad}_{A^{({q_*}-1)} }} \cdots e^{\text{ad}_{A^{(0)}}} Z - Z \|_{\kappa_{{q_*}}} = \| E_{\kappa_{{q_*}}} \|_{\kappa_{q*}} \nonumber \\
& \leq C \|Z\|_{\kappa_0} (\omega_0/\omega).
\end{align}
This concludes the proof of the Theorem. \qed

\bibliography{ref-manual,ref-autobib}

\end{document}